\documentclass[11pt]{amsart}
\usepackage{amsthm}

\usepackage[all]{xy}
\usepackage{amssymb}
\usepackage{enumerate}
\usepackage{mathrsfs}
\usepackage{graphicx}
\usepackage{caption}
\usepackage{subcaption}
\usepackage{float}
\usepackage{epigraph}

\numberwithin{equation}{section}

\newtheorem{thm}{Theorem}[section]

\newtheorem{prop}[thm]{Proposition}
\newtheorem{rem}{Remark}[section]

\newcommand{\eq}[1]{(\ref{#1})}

\renewcommand{\Re}{\operatorname{\rm Re}}
\renewcommand{\Im}{\operatorname{\rm Im}}

\newcommand{\beqast}{\begin{eqnarray*}}
\newcommand{\eqast}{\end{eqnarray*}}
\newcommand{\beqa}{\begin{eqnarray}}
\newcommand{\eqa}{\end{eqnarray}}

\newcommand{\bbe}{\begin{equation}}
\newcommand{\ee}{\end{equation}}

\renewcommand{\Re}{\operatorname{\rm Re}}
\renewcommand{\Im}{\operatorname{\rm Im}}

\newcommand{\bC}{{\mathbb C}}
\newcommand{\bE}{{\mathbb E}}

\newcommand{\bQ}{{\mathbb Q}}

\newcommand{\bR}{{\mathbb R}}

\newcommand{\bZ}{{\mathbb Z}}

\newcommand{\cK}{{\mathcal K}}

\newcommand{\cT}{{\mathcal T}}

\newcommand{\cL}{{\mathcal L}}

\newcommand{\cM}{{\mathcal M}}
\newcommand{\cN}{{\mathcal N}}
\newcommand{\cC}{{\mathcal C}}

\newcommand{\cU}{{\mathcal U}}

\newcommand{\hG}{{\hat G}}

\newcommand{\hf}{{\hat f}}

\newcommand{\al}{\alpha}

\newcommand{\be}{\beta}

\newcommand{\De}{\Delta}
\newcommand{\de}{\delta}
\newcommand{\eps}{\epsilon}
\newcommand{\ka}{\kappa}
\newcommand{\la}{\lambda}
\newcommand{\lp}{\lambda_+}
\newcommand{\lm}{\lambda_-}
\newcommand{\La}{\Lambda}
\newcommand{\mum}{\mu_-}
\newcommand{\mup}{\mu_+}

\newcommand{\sg}{\sigma}

\newcommand{\om}{\omega}

\newcommand{\ze}{\zeta}

\newcommand{\ga}{\gamma}
\newcommand{\gap}{\gamma_+}
\newcommand{\gam}{\gamma_-}

\newcommand{\Ga}{\Gamma}

\newcommand{\dd}{\partial}

\newcommand{\hh}{\hat h}

\newcommand{\bfo}{{\bf 1}}

\begin{document}

\title[Reliable pricing in and calibration of the rough Heston model]
{Correct implied volatility shapes  and reliable pricing in  the rough Heston model}

\author[
Svetlana Boyarchenko and
Sergei Levendorski\u{i}]
{
Svetlana Boyarchenko and
Sergei Levendorski\u{i}}

\begin{abstract}
We use 
%a simple modification of the formula for the characteristic function
%in the rough Heston model, 
%amenable to more accurate calculations, 
modifications of  the Adams method   and 
very fast and accurate sinh-acceleration method of the Fourier inversion (iFT)
(S.Boyarchenko and Levendorski\u{i}, IJTAF 2019, v.22) to evaluate prices of vanilla options; for options of moderate and long maturities and strikes not very far from the spot, thousands of prices can be calculated in several msec. with relative errors of the order of 0.5\% and smaller running Matlab on a Mac with moderate characteristics.  We demonstrate that for the calibrated set of parameters  in 
Euch and Rosenbaum, Math. Finance 2019, v. 29, the correct implied volatility surface is significantly flatter
and fits the data very poorly,  hence,   the calibration results in op.cit.
 is an example of the {\em ghost calibration} (M.Boyarchenko and Levendorki\u{i},
Quantitative Finance 2015, v. 15): the errors of the model and numerical method almost cancel one another. 
We explain how calibration errors of this sort
are generated by each of popular versions of numerical realizations of iFT 
(Carr-Madan, Lipton-Lewis and COS methods) with
prefixed parameters of a numerical method, resulting in spurious volatility  smiles and skews.    We suggest a general {\em Conformal Bootstrap principle} which allows one to avoid ghost calibration errors. 
%The principle  is applicable if the Fourier transform technique is used to construct the implied volatility surface. 
 We outline schemes of application of  Conformal Bootstrap principle and the method of the paper to the design of accurate and fast calibration procedures.

\end{abstract}

\thanks{
\emph{S.B.:} Department of Economics, The
University of Texas at Austin, 2225 Speedway Stop C3100, Austin,
TX 78712--0301, {\tt sboyarch@utexas.edu} \\
\emph{S.L.:}
Calico Science Consulting. Austin, TX.
 Email address: {\tt
levendorskii@gmail.com}}
\maketitle


\tableofcontents

\section{Introduction}\label{s:intro}
Starting with the celebrated Heston model \cite{heston-model}, affine models have become one of the most popular class of stochastic
volatility models, term structure models, and models in FX. The popularity is due to the fact that
the characteristic function in an affine model can be explicitly calculated solving an associated system
of generalized Riccati equations \cite{DFS}, hence, the Fourier transform technique allows
one to express prices of options of the European type as oscillatory integrals. However, many affine diffusion models fail to
accurately reproduce the volatility dynamics (models with jumps can reproduce wide variety of volatility surfaces but are not popular among practitioners because simple hedging is impossible).
%In a number of empirical studies, it is documented that, in many cases, the 
%volatility is rough: the log-volatility follows a fractional Brownian motion with the Hurst parameter of the order of 0.1.  
 As a remedy, rough volatility models are suggested. In a number of publications, it is stated that the important advantage
of rough volatility models is their ability to accurately reproduce the volatility surface. In particular, contrary to the Heston model and many other affine models, the ATM skew
in rough volatility models explodes as the maturity goes to 0.  See \cite{BayerFritzGatheral2015,GatheralJaissonRosenbaum2018,JacquierMartiniMuguruza2018,EuchRosenbaum2019,FordeZhang2017,FordeZhang2017,FordeSmithViitasaary2021,FrizGassiatPigato2021,FrizGassiatPigato2022} and the bibliographies therein. However, in several recent empirical publications, e.g.,
\cite{Romer2022,JaberLi2024}, the authors find that rough volatility models do not reproduce the volatility surface
accurately and certain affine diffusion models perform the task better. In particular, one of the general conclusions
in the abstract of \cite{JaberLi2024} is ``The skew of rough volatility models increases too fast on the short end, and decays too slow on the longer end\ldots". The first aim of the paper is to analyze to which extent incorrect shapes of implied volatility surfaces are caused not by models per se but by inaccurate numerical methods used for pricing in the calibration procedure.

Among the host of rough volatility models,
we consider the model in \cite{EuchRosenbaum2019} % and its extension in \cite{AbiJaberLarssonPulido10},
where the log-characteristic function has an affine structure as in standard affine models but can be calculated
only numerically solving the fractional Riccati equation. Hence, accurate calculations are difficult.
The last remark pertains to calculations in a host of affine models as well if the log-characteristic function can be calculated only numerically solving the system of generalized
 Riccati equations  - see the analysis in \cite{pitfalls}.
 In the case of the fractional Riccati equations, the potential for serious errors is greater. 
 %difficulties are  more serious. Popular models for the Fouier
% inversion used in the computational finance are also prone to serious errors.  
 The errors are not detected
 by calibration algorithms. 
In particular, the problem of a reliable calibration 
 of deep neural networks (DNN) in application to option prices has two sources of unreliability: DNN itself and a pricing algorithm. The former problem is extensively studied in the literature: see, e.g.,
 \cite{Berneretall2022}. The aim of the paper is to study the latter problem, and suggest remedies.

 Serious discrepancies between the empirical volatility curves and
 curves in the Heston model calibrated using  the Carr-Madan (CM method) \cite{carr-madan-FFT} and COS method
 \cite{COS} method are well-documented
 in an extensive empirical study \cite{HestonCalibMarcoMeRisk}. In particular, it was demonstrated
 that if CM or COS methods are used, the implied volatility curves significantly differ from the correct ones for short and long maturities; for maturities from 3 month to 1 year the errors are sizable but not very large (of the order of 10\%). 
 Earlier, the errors of CM, COS and Lewis \cite{lewisFT} methods
in applications to pricing in the Heston model and CIR2 model were demonstrated with stylized examples
in \cite{paraHeston,pitfalls}. Thus,  calibration errors due to insufficiently accurate pricing methods can be rather large even in the case of the Heston model, for which the explicit formula for the characteristic exponent is available; if the characteristic exponent can be calculated only numerically, the errors can snowball \cite{pitfalls}, and popular inaccurate methods cannot produce accurate results
close and far from maturity and for deep OTM options.  In the case of the rough Heston model,
  the fractional Adams method  is typically used to evaluate
 the characteristic function. The method is prone to large errors, hence, one can expect large total errors.  
 In the present paper, we demonstrate that 
 %for moderate and short maturities,
 the correct ATM skew in the model with the parameters 
 %in 
 %Consider the example in  
%\cite{EuchRosenbaum2019} with  the parameters of the rough Heston model  
\bbe\label{parEuRos}
\al=0.62,\ \ga=	0.1, \	 \rho=-0.681,\  \theta=0.3156, \ \nu=0.331, \ v_0=0.0392,
\ee
calibrated to the real data in 
 \cite{EuchRosenbaum2019}, is several times lower than the one produced by the numerical method in \cite{EuchRosenbaum2019}  and decays rather fast at $T$ increases to $T=5$ (cf. \cite[Fig. 5.2]{EuchRosenbaum2019} and Fig.~\ref{Set1Skews}); the implied volatility curves are also incorrect (cf. \cite[Fig. 5.1]{EuchRosenbaum2019} and curves in Fig.~\ref{Set1Curves}). 
 In \cite{EuchRosenbaum2019}, the
 Lewis method for the Fourier inversion and Adams method are used, as in a number of later publications. See, e.g.,  \cite{Imperial2020}. We produce the correct implied volatility curves for one of examples in \cite{Imperial2020} (see Fig.~\ref{ImperialCurves}),
 which shows that the errors in \cite{Imperial2020} are also sizable,  in the wings especially.
 In both papers, the curves are produced for moderate maturities only. We will show that the errors of the Lewis 
 and Adams methods increase as the maturity decreases. The same observations, to  larger degrees, holds for  CM and COS methods.
 %oth CM and COS methods produce incorrect volatility surfaces.
 In particular, the errors inherent in CM method can produce spurious nice curves when the correct curves are
 almost straight slopes (see Fig. \ref{Set1ImpVolsurfacesXiT152}).
 Thus, it is possible that the drawbacks of the rough Heston models documented
 in  \cite{JaberLi2024} are artifacts of inaccurate pricing procedures used for the calibration; the veracity of the conclusions about the performance of affine diffusion models in  \cite{Romer2022,JaberLi2024} also strongly depends on the quality
 of the numerical methods.

If an insufficiently accurate pricing algorithm is used, a certain set of parameters of the model can be rejected only because for this set, the chosen numerical method cannot calculate prices even remotely accurately (the effect of {\em sundial calibration}
 \cite{paraHeston}). On the other hand, at the boundary of the region of the parameter space where the errors of the method are not too large, the ``true calibration error" (the error that would be calculated using an error-free pricer) and error of the numerical method can cancel one another,
 and an incorrect model declared a good fit: {\em the ghost calibration} \cite{one-sidedCDS} (at the boundary of a well-lighted region, one sees ghosts). Therefore, the calibration result in \cite{EuchRosenbaum2019} is, apparently, the result of the ghost calibration. In Sect. \ref{s:numer}, we produce  a series of numerical examples which demonstrate that   the popular  methods
 are prone to sizable errors with serious implications for model calibration;  the study of errors is applicable to any model that uses the Fourier transform technique, not to the rough Heston model only.
 
 To demonstrate these effects, we use a very accurate and fast method for the Fourier inversion
 (sinh-acceleration \cite{SINHregular}), and
%The third contribution of the paper is an efficient  procedure for evaluation of the characteristic function in the rough Heston model. We use a straightforward modification of the formulas for the characteristic function in
%\cite{EuchRosenbaum2019,AbiJaberLarssonPulido10} and 
simple novel modifications of  the fractional
Adams method for the solution of the fractional Riccati equation, amenable to more accurate calculations
in the presence of a large spectral parameter
(Sect. ~\ref{s:improved}). The standard (fractional) Adams method does not take this factor into account, thereby introducing large errors when options of short maturities are priced.
%\footnote{
%In \cite{EuchGatheralRosenbaum2017,GatheralKeller-Ressel2018,GatheralRadoicic2019}, a somewhat
%different version of the rough Heston model is used; the numerical solution uses the same methods as in
%\cite{EuchRosenbaum2019}, hence, the  analysis of the errors in the present paper is applicable.} 
 Large errors of
the Adams method are documented
in \cite{RoughNotTough}, where the fractional Riccati equation is solved using
 the asymptotic expansion of the solution of the fractional
Riccati equation near 0 and the Richardson-Romberg extrapolation \cite{Pages2007} farther from 0, and CM method
is applied.
For the example considered in \cite{RoughNotTough}, the method in the present paper is significantly faster and, for options of short maturities, produces more accurate results than the hybrid  and CM methods taken together. Table \ref{table:T=0.5:5} demonstrates that for options of moderate maturities and strikes not far from the spot, thousands of prices can be calculated in less than 2 msec. with relative errors of the order of 0.5\% and smaller running Matlab on a Mac with moderate characteristics. The reader observes that the accuracy
decreases as the maturity decreases and as the strike moves from the spot. 
In Sect.~\ref{s:numer}, the reader can find additional tables which show that: 1) the relative errors
of implied volatilities are smaller, for far OTM options, much smaller; 2) for strikes in narrower ranges
around the spot, the same accuracy can be achieved several times faster using smaller grids for the Fourier inversion and smaller number of time steps;  3) as the maturity decreases, the range around the spot where a fixed error tolerance can be achieved using grids of moderate sizes shrinks, and the lengths of the grids increase; 4) however, in the region where the OTM option prices are larger than 0.001\% of the spot price, the same error tolerance can be satisfied in several hundred of milliseconds.
 \begin{table}
\caption{\small Simultaneous calculation for  $(K,T)\in [600, 1400] \times[0.5,5]$ using  SINH-acceleration. Prices of OTM and ATM options in the rough Heston model with the parameters \eq{parEuRos}, and absolute and relative errors of prices. Strike $S_0=1000$.
%Parameters: $\al=0.62$, $\ga=	0.1$, 	$\rho=-0.681$, $\theta=	0.3156$, $\nu=0.331$, $v_0=0.0392$.
}
{\small 
\begin{tabular}{c|ccccccccc}
\hline\hline
$K$ & 600 & 700 & 800 & 900 & 1000 & 1100 & 1200 & 1300 &1400
\\\hline
T=5 & 61.0845 &	96.9915 &	141.184&	192.939 &	251.436 &	215.848 &
	185.394 &	159.363&	137.122\\
err. &-0.0067 &	-0.0084 &	-0.0097 &	-0.011 &	-0.011 &	-0.010 &	-0.011 &	-0.011 &	-0.011 
\\
  rel.err. & -1.1E-04 &	-8.6E-05 &	-6.9E-05 &	-5.6E-05 &	-4.5E-05 &	-4.8E-05 &	-5.9E-05 &	-7.1E-05 &	-8.4E-05
\\\hline
T=0.5 & 0.06386 &	0.90464 &	6.086 &	23.898&	63.439 &	27.207 &	9.738 &	2.9247 &	0.74450\\
err. & -8.4E-04 &	-0.0065 &	-0.024 &	-0.047 &	-0.056 &	-0.043 &	-0.022 &	-0.0084 &	-0.0022 \\
 rel.err. & -0.013 &	-0.0071 &	-0.0040  &	-0.0020 &	-8.8E-04 &	-0.0016 &	-0.0023 &	-0.0029 &	-0.0030
   \\\hline

\end{tabular}
}
\begin{flushleft}{\tiny
Sinh-acceleration with $\om_1=0.4293$, $	b=0.8687$, $\om=0.1$, 	
$\ze=0.2405$, $N=11$ for OTM and ATM puts and $\om_1=-1.4293$, $	b=0.8687$, $\om=-0.1$, 	
$\ze=0.2405$, $N=11$ for OTM calls. The characteristic function is evaluated at points $(\xi,t)$, where
$\xi$ are on the grids used in the sinh-acceleration, and $t=j\De, j=1,2,\ldots, 80$, $\De=5/80$,
using Modification III of Adams method with 2 iterations at each time step.  Total CPU time, average over 1000 runs: approx 7.7 msec at 3,200 points.
We show only prices at $T=0.5, 5$ and several strikes. 
}
\end{flushleft}
\label{table:T=0.5:5}
 \end{table}
 
We are able to obtain reliable error bounds and produce results with relative errors of the order of E-05
and better
using 
  {\em Conformal Bootstrap principle} formulated in Sect.~\ref{ss: conformal bootstrap} and used in numerical examples in Sect.~\ref{s:numer}. The principle is used to
reliably assess the total error of the numerical methods for the Fourier inversion and  evaluation
of the integrand.  The principle is especially useful in  complicated models such as the rough Heston model, where the analytical properties of the integrand necessary for the derivation of accurate error bounds in a Fourier inversion algorithm are unknown, and accurate bounds for errors of numerical calculations of the characteristic function are lacking.

%In Sect.~\ref{ss:hybrid_comparison}, we explain how the perfomance of 
% However, we cannot claim that the method of the numerical solution
%of the fractional Riccati equation in the present paper is more accurate than the hybrid method.
%The reason is that in \cite{RoughNotTough}, a rather inaccurate CM method  is used, hence, we cannot separate the %errors of the hybrid and CM methods.

The rest of the paper is organized as follows. The modified Adams method, several methods for the Fourier inversion
(the conformal bootstrap principle  including) 
and numerical examples are  in Sect. \ref{s:improved}, \ref{s:FT} and \ref{s:numer}, respectively. 
In Sect.~\ref{s:calibration}, we outline  several schemes of  reliable and fast calibration procedures based
on  the Conformal Bootstrap principle and sinh-acceleration. The same schemes can be applied to any model where the characteristic function is calculated using the inverse Fourier transform. Sect. \ref{s:concl}
concludes. The modification of the Adams method with non-uniform grids and
additional Figures and Tables and are  in Appendix.

 %The final contribution of the paper is the {\em conformal bootstrap} method formulated in Sect.~\ref{ss: conformal %bootstrap}. The method allows one to calculate the results with high degree of certainty without conducting
% detailed studies of the analytical properties of the characteristic function and the errors of the numerical methods used. %The method can be especially useful for design of accurate implied volatility surfaces
 %in complicated models such as the rough Heston model, polynomial Ornstein-Uhlenbeck models, lifted Heston model
% and multi-factor SV models with jumps, where the exact studies of the properties of the characteristic function and errors %of numerical methods are essentially impossible. Once a reliable surface is constructed, one can use the fast calibration %procedure in \cite{HorvathMuguruzaTomas2021}. The method can be also used to construct the first layer in neural  % %network algorithms \cite{BayerStemper2018,Stone2020}.

\section{Formulas for characteristic function and modifications of fractional Adams method}\label{s:improved}

\subsection{Formulas for the characteristic function}\label{ss: char_function}
\subsubsection{The rough Heston model constructed in \cite{EuchRosenbaum2019}} 
Let $\al\in (0,1)$, $v, \ga,\theta,\nu>0$ and $\rho\in (-1/\sqrt{2}, 1/\sqrt{2})$; in \cite{EuchRosenbaum2017},
it is proved that $\rho\in (-1,1)$ is admissible. The (conditional) characteristic function 
of the log-price  $\Phi_\al(t,T,v,\xi):=\bE[e^{i\xi X_T}\ |\ X_t=0, V_t=v]$ in the rough Heston model is of the form
\bbe\label{chFRough}
\Phi_\al(t,T; v,\xi)=\exp[g_1(\xi,\tau)+vg_2(\xi,\tau)],
\ee
where $\tau=T-t$, 
\bbe\label{eq:g1g2}g_1(\xi,\tau)=\theta\ga\int_0^\tau h(\xi,s)ds,\ g_2(\xi,\tau)=I^{1-\al}h(\xi,\tau),\ee
%\beqa\label{g1}
%g_1(\xi,\tau)&=&\theta\ga\int_0^\tau h(\xi,s)ds,\\\label{g2}
%g_2(\xi,\tau)&=&I^{1-\al}h(\xi,\tau),
%\eqa
and $h(\xi,\cdot)$ is the solution of the fractional Riccati equation
\bbe\label{RiccRough}
D^\al_t h(\xi, t)=-\frac{1}{2}(\xi^2+i\xi)+\ga(i\xi\rho\nu-1)h(\xi,t)+\frac{(\ga\nu)^2}{2}h(\xi,t)^2,
\ee
subject to $I^{1-\al}h(\xi,0)=0$. Recall that, for $\al\in (0,1)$, $I^a$ and $D^\al$ are the fractional integral and differential operators:
\beqa\label{defIal}
I^\al u(t)&=&\frac{1}{\Ga(\al)}\int_0^t (t-s)^{\al-1}u(s)ds,\\
D^\al u(t)&=&\frac{1}{\Ga(1-\al)}\frac{d}{dt}\int_0^t (t-s)^{-\al}u(s)ds.
\eqa
Introduce the notation
\bbe\label{defF}
F(\xi,h)=-\frac{1}{2}(\xi^2+i\xi)+\ga(i\xi\rho\nu-1)h+\frac{(\ga\nu)^2}{2}h^2.
\ee
Equation \eq{RiccRough} subject to $I^{1-\al}h(\xi,0)=0$ is equivalent to the following  Volterra equation
\bbe\label{Volterra}
h(\xi,t)=I^\al F(\xi,t)=\frac{1}{\Ga(\al)}\int_0^t(t-s)^{\al-1}F(\xi,h(\xi,s))ds.
\ee
In \cite{EuchRosenbaum2019}, \eq{Volterra} is solved (numerically) using the fractional Adams method. It is 
not explained how $g_1$ and $g_2$ are evaluated. Presumably, using the piece-wise linear interpolation as in the fractional Adams method: the trapezoid rule and fractional trapezoid rule, respectively. Since $h$ is not smooth at 0 and an additional fractional integral needs to be evaluated, the errors increase. 
We use the following version of \eq{chFRough}, thereby avoiding additional errors.
\begin{prop}\label{prop:newRoughHestonChExp}
Let $\al\in (0,1)$, $v, \ga,\theta,\nu>0$, $\rho\in (-1, 1)$, and let $h(\xi,t)$ be
the solution of \eq{Volterra}. Then
\bbe\label{chFRough2}
\Phi_\al(t,T,v,\xi)=\exp\left[\int_0^\tau (\ga\theta h(\xi,s)+vF(\xi,h(\xi,s)))ds\right].
\ee
\end{prop}
\begin{proof} It suffices to note that $I^{1-\al}I^\al=I^1$.
%The standard conceptual proof:
%$I^\al$ is the convolution operator with the kernel $t^{\al-1}/\Gamma(\al)$, equivalently,
% the pseudo-differential operator (pdo) with the symbol $(0+i\xi)^{-\al}$. The product of the symbols of two 
% pdo's is the symbol of the latter, hence, $I^{1-\al}I^\al$ is the operator with the symbol $(0+i\xi)^{-1}$, and
% \[
% I^{\al-1}h(\xi,t)=I^{\al-1}I^\al  F(\xi,h(\xi,t))=\int_0^t F(\xi,h(\xi,s))ds.
% \]
% The elementary version: apply the composition $I^{\al-1}\circ I^\al$ to the function $t^{\be-1}$, where $\be>0$.
 %, and approximate  any sufficiently regular function by linear combinations of such functions.
\end{proof}
Assuming that $h$ is evaluated sufficiently accurately, \eq{chFRough2} 
allows one to use only the trapezoid rule
and usual quadratures of higher order outside an appropriate vicinity of 0, and since we avoid the additional
fractional integration, the accuracy of the final result increases. 

\begin{rem}{\rm In \cite{RoughNotTough}, the modified version is rejected as being more complicated for
the asymptotic method used. For the modifications of the Adams method that we use, the advantages
of \eq{chFRough2} are significant, for options of short maturity especially.
}
\end{rem}

\subsubsection{A generalization of \eq{chFRough2}}
In \cite[Theorem 4.3 and Example 7.2]{AbiJaberLarssonPulido10} (see also \cite[Eq. (9)-(10), (12)]{RoughNotTough}), one finds a generalization of the rough Heston model \cite{EuchRosenbaum2019}. Assuming $\Im\xi\in [-1,0]$ and $\Im\eta<0$, the characteristic 
function of the joint distribution of $(X_T, V_T)$ 
\[
\Phi_\al(t,T,v,\xi,\eta):=\bE[e^{i\xi X_T+i\eta V_T}\ |\ X_t=0, v_T=v]
\]
admits the representation
\bbe\label{chFRoughV}
\Phi_\al(t,T,v,\xi,\eta)=\exp[g_1(\xi,\eta,\tau)+vg_2(\xi,\eta, \tau)],
\ee
where $\tau=T-t$, $g_1$ and $g_2$ are defined via a function $h(\xi,\eta,\tau)$ as above.
%similarly to
%\eq{g1}-\eq{g2}.
%\beqa\label{g1V}
%g_1(\xi,\eta, \tau)&=&\theta\ga\int_0^\tau h(\xi,\eta, s)ds,\\\label{g2V}
%g_2(\xi, \eta, \tau)&=&I^{1-\al}h(\xi,\eta, \tau),
%\eqa
For $(\xi,\eta)$ fixed, $h(\xi,\eta, t)$ is the solution to the Volterra equation
\bbe\label{VolterraV}
h(\xi,\eta, t)=\frac{1}{\Ga(\al)}\left(i\eta t^{\al-1}+\int_0^t(t-s)^{\al-1}F(\xi,h(\xi,\eta, s))ds\right),
\ee
where $F(\xi,h)$ is given by \eq{defF}. Applying $I^{1-\al}$ to \eq{VolterraV}
and taking into account that $I^{1-\al}t^{\al-1}=\Ga(\al)$,
%\[
%I^{\al-1} \frac{1}{\Ga(\al)}t^{\al-1}=1, \]
we obtain the following analog of 
Proposition \ref{prop:newRoughHestonChExp}.
\begin{prop}\label{prop:newRoughHestonChExpV}
Let $\al\in (0,1)$, $v, \ga,\theta,\nu>0$, $\rho\in (-1, 1)$, and let $h(\xi,\eta, t)$ be
the solution of \eq{VolterraV}. Then
\beqast\label{chFRoughV2}
\Phi_\al(t,T,v,\xi,\eta)
&=&\exp\left[i\eta \left(v+\frac{\ga\theta t^\al}{\Ga(\al+1)}\right)+\int_0^\tau (\ga\theta h(\xi,\eta, s)+vF(\xi,h(\xi,\eta, s)))ds\right].
\eqast
\end{prop}
Note that an accurate numerical solution of \eq{VolterraV} requires different and more involved
modifications of the Adams method than the ones in the present paper; we will consider \eq{VolterraV} in a separate publication.

\subsection{Modifications of the fractional Adams method}\label{ss: Adams}
In the Adams method, % (see Sect. \ref{ss:fract_Adams} for details), 
 one %sets $g(\xi,t)=F(\xi, h(\xi,t))$, 
 fixes a uniform grid $(t_j)_{j\in \bZ_+}$, $t_j=j\De$, and 
calculates the approximations $\hh(\xi,t_k), k=1,2,\ldots,$ in two steps. First, the predictor
$\hh^P(\xi,t_k), k=1,2,\ldots,$ is calculated, and then the more accurate approximation $\hh(\xi,t_k), k=1,2,\ldots.$
To calculate the former, the rectangular rule is used. At this step, in the region of large $|\xi|$ and small $t_j$,
significant errors appear. The errors are especially clearly seen
at the first step of the induction procedure, which we write explicitly:
\bbe\label{hhP10}
\hh^P(\xi, t_1)=b_{0,1}F(\xi,\hh(\xi,t_0)=b_{0,1}(-0.5(\xi^2+i\xi)),
\ee
where $b_{0,1}=\De^\al/\Ga(\al+1)$. The RHS of \eq{hhP10} is of the order of $\De^\al|\xi|^2$, however,
%it is easy to prove that as $c\to 0$,  
\bbe\label{Ash0}
\hh(\xi, t)= \frac{-0.5(\xi^2+i\xi)}{\Ga(\al+1)}t^\al(1+o(1)), 
\ee
uniformly in $(\xi,t)$ in the region $\{(\xi,t)\ |\ 0\le t^\al |\xi|^2< c\}$. See \cite{RoughNotTough},
where the full asymptotic expansion is calculated. We construct modifications of the Adams method
changing the prediction step so that the asymptotics \eq{Ash0} is taken into account.
We use the same coefficients $a_{j,k}$ as in the fractional Adams method. 
For $k=0,1,\ldots, M-1$, set 
\[
a_{k+1,k+1}=\frac{\De^\al}{\Ga(\al+2)}, \
a_{0,k+1}=\frac{\De^\al}{\Ga(\al+2)}(k^{\al+1}-(k-\al)(k+1)^\al),\]
and, in the cycle $j=1,2,\ldots, k$, calculate
\[
a_{j,k+1}=\frac{\De^\al}{\Ga(\al+2)}((k-j+2)^{\al+1}+(k-j)^{\al+1}-2(k-j+1)^{\al+1}).
\]
\subsubsection{Modification I}\label{sss:modif1} 
\begin{enumerate}[I.]
\item
Fix a small  $c>0$, e.g., $c=0.1$, and the number $n$ of iterations at each prediction
step $k$.
\item
Set $k=0$ and $\hh(\xi, 0)=0$. While $|\hh(\xi,t_k)|\le |\xi|/10$, 
\begin{enumerate}[(a)]
\item calculate
\beqa\label{eq:As1}
\hh_{as}(\xi, t_{k+1})&=& \frac{-0.5(\xi^2+i\xi)}{\Ga(\al+1)}t_{k+1}^\al,\\\label{hh0}
\hh_0(\xi,t_{k+1})&=&\sum_{0\le j\le k}a_{j,k+1}F(\xi,\hh(\xi,t_j)),\\\label{hh1}
\hh_1(\xi,t_{k+1})&=&\hh_0(\xi,t_{k+1})+a_{k+1,k+1}F(\xi,\hh_{as}(\xi,t_{k+1})).
\eqa
\item For $j=1,\ldots, n$, reassign 
\bbe\label{iter}
\hh_1(\xi,t_{k+1})=\hh_0(\xi,t_{k+1})+a_{k+1,k+1}F(\xi,\hh_1(\xi,t_{k+1})).
\ee
\item
Set $\hh(\xi,t_{k+1})=\hh_1(\xi,t_{k+1})$ and reassign $k:=k+1$. 
\end{enumerate}
\item
While $k<M$, calculate $\hh_0(\xi,t_{k+1})$ using \eq{hh0}, then set
\[
\hh_1(\xi,t_{k+1})=\hh_0(\xi,t_{k+1})+a_{k+1,k+1}F(\xi,\hh(\xi,t_{k})),
\]
and repeat steps (b), (c) above.

\end{enumerate} 
\subsubsection{Modification II}\label{sss:modif2} 
The piece-wise linear interpolation is inaccurate near 0
because $\dd_t\hh(\xi,t)$ is unbounded as $t\downarrow 0$. To decrease the interpolation error, we calculate numerically $\hh^1(\xi,t_{k+1}):=\hh(\xi,t_{k+1})-\hh_{as}(\xi, t_{k+1})$.
 Instead of the function
$F(\xi,h)$ given by \eq{defF}, the function
\bbe\label{defFas1}
F_{as1}(\xi,h_{as},h^1)=\ga(i\xi\rho\nu-1)(h_{as}+h^1)+\frac{(\ga\nu)^2}{2}(h_{as}+h^1)^2
\ee
is used. Set  $\hh^1(\xi, 0)=0$, and then, in the cycle $k=0,1,\ldots, M-1$
\begin{enumerate}[I.]
\item Calculate $\hh_{as}(\xi, t_{k+1})$ using \eq{eq:As1}, then evaluate
\beqa\label{eq:Ask}
\hh_0(\xi,t_{k+1})&=&\sum_{0\le j\le k}a_{j,k+1}F_{as1}(\xi,\hh_{as}(\xi,t_j),\hh^1(\xi,t_j)),\\\label{hh1k}
\hh^1(\xi,t_{k+1})&=&\hh_0(\xi,t_{k+1})+a_{k+1,k+1}F_{as1}(\xi,\hh_{as}(\xi,t_{k}),\hh^1(\xi,t_{k}))).
\eqa
\item For $j=1,\ldots, n$, reassign 
\bbe\label{iterk}
\hh^1(\xi,t_{k+1})=\hh_0(\xi,t_{k+1})+a_{k+1,k+1}F_{as1}(\xi,\hh_{as}(\xi,t_{k}),\hh^1(\xi,t_{k}))).
\ee
\item
Set  $\hh(\xi,t_{k+1}):=\hh^1(\xi,t_{k+1})+\hh_{as}(\xi, t_{k+1})$.
\item
Calculate  the integral on the RHS \eq{chFRough2} using the trapezoid rule.
\end{enumerate}

\begin{rem}{\rm One can use  asymptotic expansions of higher orders but in our numerical experiments
with the two-term asymptotic expansion, the latter brings no advantages.}
\end{rem}

\subsubsection{Modification III}\label{sss:modif3} Introduce $\tilde\hh(\xi,t):=(1+|\xi|)^{-1}\hh(\xi,t)$,\\
$\tilde\hh_{as}(\xi, t)=(1+|\xi|)^{-1}\hh_{as}(\xi, t)$, $\tilde\hh^1(\xi,t_{k+1}):=\tilde\hh(\xi,t_{k+1})-\tilde\hh_{as}(\xi, t_{k+1})$,
\[
\tilde F_{as1}(\xi,\tilde h_{as},\tilde h^1)=\ga(i\xi\rho\nu-1)(\tilde h_{as}+\tilde h^1)+(1+|\xi|)\frac{(\ga\nu)^2}{2}(\tilde h_{as}+\tilde h^1)^2,
\]
and then
\begin{enumerate}[I.]
\item
calculate $\tilde \hh^1(\xi,t_k), 0\le k\le M$, following the steps in Sect.~\ref{sss:modif2} with
$\tilde F_{as1}(\xi,\tilde h_{as},\tilde h^1)$ in place of $F_{as1}(\xi,h_{as},h^1)$;
\item
 set $\hh(\xi,t_k)=(1+|\xi|)(\tilde \hh^1(\xi, t_k)+
\tilde\hh_{as}(\xi, t_k))$, $0\le k\le M$;
\item
calculate  the integral on the RHS \eq{chFRough2} using the trapezoid rule.
\end{enumerate}
\begin{figure}
\begin{tabular}{cc}

 \begin{subfigure}[h]{0.45\textwidth}

 \centering
    \includegraphics[width=0.9\textwidth,height=0.4\textheight]{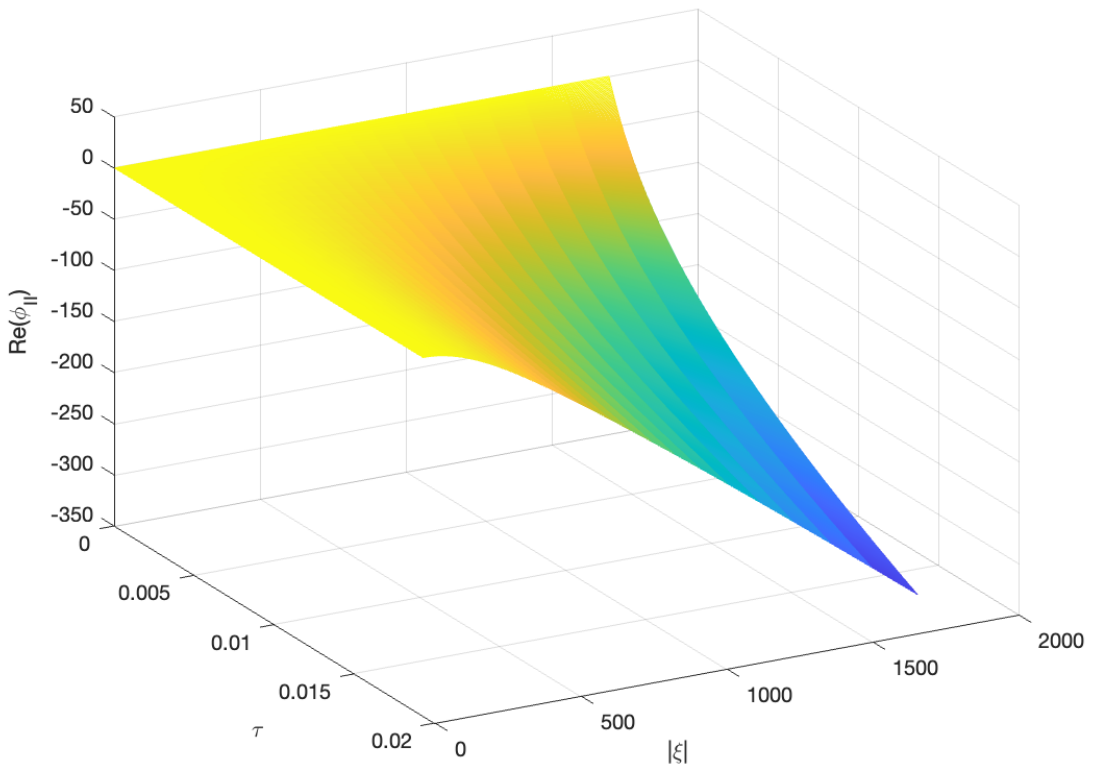}
    \caption{}\label{RePhiII}
\end{subfigure}
&
\begin{subfigure}[h]{0.45\textwidth}
\centering
    \includegraphics[width=0.9\textwidth,height=0.4\textheight]{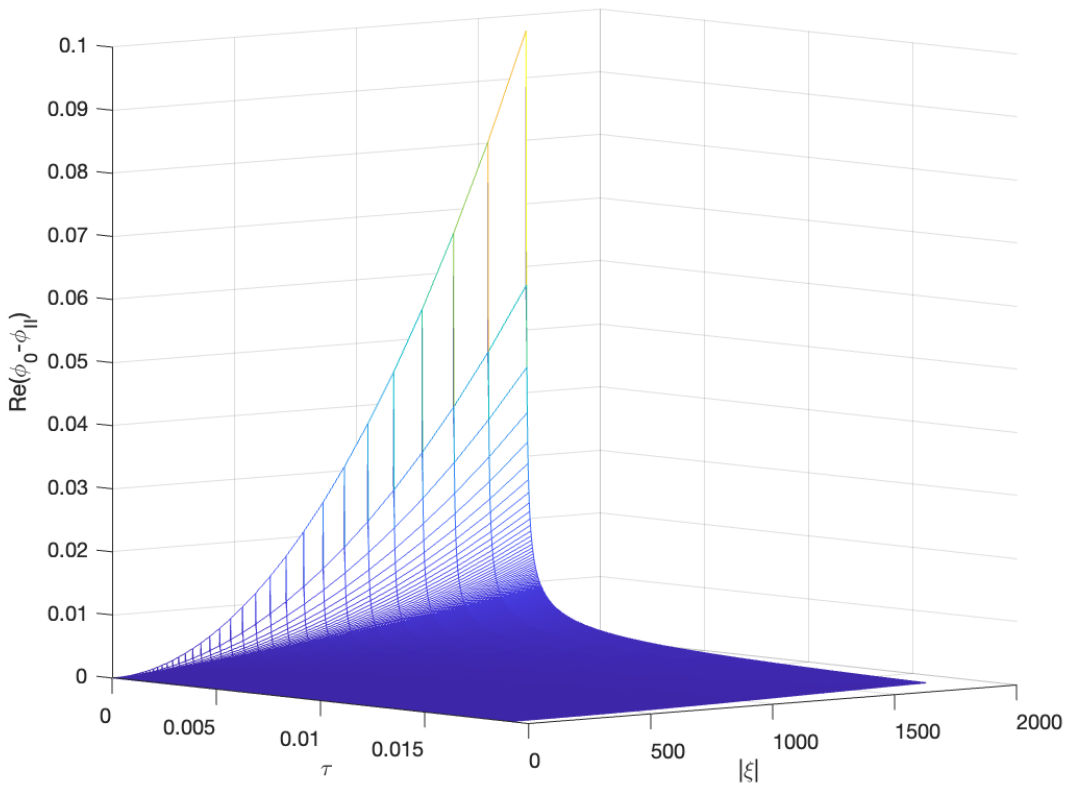}\caption{} \label{DifPhi1vsXia}
\end{subfigure}

\end{tabular}
\caption{\small Parameters of the model $\al=0.62$, $\ga=	0.1$, 	$\rho=-0.681$,
$\theta=	0.3156$, $\nu=0.331$, 
$v=0.0392$. (A): $\Re \phi$ calculated using modification II of the Adams method; (B): the difference between  $\Re\phi$ produced by the method 
 \cite{EuchRosenbaum2019} and $\Re\phi$ on Panel (A). In both cases, $T=1/52$ and $M=1000$.   The nodes $\xi$ are on the line
$\{\Im\xi=-1.5\}$. The differences shown on Panel (A) translate into the relative errors of evaluation of
the terms in the quadrature used. If $|\xi|$ are not small, it is seen that even at $T=1/52$, the errors are not negligible,
and for $T=1/252$, the errors are large. Note that even marginally accurate evaluation of options of short maturity requires
long grids.
 } 
\label{phi}
\end{figure}
\begin{rem}\label{rem:characteristic of modifications}
{\rm
Simple modifications that we use are in a very good agreement: in our numerical experiments, the difference between the results is, typically, less than $10^{-11}$. On the other hand, the differences with 
the results produced by the standard Adams method are sizable, and large for short maturitities and large $|\xi|$. See Panel  (B) in Fig.~\ref{phi}.
For options of moderate maturities, the real part of $\phi(t,T,v,\xi):=\ln \Phi_\al(t,T,v,\xi)$ in \eq{chFRough2}
decreases fast as the time to maturity $\tau:=T-t$ increases. See Panel  (A) in Fig.~\ref{phi}.
This explains why for options of moderate maturities,  the error of the simplified trapezoid rule with 
several dozens of terms is
smaller than E-06.
In the case of $\phi(t,T,v,\xi,\eta):=\ln \Phi_\al(t,T,v,\xi,\eta)$ in \eq{chFRoughV}, $\Re \phi(t,T,x,v,\xi,\eta)$ decreases
if $\Im\eta\to +\infty$ as well, hence, evaluation of more complicated options using the Fourier inversion in 2D
can be made fast as well: one needs to apply the sinh-acceleration w.r.t. $\eta$.
}
\end{rem}

\begin{rem}\label{rem:ModIIvsModIII}
{\rm
Modification III has the following two advantages: 
\begin{enumerate}[1.]
\item  if $|\xi|$ is very large, the stability of  the numerical solution of the fractional Riccati equation improves. In our numerical experiments, we encountered situations when a trajectory calculated using Modification II with $M=1000$ ``blows down" ($\Re\psi_{II}(\xi,t)\to -\infty$ as $t\to 0.5$)
but the ``blow down" does not occur if we increase $M=2000$ or use Modification III with $M=1000$. In the latter two cases,
$\phi_{II}(\xi,0.5)$ and $\phi_{III}(\xi,0.5)$ are very close;
\item
$\phi_{II}(\xi,t)$ is an analytic function of $\xi$ in the region where the blow-up or ``blow down" does not occur.  This makes $\phi_{II}(\xi,t)$ not quite suitable for the efficient error control
using the conformal bootstrap principle: one cannot efficiently control 
the errors of the evaluation of $\phi_{II}(\xi,t)$ because the error is an analytic function of $\xi$. $\phi_{III}(\xi,t)$ is not an analytic function.
\end{enumerate}
}
\end{rem}

\begin{rem}\label{rem:non-uniform}
{\rm 
The accuracy of calculations can be increased using  grids $\{t_k\}$ depending on $\xi$ (see Sect.~\ref{ss:xi0dependentgrids}).
%(the coefficients stemming from the piece-wise linear approximation are recalculated trivially). 
Non-uniform grids can be indispensable for pricing options of long maturities.

}
\end{rem}

\section{Several methods for Fourier inversion}\label{s:FT}

\subsection{Flat iFT and simplified trapezoid rule}\label{ss:Flat iFT and simpl. trap}
 In popular models, the characteristic function admits  analytic continuation to a strip around the real axis.
This implies that 
the following scheme (standard from the viewpoint of  Analysis) suggested in \cite{BL-FT,genBS,KoBoL} is more efficient than the scheme in 
\cite{heston-model} based on the L\'evy inversion formula.
Let the riskless rate $r\ge 0$ be constant, and let $S_T=S_0e^{X_T}$ be the price of the underlying non-dividend
paying asset (or index) at time $T$. Let $\Phi(\xi,T)=\bE[e^{i\xi X_T}] $ be the characteristic function of $X_T$ under a no-arbitrage measure $\bQ$ chosen for pricing (the expectation is conditioned on the spot values of additional factors as in SV models).
Then $\Phi(0,T)=1$, and if $\bE^\bQ[e^{X_T}]<\infty$, $\Phi(-i,T)=e^{rT}$. 
Assume
that there exist $\mum(T)<-1<0<\mup(T)$ s.t. for $\be\in (-\mup(T),-\mum(T))$, the exponential
moments $\bE^\bQ[e^{\be X_T}]$ are finite. 
Equivalently, $\Phi(\xi,T)$ admits analytic continuation to a strip 
$S_{(\mum(T),\mup(T))}:=\{\xi\ |\ \Im\xi\in (\mum(T),\mup(T))\}$.  Then the price of the call option
with strike $K$ and maturity $T$ can be calculated as follows. 
The payoff function $G(S_0,K,x)=(S_0e^x-K)_+$ admits a represention
\bbe\label{hG}
G(S_0,K; x)=\frac{1}{2\pi}\int_{\Im\xi=\om_1}e^{ix\xi}\hG(S_0,K;\xi)d\xi,
\ee
where  $\om_1\in (\mum(T),-1)$ is arbitrary, and $\hG(S_0,K;\xi)=-Ke^{i\xi\ln (S_0/K)}/(\xi(\xi+i))$ is the Fourier transform of
$G(S_0,K; x)$ w.r.t. $x$. We  substitute the integral representation \eq{hG} of $G(S_0,K; X_T)$ into the pricing formula $V_{call}(S_0, K;T)=e^{-rT}\bE[(S_0e^{X_T}-K)_+]$, and change the order of integration and summation
(the use of the Fubini theorem can be justified in all popular models). The result is 
\bbe\label{EuroCall}
V_{call}(S_0,K;T)=-\frac{Ke^{-rT}}{2\pi}\int_{\Im\xi=\om_1}\frac{e^{i\xi\ln (S_0/K)}\Phi(\xi,T)}{\xi(\xi+i)}d\xi.
\ee
Similarly, the price of the put is given by the RHS of \eq{EuroCall} with arbitrary $\om_1\in (0,\mup(T))$
(repeat the proof for the call starting with $G(x)=(K-e^x)_+$ or use the put-call parity on the LHS of \eq{EuroCall} and
the residue theorem on the RHS). The price of the covered call is given by the RHS with $\om_1\in (-1,0)$.
Since $\overline\Phi(\xi,T)=\Phi(-\bar{\xi},T)$ and $\overline{\hG(\xi)}=\hG(-\bar\xi)$, an equivalent form of
\eq{EuroCall} is
\bbe\label{EuroCallSym}
V_{call}(S_0,K;T)=-\frac{Ke^{-rT}}{\pi}\Re\int_{\Im\xi=\om_1}\frac{e^{i\xi\ln (S_0/K)}\Phi(\xi,T)}{\xi(\xi+i)}d\xi.
\ee
After truncation, the integral on the RHS of \eq{EuroCall} (or \eq{EuroCallSym}) can be calculated using either trapezoid rule or Simpson rule.

However, since the integrand on the RHS of \eq{EuroCall} is analytic in a strip
$S_{(\lm,\lp)}$ around the line of integration ($\lm=\mum(T), \lp=-1$ in the case of calls, $\lm=-1, \lp=0$ in the case of the covered call,
and $\lm=0, \lp=\mup(T)$ in the case of puts),
it is significantly more efficient to use the infinite trapezoid rule and then truncate the sum. The reason is
an exponential decay of the discretization error of the infinite trapezoid rule as the function of $\ze$, where $\ze$ is
the step. In Mathematical Finance, Lee \cite{Lee04} and Feng and Linetsky \cite{feng-linetsky08} were the first  to use this important property of the infinite trapezoid rule.
Let $H^1(S_{(\lm,\lp)})$ denote the Hardy space of functions analytic in the strip $S_{(\lm,\lp)}$ such that
\[
\int_{\lm}^{\lp} |f(\eta+i\om)|d\om \to 0\quad {\rm as}\ (\bR\ni)\eta\to\pm\infty\]
and the following analog of the Hardy norm 
%(we call it the {\em quasi-Hardy norm}: q-Hardy norm) 
is finite:
\begin{equation}\label{Hardynorm}
||f||_{S_{(\lm,\lp)}}:=\lim_{\om\uparrow \lp}\int_\bR|f(\eta+i\om)|d\eta+
\lim_{\om\downarrow \lm}\int_\bR|f(\eta+i\om)|d\eta<\infty.
\end{equation}
Fix $\om_1\in (\lm, \lp)$, and denote $d(\om_1)=\min\{\om_1-\lm, \lp-\om_1\}$. For $\ze>0$, construct a grid $\xi=i\om_1+\ze\bZ$,
and denote by $E_{\rm disc}(\zeta, \infty)$ the error of the infinite trapezoid rule
\[
\int_{\Im\xi=\om}f(\xi)d\xi\approx \ze\sum_{j\in\bZ}f(\xi_j)
\]
The following bound is proved in \cite{stenger-book} using the heavy machinery of the
sinc-functions (a simple proof can be found in \cite{paraHeston}):
%\begin{thm}\label{thm:discerr} (\cite[Theorem 3.2.1]{stenger-book})
\begin{equation}\label{discerrbound}
 \left|E_{\rm disc}(\zeta, \infty)\right|\le \frac{e^{-2\pi d(\om_1)/\zeta}}{1-e^{-2\pi d(\om_1)/\zeta}}||f||_{S_{(\lm,\lp)}}.
 \end{equation}
% \end{thm}
Let the error tolerance $\eps>0$ for the discretization error be small, and let $|\mu_\pm|$ be not too large.
 Then we 
choose $\om_1=(\lm+\lp)/2$, set $d(\om_1)=k_d(\lp-\lm)/2$, where $k_d<1$ is close to 1, e.g., $k_d=0.95$,
and use the following approximate recommendation: 
\bbe\label{ze_inf}
\ze=2d(\om_1)/\ln(100/\eps).
\ee
If the strip of analyticity is very wide, we choose a substrip around the line of integration with moderately large $|\la_\pm|$ and apply the prescription above. 

Once $\ze$ is chosen and the sum is truncated, we have the pricing formula. In the case of \eq{EuroCall}, 
\bbe\label{EuroCallSimp}
V_{call}(S_0,K; T)=-\frac{Ke^{-rT}\ze}{2\pi}\sum_{|j|\le N}\frac{e^{i\xi_j\ln (S_0/K)}\Phi(\xi_j,T)}{\xi_j(\xi_j+i)},
\ee
where $\xi_j=i\om_1+j\ze$. The number of terms can be decreased almost two-fold: similarly
to  \eq{EuroCallSym},
\bbe\label{EuroCallSimpSym}
V_{call}(S_0,K; T)=-\frac{Ke^{-rT}\ze}{\pi}\Re\sum_{0\le j\le  N}(1-\de_{j0}/2)\frac{e^{i\xi_j\ln (S_0/K)}\Phi(\xi_j,T)}{\xi_j(\xi_j+i)},
\ee
where $\de_{jk}$ is the Kronecker symbol. We call this method Flat iFT (flat inverse Fourier transform) method.
To choose $N$ so  that the truncation error is sufficiently small, it is necessary to know the rate of decay
of $\Phi(\xi,T)$ as $\xi\to\infty$ along the contour of integration. 
Let $\Phi(\xi,T)=\exp[\phi(\xi,T)]$, 
 and let 
an upper bound for $\Re \phi(\xi,T)$ be known:
\bbe\label{upperphi0}
\Re \phi(\xi,T)<-g(|\xi|,T),
\ee
where $g(|\xi|,T)$ is a monotonically increasing function of $|\xi|$.
Then the truncation of the series
at $|\xi|=\La_0$ introduces the error of the order of  $e^{-g(\La_0,T)}/\La_0$.
If an analytic formula for $\phi(\xi,T)=\ln\Phi(\xi,T)$ is available,
then an efficient bound \eq{upperphi0} can be derived. See \cite{paraHeston,pitfalls}. In the case of the rough Heston model, an analytic formula is not available.
In Sect.~\ref{ss:bound_psi},
we derive an approximate upper bound \eq{ash} for $\Re\phi(\xi,t)$ in the rough Heston model, and use the bound to derive a prescription for the choice of $\La_0$.  In the case of Flat iFT, \eq{ash} and the prescription 
should be used with $\om=0$. After $\La_0$ is calculated, we
set $N=\mathrm {ceil}\,\La_0/\ze$.

 % xigrid in Example T=40, M=40000 xigrid=0.0000 - 0.5000i   0.3331 - 0.5000i   0.8176 - 0.5000i   1.5144 - 0.5000i

%.  T=2 \vec{h}_77=--20.5806 +18.9985i lim -22.1234 +20.5740i
\subsection{Carr-Madan method}\label{ss:CM method}
The implementation of Flat iFT is very simple, and can be easily parallelized if the option prices for  several dozen of strikes need to be calculated. Nevertheless, in noughties, the unnecessary complicated (and slower and less accurate)  CM method
 became popular. The main idea of the method is to use the Fast Fourier transform (FFT) to evaluate
the option prices at several strikes. However, FFT produces the results at points of uniformly spaced 
grids in the $\ln(K)$-space. Therefore, in order to evaluate the option prices for given strikes,
an interpolation procedure needs to be employed. To satisfy even a moderate error tolerance,
a fine grid $x_j=x_0+j\De$, $j=1,2,\ldots, M=2^m$, with $\De<<1$ is necessary; to make an accurate Fourier inversion,
a small step $\ze$ in the dual space must be used (in \cite{carr-madan-FFT}, $\ze=0.25$ or $\ze=0.125$ are recommended). The Nyquist  relation $\De\ze=2\pi/M$ requires $M$ to be of the order of several thousand. In \cite{carr-madan-FFT},
the basic recommendation is $M=4,096$ and it is mentioned that larger $M=8,192$ or $M=16,384$ may be needed. Hence, the calculations become computationally more costly as compared to
Flat iFT, and an unnecessary interpolation error is introduced. Table \ref{table: iFT-FFT}
illustrates the adverse impact of the interpolation errors on the quality of calibration: the number of
strikes for which the calculated prices are outside the no-arbitrage bounds  increases because of the interpolation.
In the case of the rough Heston model, the evaluation of $\Phi(\xi,T)$ for $\xi$ large in absolute value, accurate calculations are especially difficult and time consuming.
%In Sect.~\ref{s:numer}, we demonstrate that
The implied volatility surface produced by CM method 
can significantly differ from the correct one 
(see Fig.~ \ref{Set1ImpVolsurfacesXiT152}).
%,\ref{Set2ImpSurfacesT152}). 
In particular, essentially flat volatility curves
can become nice volatility smiles, and changing the dampening factor (the line of integration) in CM method and keeping the same step size
 and the grid size recommended in CM method one can significantly change the smiles and surface. The  implied volatility surface can significantly change as one changes the step and/or grid size. Furthermore, the errors are systematic,
 and, in many cases, prices of deep OTM options produced by CM method are ``prices" of the systematic errors of the method, which can be ``useful".\footnote{See Sect.~\ref{DistortionCM} for  
explanations for the popularity of the CM method in noughties. During
the financial crisis, CM method  caused problems in the financial industry, and has net been used by practitioners ever since.}

\subsection{Lewis method}\label{ss:Lewis}
%\subsection{Lipton-Lewis formula and Lewis method}\label{LLformula}
The specific choice $\om_1=-1/2$ was suggested later by A.Lewis and A.Lipton 
\cite{Lewisbook,liptonFX}, and the formula for the covered call was rewritten in the  form
\bbe\label{EuroPriceLL}
V(S_0, K; T) = -\frac{(K/S_0)^{1/2}}{\pi}\Re\int_0^{+\infty}
\frac{e^{iy\ln(S_0/K)}\Phi(T,-i/2+y)}{y^2+0.25}dy.
\ee
In the Lewis method \cite{lewisFT}, it is recommended to change the variable in order to reduce to the integral over $(0,1)$, and then apply the Gauss-Legendre quadrature. In our numerical examples, we use the change of variables $y=u/(1-u)$,
and demonstrate that the Lewis method requires the evaluation of more terms
than the sinh-acceleration and the CPU time is several times larger even in a favorable region not too close to maturity and not too far from the spot;  the accuracy better than E-08 is almost impossible to achieve using double precision arithmetic. The reason is that the Lewis method
does not take into account the fact that an accurate method should not use too many nodes $u_j$ close to 1.
The corresponding $y_j$ are very large, hence, the characteristic function is very difficult to evaluate accurately. If adaptive quadratures are used, the efficiency of the method can be improved but the problem does not vanish and can become worse since larger $y_j$ can appear. In addition, as it is demonstrated in an extensive
empirical study of the calibration of the Heston model in \cite{HestonCalibMarcoMeRisk}, this approach is slower than the fractional-parabolic deformation technique
(the sinh-acceleration used in the present paper is more efficient than the  fractional-parabolic deformation),  and a reliable error control is almost impossible especially if the integrand cannot be evaluated very accurately.  
We produce numerical examples to show that the sinh-acceleration is faster and more accurate that the Lewis method
when applied to the rough Heston model.

Finally, note that the performance of Lewis method strongly depends on the choice of the change of variables.
In the example in the present paper, the errors  increased greatly when we used $y=-\ln u$ instead of $y=u/(1-u)$.

\subsection{COS method}\label{ss:COS}
%As we showed in \cite{iFT0,iFT}, in the case of pricing of European options, the pricing formula of the COS method
%is the approximation of the price of the option with a modified payoff in Flat iFT method. 
As we explained in \cite{iFT0,iFT}, 
 ``COS method \cite{COS} is based on an approximation of the pdf by a linear combination of cosines.
In the case of pricing European options,  after several additional approximation steps, the explicit pricing formulas
are derived. 
 The resulting formulas of COS method  can be derived differently and simpler. In the case of the call option
with strike $K=1$, 1) choose $a<0<b$; 2) define $G_{ap}(x)=(e^x-1)_+, a\le x\le b$, $G_{ap}(x)=(e^{2b-x}-1)_+, b\le x\le 2b-a$, and
extend $G_{ap}$ as a periodic function with the period $2(b-a)$; 3) replace the payoff of the call with $G_{ap}(X_T)$; 4)
expand $G_{ap}(x)$ into the Fourier series, substitute $G_{ap}(X_T)$ into the standard iFT formula ({\em with the integration over the real line}: if a different line is used, the errors significantly increase), and apply the 
simplified trapezoid rule. The result is of the form  \eq{EuroCallSimpSym} with $\xi_j=j\ze$, $\ze=\pi/(b-a)$, and
$\hG(\xi_j)$ in place of $-1/(\xi_j(\xi_j+i)$, where $G(x)=\bfo_{(a,2b-a)}(x)G_{ap}(x)$.
The procedure is applied to the put as well. In the case of the call, the modified payoff has a rather sharp kink at $x=b$,
therefore, typically, the resulting payoff modification error and discretization error are 
larger in the case of calls. This explains the recommendation
in \cite{COS} to calculate put prices, and then use the put-call parity to calculate prices of calls. However, depending
on the properties of the characteristic exponent, it is possible that the calculation of calls entails smaller errors."

From the complex-analytical point of view, the advantage of COS method is that the Fourier transform of the modified payoff
is an entire function, hence, the strip of analyticity is wider, and the discretization error smaller. On the other hand,
as compared with flat iFT, COS method has an additional source of errors: the payoff modification error. Furthermore,
the line of integration $\bR$ is fixed. If the upper boundary of the strip of analyticity of $\Phi(\xi,T)$ is close to $\bR$: $\mup\in (0,0.5)$, then the Lipton-Lewis choice of the line of integration $\{\xi\ |\ \Im\xi=-0.5\}$ leads to a smaller discretization error
than the one of COS method (hence, less of CPU time is wasted),
and no additional error is introduced. 
Next, the ad-hoc error estimates in \cite{COS} are complicated, formulated in terms of moments 
(which cannot adequately assess the discretization error of the infinite trapezoid rule), and a sufficiently good $N$ is chosen using the doubling procedure. A simplified prescription $[a,b]=[-L\sqrt{T},L\sqrt{T}], L\in [6,12]$, in \cite{COS}
implies $\ze=\pi/(2L\sqrt{T})\in T^{-1/2}[0.1309,0.2618]$. Thus, at the first step, the approximation and discretization errors are fixed, and the papers using COS method implicitly presume that these errors are small. However, it is evident
that unless the strip of analyticity is very wide, the resulting discretization error is very large for options of short maturity,
and, for options of long maturity, significantly larger $\ze$'s are admissible if the Flat iFT with the Lipton-Lewis choice of the line of integration is applied.
   In \cite{COS} and other papers that use  COS method,
the ``efficiency" of the method is illustrated as follows. (1) Choose two parameters of COS method which control   the approximation  and discretization errors, without accurate error bounds. Thus, both errors are fixed. (2) Increase $N$ two-fold until the difference becomes smaller 
than the desired error tolerance.
In \cite{iFT0,iFT}, we produced several examples that show that even in the case of L\'evy processes,
it is very difficult to satisfy the error tolerance better than E-07 although the ``doubling error" becomes negligible. 
The same scheme is used in \cite{KamuranEmreErkan2020} to ``prove" that the absolute errors of the ``prices"
presented in Tables in \cite[Section 6]{KamuranEmreErkan2020} are of the order of 9E-04 if $N=160$, and 0 if $N=320$; the CPU time of the order of 1 sec. per point (see \cite[Table 4 on p.55]{KamuranEmreErkan2020}).
In fact, the errors are much larger. See Table \ref{table: COS1_K},
where we reproduce the results  in \cite[Table 1 on p.52]{KamuranEmreErkan2020} and show more accurate results produced using the method of the present paper at the  CPU cost thousands of times smaller. The total CPU time
for calculation option prices at thousands of points $(K/S_0,T)\in [0.8, 1.2]\times [0.5, 1]$ is less than several msec. The calculated quantities are
\bbe\label{VCOSex}
V_{``call"}(S_0,K,T)=e^{-rT}\bE[(S_T-K)_+\ |\ S_0, V_0=v],
\ee
where $S_0=100, K=80,100, 120$ and $(S_t,V_t)$ is the process in the rough Heston model with 
 the parameters 
%  \bbe\label{COSEx6.1}
 $ \al=0.6,  \ga= 0.1,\ \theta = 0.3156,  \nu = 0.331, \rho=-0.681,  v = 0.0392$,
 $r=0.3$.
 In \cite{KamuranEmreErkan2020}, $V(S_0,K,T)$ given by \eq{VCOSex} is called the call option price.
 However, in the rough Heston model, by construction, $\bE[S_T]=S_0$, hence, \eq{VCOSex} with $r>0$ violates
 the no-arbitrage principle. Since we are interested in the numerical performance of COS method and not in the calibration
 of the model to the real data, we use this example nonetheless. The quality of a method  manifests itself in
  the relative errors of the calculation of OTM options. In Table \ref{table: COS1_K}, we show the errors and relative
  errors of the ATM and OTM ``calls" given by \eq{VCOSex}
and OTM ``put" given by   $V_{``put"}(S_0,K,T)=V_{``call"}(S_0,K,T)-e^{-rT}(S_0-K).$
 % \beqast
 %V_{``put"}(S_0,K,T)&=&e^{-rT}\bE[(K-S_T)_+\ |\ e^{X_0}=S_0, V_0=v]\\
% &=&V_{``call"}(S_0,K,T)-e^{-rT}\bE[S_T-K]\\
% &=&V_{``call"}(S_0,K,T)-e^{-rT}(S_0-K).
% \eqast
In the numerical example in  \cite[Section 6.2]{KamuranEmreErkan2020}  (parameters $\al=0.6$, $\ga=2$, 	
$\theta=	0.025$, $\nu=0.2$, $\rho=-0.6$, $v_0=0.025$; $K/S_0\in [0.8,1.2]$), as in \cite{EuchRosenbaum2019}, the implied volatility curves are presented in the form of graphs which we cannot reproduce. 
 In Fig. \ref{Set2Curves}, we show implied volatility curves calculated using the method of the present paper.
 The curves constructed in \cite{KamuranEmreErkan2020} using the COS method are different (the reader can download \cite{KamuranEmreErkan2020} and compare the curves in
 \cite{KamuranEmreErkan2020} and  Fig.~\ref{Set2Curves}
 in Appendix),
and even the curves for moderate maturities (for moderate maturities, accurate calculations using Flat iFT are easy)
are rather irregular.
Furthermore, apparently, the COS method used in \cite{KamuranEmreErkan2020} fails
to produce results for the maturity 1 day and for log-strikes outside a small
vicinity $\ln K\in [-0.06,0.03]$, if $T=0.01,0.02,01$.
For $T=1/252$, the curve obtained using the sinh-acceleration is reliable for $\ln (K/S_0)\in [-0.21,0.075]$, and for $T=0.01$, in the region  
$\ln (K/S_0)\in [\ln(0.8),0.13]$. In the other cases, the curves are reliable in the region  $K/S_0\in [0.8,1.2]$.

%{\sc Example in   \cite[Sect. 6.1]{KamuranEmreErkan2020}.} $S_0=K_0=100$.

\subsection{Flat iFT-BM and Flat iFT-NIG methods}\label{ss:Flat iFT-BM} A number of numerical examples in \cite{iFT0,iFT}
demonstrate that COS method is slower than Flat iFT method although in rare exceptions (processes very close to the Brownian motion (BM), when about 20 terms suffice to satisfy a very small error tolerance),  COS method is marginally faster: the number of terms decreases by 1 or 2 - but if chosen by hand and not using the universal recommendations.
The additional error of  COS method is partially compensated by the increase of the width of the strip of analyticity
around the line of integration: instead of one of the three strips $S_{(\mum(T),-1)}, S_{(-1,0)}, S_{(0,\mup(T))}$,
the strip $S_{(\mum(T),\mup(T))}$ can be used. In this section, we demonstrate that the same effect is achievable without introducing additional errors. We use the same straightforward idea as in \cite{ConfAccelerationStable}, where
we eliminated the zero at $\xi=0$ of the integrand in the formula for the cumulative probability distribution function of a stable L\'evy process. In the current setting, we eliminate two zeros, at $\xi=0$ and $\xi=-i$.
Let $\Phi_{ad}(\xi,T)$ be the characteristic function in
a model, where vanilla prices can be calculated faster than in the initial model. Denote by $V_{call}(\Phi; S_0,K;T)$
the call price in the model with the characteristic function $\Phi$; as above, the asset pays no dividends
and interest rate $r$ is constant. 
\begin{prop}\label{prop: Flat iFT-BM}
Let $\Phi(\xi,T)$ and $\Phi_{ad}(\xi,T)$ admit analytic continuation to a strip $S_{(\mum(T),\mup(T))}$,
where $\mum(T)<-1<0<\mup(T)$,
and let $\Phi(-i,T)=\Phi_{ad}(-i,T)=e^{rT}$. 

Then, for any $\om_1\in (\mum(T),\mup(T))$,
\bbe\label{eq:Flat iFT-add}
V_{call}(\Phi; S_0,K;T)=V_{call}(\Phi_{ad}; S_0,K;T)-\frac{Ke^{-rT}}{2\pi}\int_{\Im\xi=\om_1}\frac{e^{i\xi\ln (S_0/K)}
(\Phi(\xi,T)-\Phi_{ad}(\xi,T))}{\xi(\xi+i)}d\xi.
\ee
The equality \eq{eq:Flat iFT-add} is valid for put and covered calls as well.
\end{prop}
\begin{proof} Let $\om_1\in (\mum(T),-1)$. Then \eq{eq:Flat iFT-add} is valid. The apparent singularities of the integrand are removable because $\Phi(\xi,T)-\Phi_{ad}(\xi,T)$ is analytic in the strip $S_{(\mum(T),\mup(T))}$
and $\Phi(\xi,T)-\Phi_{ad}(\xi,T)=0$ at $\xi=0,-i$. Hence, the integrand on the RHS of \eq{eq:Flat iFT-add}
is analytic in the strip, and one may move the line of integration to any line 
$\{\xi\ |\ \Im\xi=\om_1\}, \om_1\in (\mum(T),\mup(T))$.
The proof for puts and covered calls is essentially the same. 
\end{proof}
The integral on the RHS of \eq{eq:Flat iFT-add} is calculated using Flat iFT.
If $\mup(T)-\mum(T)>>1$, and $\om_1(T)=(\mup(T)+\mum(T))/2$ is chosen, the half-width
of the strip of analyticity used to derive the recommendation for the choice of the step $\ze$ and $\ze$
increase significantly, and the number of terms of the simplified trapezoid rule and CPU time decrease,
also significantly.

Natural choices for $\Phi_{ad}$ are the characteristic functions in the following models: 
\begin{enumerate}[(1)]
\item
 the BM with the characteristic exponent $\psi(\xi)=\sg^2\xi^2/2-i\mu\xi$; $\sg>0$,
 $\mu=r-\sg^2/2$;
 \item
 Normal Inverse Gaussian process  (NIG) \cite{B-N} or the generalization of NIG 
 (tempered stable L\'evy processes (NTS) constructed in \cite{B-N-L}), with the same or wider strip of analyticity;
 \item
 in applications to rough Heston model, it is feasible that the use of $\Phi_{ad}$
 in the Heston model with the same parameters $\ga, \theta, \nu, \rho$ can be advantageous.
 \end{enumerate}
 We call the resulting method with the choices (1) and (2) {\em Flat iFT-BM} and {\em Flat iFT-NIG} 
 (more generally, {\em Flat iFT-NTS}) methods.
 In the numerical examples in the paper, we use the simplest variant: Flat iFT-BM. 
 In our numerical examples that we considered, the analogs: Lewis-BM and SINH-BM mejods do not bring advantages as
 compared with the Lewis and SINH-methods.
 We live to the future the study of possible advantages of choices (2) and (3).

\subsection{Summation by parts in the inifinite trapezoid rule}\label{ss: summation by parts}
For the explicit formulas, see \cite{Contrarian}. The summation by parts significantly decreases
the product $\ze N$ necessary to satisfy the given error tolerance if the strike is not close to
the spot.  Hence, it is natural to separate the region of strikes into two regions: close to the spot,
where Flat iFT-BM (or Flat iFT-NIG) is used, and the region farther from the spot, where, in addition, the
summation by parts is used.

\subsection{SINH-acceleration}\label{ss: SINH}
In the real-analytic interpretation \cite{carr-madan-FFT},  choices of different lines
of integration are choices of different {\em dampening factors}. In Complex Analysis, one observes that the Fourier transform $\hf$ of a sufficiently regular function $f$ is an analytic function in a wide region $\cU_0$ of the complex plane and meromorphic function in a wider 
region $\cU$.  We choose $\cU$ so that $\hf(\xi)\to 0$ sufficiently fast as 
$\xi\to \infty$ remaing in $\cU$. The inverse Fourier transform can be calculated deforming the line of integration into any sufficiently regular curve in $\cU_0$; crossing poles, one can reduce
to the integral over any sufficiently regular curve in $\cU$ (plus residues at the poles crossed
in the process of deformation). In the case of the Heston model, under additional restriction on
the parameters, it is  proved in \cite{Lucic}  
that $\Phi(T,\xi)$ is analytic in $\cU=\bC\setminus i((-\infty,\mum(T)]\cup [\mup(T),+\infty))$, where $\mum(T)<-1<0<\mup(T)$; in  \cite{paraHeston}, this fact is proved for jump-diffusion generalizations of the Heston model, with more than one factor driving the dynamics of the volatility process,
and algebraic equations for $\mum(T)$ and $\mup(T)$ were derived. 
For wide classes of affine jump-diffusion processes, it is proved in \cite{pitfalls} that $\Phi(\xi,T)$ is an analytic 
function on the union  $\cU_0(\mum(T),\mup(T),\gam,\gap)$ of a strip
$S_{(\mum(T),\mup(T))}$, where $\mum(T)<-1<0<\mup(T)$, and a cone
$\cC_{\gam,\gap}:=\{\xi=\rho e^{i\varphi}\ |\ \varphi\in (\gam,\gap) \vee 
 \varphi\in (\pi-\gam, \pi-\gap) \}$, where $\gam\in (-\pi/2,0), \gap\in (0,\pi/2)$
 (typically, $\ga_\pm=\pm\pi/4$), and decays as $\xi\to \infty$ remaining in the cone. 
 Once the existence of such a strip and cone is established, we choose a deformation of the contour of integration  into a contour $\cL_{\om_1,b,\om}:=\chi_{\om_1,b,\om}(\bR)$, where the conformal map $\chi_{\om_1,b,\om}$
 ({\em sinh-deformation}) is defined by
 \bbe\label{eq:sinh}
 \chi_{\om_1,b,\om}(y)=i\om_1+b\sinh(i\om+y),
 \ee
and $\om_1\in \bR$, $b>0$ and $\om\in (\gam,\gap)$. The parameters of the deformation
 are chosen so that in the process of deformation, the contour remains in the domain of analiticity
 of $\Phi(\xi,T)$, and the singularities at $0$ and $-i$ are not crossed. The deformation being made,
 we change the variable $\xi=\xi(y)=\chi_{\om_1,b,\om}(y)$ in \eq{EuroCall}
 \bbe\label{EuroPricesinh}
V(S_0, K;T) = -\frac{bKe^{-rT}}{2\pi}\int_{\bR}\frac{e^{i\xi(y)\ln(S_0/K)}\Phi(\xi(y),T)}{\xi(y)(\xi(y)+i)}\cosh(i\om+y)dy,
\ee
and apply the simplified trapezoid rule: 
\bbe\label{EuroPricesinhtrap}
V(S_0, K;T) = -\frac{b\ze Ke^{-rT}}{\pi}\Re\sum_{j=0}^N e^{i\xi(j\ze)\ln(S_0/K)}g(j\ze,T)(1-\de_{0j}/2),
\ee
where
$
g(y,T)=\frac{\Phi(\xi(y),T)}{\xi(y)(\xi(y)+i)}\cosh(i\om+y)$. Explicit recommendations for the choice
of the parameters of the deformation $\om_1,b,\om$ and parameters $\ze,N$ of the simplified trapezoid rule
are derived in \cite{SINHregular}. We add several useful details.
\begin{enumerate}[I.]
\item
Find $\mu_\pm(T)$ and $\ga_\pm$. 
\item
If $S_0\le K$, use $\om\le 0$ and calculate the price of either the call or covered call; 
if $S_0\ge K$, use $\om\ge 0$ and  calculate the price of either the put or covered call.
\item
\begin{enumerate}[(a)]
\item
 If the call is priced, set $\lm=\mum(T), \lp=-1$, $\om=\gam/2$, $d_0=-\om$.
\item
If the put is priced, set $\lm=0, \lp=\mup(T)$, 
$\om=\gap/2$, $d_0=\om$.
\item
If the covered call is priced, set $\lm=-1,\lp=0$. If $S_0< K$, set
$\om=\gam(T)/2$, $d_0=-\om$.
If $S_0> K$, set
$\om=\gap(T)/2$, $d_0=\om$.
\item
For ATM options, it is optimal to set $\om=(\gam+\gap)/2$, $d_0=(\gap-\gam)/2$.
\end{enumerate}
\item
Choose $k_d<1$ close to 1, e.g., $k_d=0.9$, and set $d=k_dd_0$, $\ze=2\pi d/\ln(100/\eps)$,
%\item
%Set $b=(\lp-\lm)/(\sin(\om+d)-\sin(\om-d))$ and 
\[
b=\frac{\lp-\lm}{\sin(\om+d)-\sin(\om-d)},\ \om_1=\frac{\lm\sin(\om+d)-\lp\sin(\om-d)}{\sin(\om+d)-\sin(\om-d)}.
\]
%\beqast
%b&=&(\lp-\lm)/(\sin(\om+d)-\sin(\om-d)),\\
%\[ \om_1=(\lm\sin(\om+d)-\lp\sin(\om-d))/(\sin(\om+d)-\sin(\om-d)).\]
%\eqast

\item
As in the case of Flat iFT, to choose $N$ so  that the truncation error is sufficiently small, it is necessary to know the rate of decay
of $\Phi(\xi,T)$ as $\xi\to\infty$ along the contour of integration. 
Let $\Phi(\xi,T)=\exp[\phi(\xi,T)]$, and let 
an upper bound \eq{upperphi0} for $\Re \phi(\xi,T)$ be known.
 In the $y$-coordinate, the series decays as
$(K\ze b/\pi)e^{-g(|\xi(y_j)|,T)}/|\xi(y_j)|$. Since $|\xi(y)|$ increases as an exponential function of $y$ as $y\to\pm\infty$,  the truncation error is smaller than the last term of the truncated sum if $g(|\xi(y_j)|,T)$ is larger.
We find the positive solution $\La_0$ of the equation $
e^{-g(\La_0,T)}/\La_0 = b\pi\eps/(K\ze)$, and set $\La=\ln(2\La_0/(Kb))$,  $N=\mathrm {ceil}\,\La/\ze$.
\end{enumerate}
\begin{rem}\label{rem:choice of ze}{\em The recommendation $\ze=2\pi d/\ln(100/\eps)$
 presumes that $||f||_{S_{(\lm,\lp)}}$, the analogue \eq{Hardynorm} of the Hardy norm of the integrand, is bounded by 100.  A safer alternative which we used in several publications is to use the approximation 
$||f||_{S_{(\lm,\lp)}}\approx |f(i(\om+d))|+|f(i\om-d))|$.
}
\end{rem}

\subsection{Ad-hoc bound for $\phi$ and choice of $N$ in the rough Heston model}\label{ss:bound_psi}
Possibly, the bound informally derived in this section can be rigorously proved but
it requires a subtle work with joint asymptotic expansions as $t\to+\infty$ and $\xi\to \infty$ in a cone around
$\bR$. The formally proved  asymptotic expansion in  \cite{GatheralRadoicic2019} (all terms of the expansion in \cite{GatheralRadoicic2019}  bar the leading term lead to divergent integrals when substituted into
the fractional Riccati equation)
%\footnote{all terms of the expansion in \cite{GatheralRadoicic2019}  bar the leading term lead to divergent integrals when %substituted into
%the fractional Riccati equation. Furthermore, a rigorous proof of any bound requires
%one to work with the corresponding Volterra equation, and the non-local nature of the problem implies that
%it is necessary to control the contribution of the integral over $(0,1)$. This integral is of the order of $O(t^{\al-1})$
%but the next term of the asymptotics in  \cite{GatheralRadoicic2019} is of the order $t^{-\al}$, contradiction.}  
implies that, for $\xi$ fixed and 
$t\to+\infty$, the leading term of the asymptotics of $h(\xi,t)$ is the same as in the Heston model:
\bbe\label{aspsitinf}
h(\xi,t)\sim h^\infty(\xi):=(\ga\nu)^{-2} [-\ga(i\rho\nu\xi-1)-[(\xi^2+i\xi)(\ga\nu)^2+\ga^2(i\rho\nu\xi-1)^2]^{1/2}].
\ee
As $\xi\to \infty$ along any ray in the right half-plane, 
\bbe\label{hinfinf}
h^\infty(\xi)= -\xi \frac{i\rho+\sqrt{1-\rho^2}}{\ga\nu}+ O(1).
\ee
Set $\ka_\infty=\sqrt{1-\rho^2}/(\ga\nu)$. Motivated by \eq{hinfinf},
we surmise that  there exist a cone $\cC$ round $\bR_+$ and $R_0,T_1>0$ such that
for all $\xi\in \cC$ s.t. $|\xi|\ge R_0$ and all $t>T_1|\xi|^{-1/\al}$
\bbe\label{asboundh1}
\Re h(\xi,t)\le -\frac{\ka_\infty}{2}\Re\xi.
\ee
If $T_0>0$ is sufficiently small, $t\in (0,T_0|\xi|^{-1/\al})$, and $(\cC\ni)\xi\to \infty$,
\bbe\label{as0h}
\Re h(\xi,t)\sim -\frac{t^\al}{2\Ga(\al+1)}\Re\xi^2.
\ee
Let $\xi = y e^{i\om}$, where $y>0$ and $\om\in (-\pi/4,\pi/4)$. From the definitions \eq{eq:g1g2} of $g_1,g_2$,
it is immediate that simple approximate upper bounds for $\Re g_j, j=1,2$, hence, for $\Re\phi$  can be obtained
using \eq{asboundh1} and \eq{as0h} separately, and taking the maximum of the results.
We obtain
\bbe\label{ash}
\Re \phi(\xi,t)\le -\min\{(G_1(t)/2) \cos\om y,G_2(t) \cos(2\om) y^2\},
\ee
where  
\beqa\label{psiG1}
G_1(t)&=&\ka_\infty (\theta\ga t+v t^{1-\al}/(1-\al)),\\
\label{psiG2}
G_2(t)&=&\theta\ga\frac{ t^{1+\al}}{2\Ga(2+\al)}+v t.
\eqa
The accuracy of the approximation \eq{hinfinf} and bound \eq{ash} increases as $|\xi|$ increases 
but is far from perfect. See Fig.~\ref{CurvePhiSet1T10} for an illustration.

Given a small error tolerance $\eps>0$, we derive an approximation to the truncation parameter $\La$ letting $C=10$, and

1) find the positive solution $\La_{01}$ of the equation $((G_1(t)/2)\cos\om) y+\ln y- E=0$, where $E=\ln(CK/(\pi\eps)$, using the Newton method with the initial approximation 1;

2) find the positive solution $\La_{02}$ of the equation $(G_1(t)\cos(2\om))y^2+\ln y- E=0$ making the change of the variable
$y_1=y^2$ and using the Newton method with the initial approximation 1;

3) set $\La_0=\max\{\La_{01},\La_{02}\}$ and $\La=\ln(2\La_0/b)$, $N=\mathrm{ceil}\,\La/\ze$.

%Set 1: $\xi_{15}=1.1299 - 0.5 i, \xi_{30}=4.3712 - 0.5 i, \xi_{45} =15.5301 - 0.5 i, \xi_{56} =39.1661 - 0.5 i $
%Set 2:    54.8196 - 0.5000i 83.4551 - 0.5000i 1.2705e+02 - 5.0000e-01i 1.9341e+02 - 5.0000e-01i

 \begin{figure}
\begin{tabular}{cc}

 \begin{subfigure}[h]{0.45\textwidth}

 \centering
    \includegraphics[width=0.9\textwidth,height=0.4\textheight]{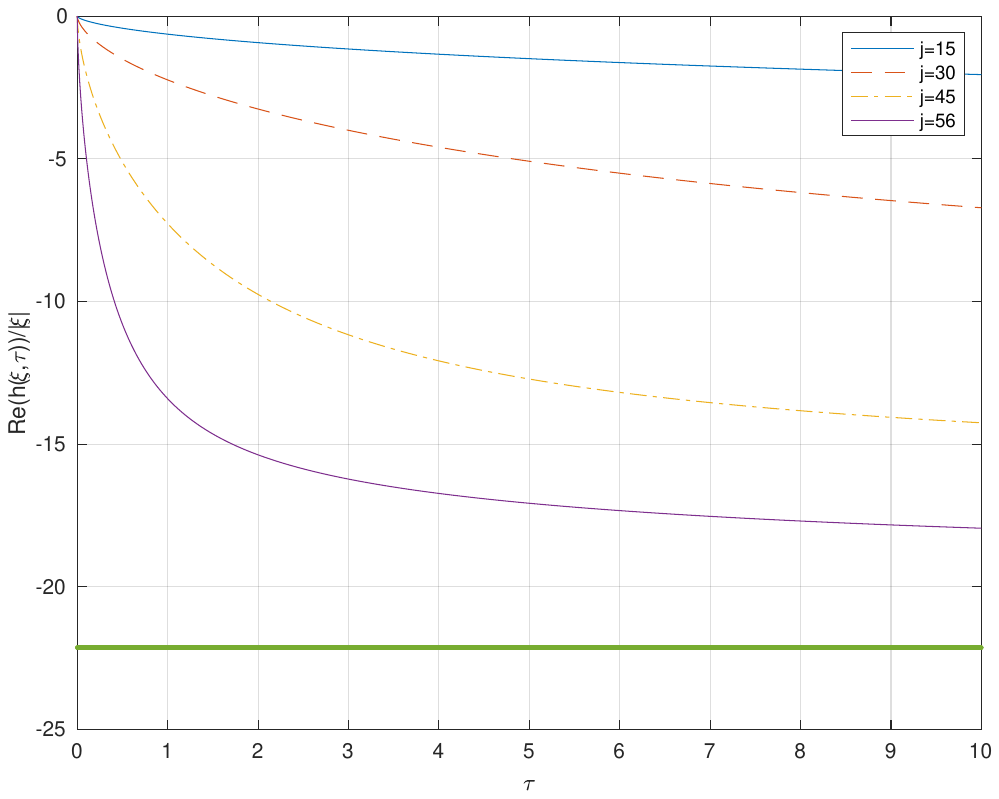} 
    \caption{}\label{CurveHSet1T10}
\end{subfigure}
&
\begin{subfigure}[h]{0.45\textwidth}
\centering
    \includegraphics[width=0.9\textwidth,height=0.4\textheight] {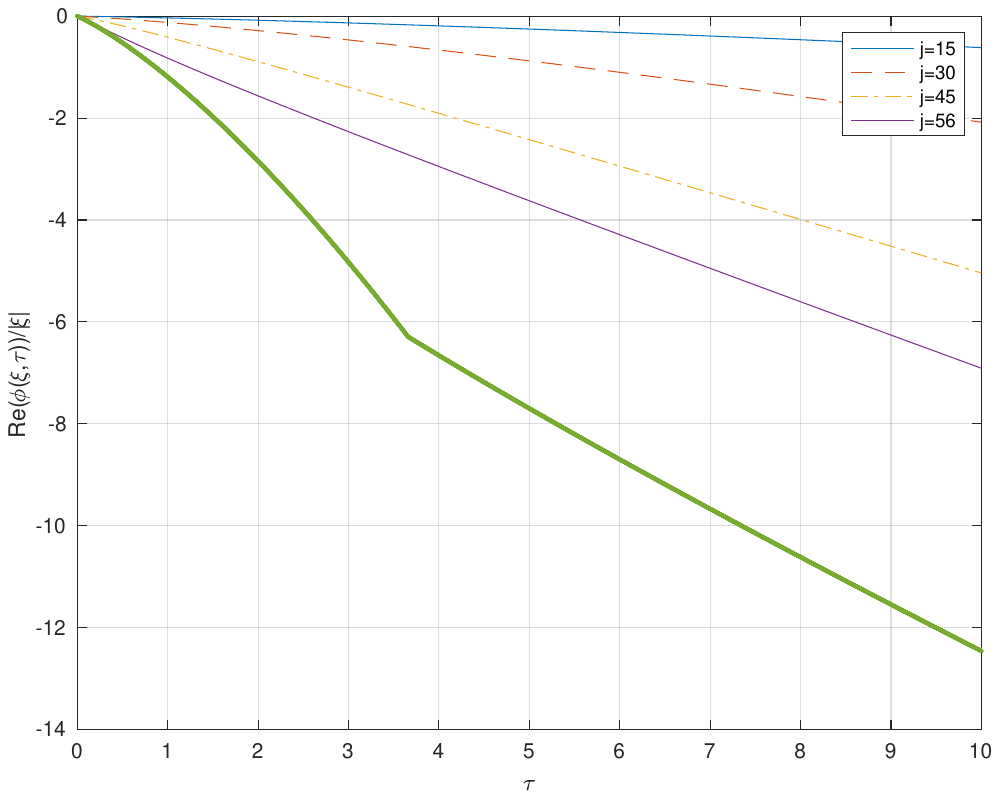}
    \caption{} \label{CurvePhiSet1T10}
\end{subfigure}

\end{tabular}
\caption{Panel (A). Dots: $h=\ka_\infty=\sqrt{1-\rho^2}/(\ga\nu)$, other lines: $h=h(\xi_j,t)/|\xi_j|$, 
for $\xi_{15}=1.1299 - 0.5 i, \xi_{30}=4.3712 - 0.5 i, \xi_{45} =15.5301 - 0.5 i, \xi_{56} =39.1661 - 0.5 i $.
Panel (B). Dots: $t\mapsto -\max\{G_1(t) \cos\om y,G_2(t) \cos(2\om) y^2\}$ (the curve defined by the RHS of
\eq{ash} is higher and gives a more accurate bound), other lines: $t\mapsto \Re\phi(\xi_j)/|\xi_j|$.
Parameters from \cite[Example 5.1]{EuchRosenbaum2019}.
}
\label{CurvePhiSet1T10Set1}
\end{figure}

\subsection{Conformal bootstrap principle}\label{ss: conformal bootstrap}
The conformal deformation method (sinh-acceleration in particular) allows one to
accurately assess the total error of the method comparing the two prices $V^j$, $j=1,2$, given by \eq{EuroPricesinhtrap} with two 
different contour deformations and different $N, \ze$. If the contours are not close, the terms in one sum are evaluated at points  on one curve that are 
 far from the points on the other curve. Hence, if the number of terms is several dozen or more and the difference $V^1-V^2$ is of the order of $10^{-m}$, where $m=3,4,\ldots, $ then the probability that the difference of the exact (unknown) price $V$
 from either of $V^1, V^2$ is greater than $10^{-m+2}$ is essentially 0. We used
 this observation in \cite{SINHregular,Contrarian,EfficientLevyExtremum,EfficientDoubleBarrier2}.
In the case of the rough Heston model, the existence of a cone of analyticity is unknown. Although certain results for  a strip of analyticity are available \cite{GerholdGersteneckerPinter2019}, the rate of decay of the characteristic function
at infinity is unknown. Hence, the scheme applied to the Heston model and affine SV models \cite{paraHeston,pitfalls,SINHregular} cannot be justified rigorously, even if Flat iFT method is applied.
In other complicated SV models, the strip (or tube domain in the multi-factor models) and cone of analyticity 
where the characteristic function decays is also unknown, and difficult to find. To resolve this difficulty,
we suggest the following principle.

\vskip0.1cm
\noindent
{\sc Conformal bootstrap principle I.} {\em Assume that a union $\cU$ of a strip and cone of analyticity
of the characteristic function $\Phi(\xi)$ is known, and  $\Phi(\xi)$ can be calculated with an (almost) machine precision.
Construct  at least two admissible conformal deformations of
%$\chi_j$, $j=1,\ldots, n$, of
the line of integration $\cL^0$ such that  the contours $\cL_j=\chi_j(\cL^0), j=1\ldots, n,$ are not close
and diverge at infinity, and calculate the approximations $V^j$ to the price  
using   
the corresponding changes of variables and simplified trapezoid rule
with several dozens of terms and more.

If  $|V^j-V^k|<10^{-m}$ for $j,k\in 1,\ldots, n$,
where $m$ is not too small, e.g., $m\ge 4$, 
then each of $V^j$ satisfies the error tolerance $10^{-m+2}$ with probability almost 1.}

Assume that $\Phi(\xi)$ is evaluated numerically, and the numerical procedure defines the approximation $\Phi_{ap}(\xi)$
as an analytic function. If the error of the approximation is unknown, then Conformal bootstrap principle I does not notice the error
of the approximation   $\Phi(\xi)\approx \Phi_{ap}(\xi)$. Furthermore, it is possible that a region of analyticity $\cU$ 
and the rate of decay of $\Phi(\xi)$ as $\xi\to \infty$ remaining in $\cU$ are unknown as well (this is the case for the rough Heston model).
Then we use
\vskip0.1cm
\noindent
{\sc Conformal bootstrap principle II.}  {\em Assume that we have two or more numerical procedures for
evaluation of $\Phi(\xi)$ for $\xi$ in a union $\cU$ of a strip and cone. Let $\Phi_{ap,j}(\xi), j=1,2, $ be the approximations.
At least one of the functions $\Phi_{ap,j}$ may not be an analytic function.

Then, if we use different  $\Phi_{ap,j}$ to evaluate the integrals over different contours $\cL_j$, and after   
the corresponding changes of variables and application of the simplified trapezoid rule
with several dozen of terms and more, 
the results agree with the accuracy $10^{-m}$ where $m$ is not small, e.g., $m\ge 7$, 
then 
\begin{enumerate}[(1)]
\item
$\Phi$ is analytic in a simply connected region $\cU_0\subset \cU$ containing the chosen contours;
\item
each of the results satisfies the error tolerance $10^{-m+2}$ with probability almost 1.
\end{enumerate}
}
\begin{rem}\label{rem:analyticAdams}{\rm 
The approximations produced by the fractional Adams method and Modification II
are analytic functions, and if we do not apply Modification III, we obtain the agreement between the results
of the order of E-14-E-13 even in situations when application of Modification III shows errors of the order of E-08.
 Errors of the sinh-acceleration method in the examples in Sect.~\ref{s:numer} are, essentially, the errors of the modifications
of the Adams method.}
\end{rem} 

\section{Numerical examples}\label{s:numer}
  For an independent verification
of the accuracy of our method, the reader may consult Tables \ref{table:T=2}-\ref{table:T=1/252},
where our results are compared with the results in \cite{Pages2007}.  In the tables, we show the errors and relative errors w.r.t. benchmark prices calculated using much finer and longer grids; the prices shown are the prices calculated using the rough and short grids described in the table notes. 
The calculations in the paper were performed in MATLAB 2017b-academic use, on
a MacPro Chip Apple M1 Max Pro chip (3.2-GHz processor) with 10-core CPU, 24-core GPU, 16-core
Neural Engine 32GB unified memory, 1TB SSD storage.

\subsection{Examples in the case of a large error tolerance}\label{ss:large_err_tolerance}
Consider the example with the parameters \eq{parEuRos}.
We apply the sinh-acceleration with the parameters chosen for a given error tolerance $\eps$ as prescribed in Sect.~\ref{ss: SINH} and Modification III of the Adams method  with two iterations and $M$ chosen by hand. We consider options of maturities $T=0.5$ and $T=1/12$ (half a year and one month), the strikes are in the interval $[0.8, 1.2]$, the spot is 1.
In these examples and in many other numerical experiments,  the prescriptions in Sect.~\ref{ss: SINH}
lead to a certain overkill: to satisfy the desired error tolerance, the step size can be larger, and the number of terms smaller than recommended. 
%However, the prescription is for the absolute error, and close to maturities and far in the tails, the ATM and %OTM prices are small, and, therefore, if we wish to obtain the results with small relative errors, we need to % increase the error tolerance for the absolute error. 
For $T=0.5$ (see Table \ref{table:rough T=0.5}), we use
the prescription for the error tolerance $\eps=0.01$ with a marginally smaller $N=7$, and $M=9$.
The relative errors of prices (resp., implied volatilities) are smaller than 0.008, and implied volatilities are smaller than 0.004; the OTM put and call prices are calculated separately, and the total CPU time is 0.84 msec. We show the prices for 9 strikes but the most time consuming block of the program (more than 95\% of the CPU time), namely,
the evaluation of the characteristic exponent $\phi(\xi_j,t_k)$ for $\xi_j$ on the grids used in the sinh-acceleration method (one grid for OTM and ATM puts, the other one for OTM calls) and $t_j=jT/M$, $j=1,2,\ldots, M$, can be used to calculate vanilla  prices at thousands of points $(K,T)\in [0.8, 1.2]\times [0.25,0.5]$  in 1-2 msec. For $t_j$ not on the grid, interpolation can be efficiently used because $\phi(\xi,t)$ is of the class $C^\infty(0,T_*(\xi))$ in $t$ on any interval $(0,T_*(\xi))$ where $\phi(\xi,t)$ exists. 
Note that the relative errors of implied volatilities are smaller than the ones of prices, which explains why
it is preferred to calibrate to implied volatilities rather than to prices: the goodness of fit improves. For option of a shorter maturity $T=1/12$ and deep OTM options, the relative errors of prices are approximately 15 times larger than the ones of implied volatilities (see Table \ref{table:rough T=1/12 refined}).

In the case $T=1/12$, we produce two tables. The accuracy of the same order of magnitude as
in the case $T=0.5$ can be achieved only if we use smaller $\eps=0.001$ to choose the parameters of the sinh-acceleration, hence, smaller $\ze$ and larger $N=12$. Since the errors of the fractional Adams method increase as $T$ decreases, we need to use a larger $M=20$, and the CPU time increases almost two-fold.
However, the relative errors of the order of 0.5\% are only for strikes close to the spot (see Table \ref{table:rough T=1/12}).
To increase the accuracy in the tails, it is necessary to make the contours of integration curved so that
the oscillating factor becomes fast decreasing one and decrease $\eps=10^{-5}$; $\ze$ decrease and $N$ increase further. Since the derivatives of  $\phi(\xi,t)$ w.r.t. $t$ increase as $\xi$ moves to infinity along non-horizontal rays, we need to use $M=40$ (see Table \ref{table:rough T=1/12 refined}); the total CPU time is 2.7 msec. In Table \ref{table:rough T=0.5, Lewis and iFT-BM}, we demonstrate the performance of the Lewis and Flat iFT-BM methods for $T=0.5$, and in Table \ref{table:T=5}, the performance of the sinh-acceleration,
Lewis and Flat iFT-BM methods for options of moderately large maturity $T=5$. Note that if the Lewis method is used, it is necessary to use  $M$ in the modified Adams method significantly larger than in the other two cases because among the nodes $u_j$ in the Lewis method there are nodes too close to 1, hence,
the log-characteristic function needs to be evaluated for $\xi$ too large in the absolute value. If the sinh-acceleration and Flat iFT-BM are used and the error tolerance is large, then only $\xi$ of a moderate size appear, and modifications of the Adams method with small $M$'s can satisfy the desired error tolerance.

In Table \ref{table: COS1_K}, we show similar results for a different set of parameters of the rough Heston model, and larger $T=1$. We also show that the relative errors of prices calculated using COS method are significantly larger than the errors of prices produces by the method of the paper at the CPU cost thousands of times smaller. 

\subsection{Relative errors of several methods for options of moderate and small maturities}
In Tables \ref{table:rel_errors_moderate}, \ref{table:rel_errors_short}, \ref{table:implvol_short},
we show the relative errors of prices of OTM and ATM vanillas and implied volatilities and the CPU time, the sinh-acceleration, hybrid, Flat iFT, Flat iFT-BM and Lewis methods being used. As we explained and demonstrated above,  CM and COS methods are less accurate, more time consuming and more complicated than 
Flat iFT and Flat iFT-BM respectively. Since we wish to compare the other methods with the hybrid one,
and we use the prices provided in \cite{RoughNotTough}, we choose the parameters of each scheme so that
the errors are of approximately the same order of magnitude as in \cite{RoughNotTough} or somewhat smaller,
and compare the CPU times.
The reader observes that for options of moderate maturities, the accuracy
E-06 can be achieved fast with sinh-acceleration, Flat iFT-BM and Lewis methods using comparable number of terms and CPU time, the sinh-acceleration being the fastest. The accuracy of the order of E-09 is very difficult if possible to achieve
unless the sinh-acceleration is used. For short maturity options, OTM ones especially, reliable calculations
are possible only using the sinh-acceleration, and even for the latter, the region around the strike where even marginally
reliable calculations are possible shrinks as the time to maturity approaches 1 day. % and becomes rather small.

\subsection{Comparison with the hybrid method}\label{s:hybrid}
In \cite{RoughNotTough}, a hybrid method based on the asymptotic expansion of the solution of the fractional
Riccati equation near 0 and the Richardson-Romberg extrapolation \cite{Pages2007} farther from 0 is derived. It is demonstrated that
the Adams method may require extremely fine grids to achieve even moderate accuracy; the CPU cost becomes prohibitively large. 
%This fact is to be expected since any Fourier inversion procedure is inaccurate unless
% sufficiently large (in absolute value) nodes $\xi$ are used, and for options of short maturities, among the nodes $\xi$ % used, there must be the ones with fairly large $|\xi|$. 
The modifications of the Adams method that we use are faster and more accurate
than the procedure in \cite{RoughNotTough}. The CPU time documented in \cite{RoughNotTough}
(implemented in ``C++ using a standard laptop with a 3.4-GHz processor") is several hundred msec for
each pair (maturity, strike). The CPU times in the present paper, for the same accuracy as in \cite{RoughNotTough},
are in the range 100-250 msec, for the whole volatility surface from 1 month to 2 years, even if the option prices are evaluated at several hundred thousand of points in the time-strike space. For options of shorter maturities,
much finer $t$-grids and longer $\xi$-grids, different for OTM calls and puts are needed; the CPU time is  1-2 sec. 
 However, for options of moderate maturities
in a narrower interval, e.g. from $T=0.5$ to $T=2$, less than 100 msec. suffice to satisfy the error tolerance of the 
order of E-06. For options of moderate maturities, the call options prices
published in \cite{RoughNotTough} are in  almost perfect agreement with the results produced by the method of the present paper (for the digits presented in \cite{RoughNotTough};
%\ref{table:T=1/252},\ref{table:rel_errors_moderate} and
see Tables \ref{table:T=2}-\ref{table:rel_errors_short}. We produce the results with the accuracy E-10 and better). 
The agreement decreases as time $T$ to maturity decreases, and, contrary to the expectation that
a method based on the asymptotic expansion of the solution near 0 should perform better, we claim that the errors
of the hybrid method of \cite{RoughNotTough} are larger, for OTM options especially. The explanation is two-fold:
first, in \cite{RoughNotTough}, an inherently inaccurate CM method is used, and, secondly,
to accurate price short maturity options, the fractional Riccati equation must be accurately solved for
large values of the spectral parameter $\xi$. But the asymptotic expansion in \cite{RoughNotTough}
is valid in the region that shrinks as $|\xi|\to \infty$. How to efficiently resolve the second issue remains an important
open problem; the inherent errors of CM method are analyzed in Sect.~\ref{s:FT}. We expect that if an accurate procedure
for the numerical Fourier inversion is used instead of the CM method, then,
for options of moderately small maturities and not far from the spot so that only moderately large $\xi$ appear in
the pricing formula, the asymptotic method
\cite{RoughNotTough} becomes more accurate (but still more time consuming) than the method of the present paper.

%\subsection{Effect of the interpolation used in the CM method}\label{ss:interp_effect}

\section{Application of sinh-acceleration and Conformal Bootstrap principle to calibration}
\label{s:calibration}
The examples and analysis of various numerical methods in the paper and in  \cite{paraHeston,one-sidedCDS,HestonCalibMarcoMeRisk}  indicate that if 
the parameters of the numerical pricing scheme are fixed then a calibration procedure using the scheme
can find ``a good fit" only in a rather narrow region of the parameter space, where the scheme is sufficiently accurate.
Other parts of the parameter space are filtered out, and
 with a sizable probability,
``a good fit" will be found close to the boundary of the region, where the ``true calibration error" of the model
and the error of the pricing scheme almost cancel out (sundial calibration and ghost calibration). 
Thus, to increase the calibration quality, it is necessary to increase the region in the parameter space, where the pricing scheme produces sufficiently accurate results,
and filter out possible errors at the boundary of the region. The sinh-acceleration and Conformal Bootstrap principle
can be used to achieve both goals as follows.

{\sc Pre-calibration Step I.} A subset $\Theta$ of the parameter space of the model is selected, where ``a good fit" is expected;
a universal reasonably fast pricing algorithm cannot be even moderately accurate for all points in the parameter space. If it is necessary to consider
a very wide subset $\Theta$, then it is necessary to divide $\Theta$ into several subregions, and use the scheme below for each subset. 

Let $Obs$ be the set of pairs $(K,T)$ 
in the data set.
An accurate and not unnecessary time consuming pricing
is possible only if the $(K,T)$-plane is divided into several subregions as well, and the parameters of the numerical scheme are chosen
for each subregion  separately. For illustration, assume that the options of maturities from 1 day to 5 years
are chosen as the inputs for the calibration, and the spot price is normalized to 1. 
Then we consider time intervals $[T_j, T_{j+1}]$, $j=1,2,3$, where 
$T_1:=1/252<T_2:=1/12<T_3:=0.5<T_4=5$. 
  For $j=1,2,3$, we choose $k_j$ so that, for all $(K,T)\in Obs$ such that $T\in [T_j,T_{j+1}]$
and $\ln K\in (-k_j,k_j)$, the OTM option prices are larger than $10^{-6}$; if $|\ln K|>k_j,$ the OTM option prices
are smaller than $10^{-5}$. If such $k_j$ do not exist, the number of time intervals needs to be increased.
We choose $k^+_j$ so that, if $\pm \ln K>k^+_j$, the OTM option price is very small, e.g., smaller than $10^{-8}$; these prices are deleted from the data. Set $U^0_j=\{(K,T)\ |\ |\ln K|\le k_j\}$, $U^\pm_j=\{(K,T)\ |\  k_j\le \pm \ln K\le k^+_j\}$, and  choose
2 sets of the parameters of the sinh-deformation  for each $U\in\{U^0_j, U^\pm_j, j=1,2,3\}$. Note that the choice of $k_j$ is dictated by the properties of the numerical scheme that we use. In our numerical experiments, several set of the parameters of the model worked sufficiently well and fast when the prices for $(K,T)\in U^0_j$ (where the prices are not too small) were calculated,
and different sets were necessary to use in the regions where prices were smaller than $10^{-5}$, to satisfy a small tolerance for the relative error.

\begin{rem}{\rm As our numerical examples indicate, for options of moderate and large maturities, Flat iFT-BM or Flat iFT-NIG or Flat iFT-Heston can be used as well, with the summation by parts if $U$ is separated from
the ray $\{S_0\}\times (0,+\infty)$. The CPU time remains small, and if a strip of analyticity of the characteristic function $\Phi(\xi,T)$ where $\Phi(\xi,T)$ decays at infinity is known, we can use two lines in the strip of analyticity instead of two sinh-deformed curves. The reliability of Conformal Bootstrap principle is weaker in this case, though.}
\end{rem}
{\sc Pre-calibration Step II.}  Admissible choices depend on the analytical properties of
the characteristic function $\Phi(\xi,T)$. The first crucial precalculation step is the choice of admissible strips and cones. The strip $S_{(-1,0)}$ can be used always but for accurate pricing OTM options close to maturity, 
strips of the form $S_{(\lm,-1)}$ and $S_{(0,\lp)}$, where $\lm<-1, 0<\lp$, are highly advantageous to use.
To choose the cone of analyticity $\cC_{\gam,\gap}$, the following factors must be taken into account. 
If large (in absolute value) $\gam,\gap$ are admissible then the grid in the $\xi$-space sufficient for an accurate Fourier inversion can be made shorter but since such a grid involves $\xi$ with the large ratio $\Im\xi/\Re\xi$,
an accurate numerical solution of the fractional Volterra equation requires very large $M$ lest the numerical solution blows up. Since the latter solution is much more time consuming than the numerical Fourier inversion, it is highly advisable to use small $\om$  in the sinh-acceleration procedure, (we used $\om=\pm 0.1, 0.2$, hence, $\ga^\pm=\pm 0.2, 0.4$) and try to increase $\om$ and $\pm\ga^\pm$ only when necessary,
for OTM options of short maturities. We cannot exclude the cases when the opening angle of the cone of analyticity is smaller, hence, the numerical method with $\om=\pm 0.2$ does not work (although we did not encounter such a case in our numerical experiments), and smaller $\om$ must be used. 
This step can be done using Conformal Bootstrap principle with moderately large error tolerance, the reason being
that if one of the contours is outside the domain of analyticity or very close to the boundary, the numerical method will not work at all or the difference between the two results will be very large.

{\sc Pre-calibration Step III.}  We fix two pairs $(\om^\ell, d^\ell), \ell=1,2,$ 
satisfying conditions in Sect. \ref{ss: SINH}, and calculate $\om_1^\ell, b_\ell$. Given the error tolerance $\eps$, we calculate  $\ze_\ell$ using the ad-hoc approximation in Remark \ref{rem:choice of ze}. 

 Next, for each $U$, we fix the set $\cK\cT(U)$ of 9 points $(K,T)$: 4 corners of $U$, 4 in the middle of each side, and 1 in the center.
For each quadruple $(U,\om_1^\ell, b_\ell, \om^\ell)$, we design 
\begin{enumerate}[1)]
\item
a map $\cN(U,\om_1^\ell, b_\ell, \om^\ell; \cdot): \Theta \ni \theta \mapsto (0,+\infty)$ and 
\item
a map $\cM(U,\om_1^\ell, b_\ell, \om^\ell; \cdot): \Theta \ni \theta \mapsto (0,+\infty)$
\end{enumerate}
such that, for a large random sample $\Theta_0\subset \Theta$, and 9 points $(K,T)\subset \cK\cT(U)$,
the differences between the OTM option prices $V^\ell, \ell=1,2,$ calculated for each pair  $(\theta, (K,T))
\subset \Theta_0\times \cK\cT(U)$ using the method of the paper with the parameters
$\om_1^\ell, b_\ell, \om^\ell$, $N=\mathrm{ceil}\,\cN(U,\om_1^\ell, b_\ell, \om^\ell; \theta)$ and 
$M=\mathrm{ceil}\,\cM(U,\om_1^\ell, b_\ell, \om^\ell;\theta)$, do not exceed the error tolerance, in absolute value.
The maps $\cN$ and $\cM$ can be constructed either precalculating the values at points of a moderately fine multi-grid in $\Theta$ or designing  an appropriate deep neural network. We believe that it is unnecessary to choose $\cK\cT(U)$ using a randomization procedure because the dependence of sufficiently good $M$ and $N$ on $(K,T)\in \cK\cT(U)$ is fairly regular.

{\sc Calibration.}
After the maps $\cN$ and $\cM$ are constructed, the sinh-acceleration method gives a reliable pricing map from $\Theta$ to
the set of OTM option prices. The pricing map can be used either in standard search procedures (the prices for different $U$'s can be calculated in parallel) or, as in
\cite{HorvathMuguruzaTomas2021,Romer2022},  to train the network to calculate the implied volatilities at the points of a chosen grid in $(K,T)$ plane, and then use  interpolation to calculate the implied
volatilities for pairs $(K,T)$ in the data set. However, interpolation introduces sizable (sometimes, large) errors,
close to maturity especially. In particular,   spurious wings of the volatility curves may appear.

Hence, we suggest to train the network to find $N$ and $\phi(\xi,t_\ell)$ for each subset $U=U^\pm_j, U^0_j,$ and each $\xi$ on the two grids in the dual space, and $t_\ell$ on a sufficiently fine grid (grids depend on $U$). After that, for each pair $(K,T)\in Obs\cap U$,
two prices $V(\om^n_1,b^n,\om^n; \ze, N^n;K,T)$, $n=1,2,$ can be calculated faster than in the standard procedure
for the Heston model.

\section{Conclusion}\label{s:concl}
In the paper, we analyzed two crucial components for pricing vanilla options in the rough Heston model using the Fourier transform technique: the numerical evaluation of the characteristic function $\Phi(\xi,T)$ using the fractional Adams method, and
popular numerical Fourier inversion methods. We 
suggested several improvements of the fractional Adams method, which significantly increase the accuracy of calculations.
We showed that if the sinh-acceleration is applied to the Fourier inversion,  the vanilla prices at hundreds of thousands of points can be evaluated in a fraction of a second, for a wide region of the strike-time to maturity space and error tolerance smaller than E-06;
other popular methods are less accurate and slower. 
We demonstrated that, for short maturity options,
the grids necessary for accurate calculations must be significantly larger than for options of moderately large maturities, %and the CPU time increases as well;
and popular methods become very inaccurate.
The common belief
that the same set of parameters of the numerical scheme can be used for all pairs $(K,T)$ of interest
(as in CM and COS methods) leads to incorrect results, in calibration procedures especially. Moderately accurate calculations become possible only in a region of moderate and moderately small times to maturity and not too far from the spot.
This effect was demonstrated earlier in the context of the calibration of the Heston model \cite{paraHeston,HestonCalibMarcoMeRisk}, where $\Phi(\xi,T)$ is known explicitly.
If $\Phi(\xi,T)$ can be evaluated only numerically, the total errors can snowball even in the case of the standard diffusion SV models
\cite{pitfalls}. In the case of the rough Heston model, the errors of the evaluation of $\Phi(\xi,T)$ are larger,
and, in the result, the calibration results in \cite{EuchRosenbaum2019} are seriously incorrect: the correct ATM skew
for the calibrated parameters is several times lower than the one shown in \cite{EuchRosenbaum2019}, and the 
volatility curves are different as well. We produced correct volatility  curves for an example in  \cite{KamuranEmreErkan2020}  where  COS method is used; the relative errors of COS-prices are in the range 5\%-22\%. We explained that CM method can produce spurious volatility smiles, and changing the dampening factor, one can obtain smiles of different shapes and choose the smile one likes. Both COS and CM methods are more complicated, introduce
additional errors and slower than the methods that we use.
The ``advantages" of CM method: 
referring to the universal prescription of the CM method,
one can pretend that no analysis of the properties of the integrand in the pricing formula are
necessary;
one can produce nice volatility smiles of different shapes playing with the parameters
of the method even if the correct smile is, in fact, an almost straight slope\footnote{a modified adage: if you torture
a numerical method long enough, it will confess to anything.}; 
important features of models with jumps can be
(artificially) reproduced in a diffusion model. A marginal gain of COS method is an increase of 
the strip of analyticity of the integrand, but the same gain can be achieved much simpler (Flat iFT-BM and Flat iFT-NIG
methods), without introducing additional errors, at a much smaller CPU cost.
We calculated the implied volatility curves for the parameter sets in two papers which use the Lewis method,
and demonstrated that even in a very favorable case of maturities in the range 0.444-2 years,
the application of the Lewis method results in incorrect smiles; one expects that, in the same papers, calibration results for short maturity options would be significantly worse. 
In view of these observations and the fact that in the majority of publications the details are 
lacking and only the name of a method (CM, COS or Lewis) is given,
 we believe that a majority of empirical calibration results are not reliable. In any case,  without
the explicit description of the numerical method and its parameters, one can doubt the veracity of the comparison of the performance
of various models. Certainly, in many cases that we considered, the curves and surfaces calculated using popular methods are
incorrect.

Accurate analysis of errors of each of popular methods is possible only if the key analytical properties of 
$\Phi(\xi,T)$ are known. In the case of the rough Heston model, for the first step, namely, the solution of the fractional Riccati equation, the necessary theoretical results are lacking. Therefore, any numerical procedure
is a conditional one:  essentially, one presumes that $\Phi(\xi,T)$ is sufficiently nice so that the integral converges and is sufficiently regular so that the numerical method of choice is accurate. To overcome this difficulty,
we formulate and use Conformal Bootstrap principle, and explain how this principle can be applied to construct 
reliable pricing and calibration procedures. The principle and procedures can be applied to other models where 
the necessary properties of $\Phi(\xi,T)$ are unknown or are difficult to theoretically derive.

%\bibliography{thebibliographyFM24}{}
%\bibliographystyle{plain}

\appendix

\section{}\label{s:tech}
\subsection{Grids depending on $\xi$ }\label{ss:xi0dependentgrids} The accuracy of calculations can be increased
using grids depending on $\xi$. To understand what a proper dependence of the grid on a (large in absolute value)
$\xi$ is, we take  $\xi=re^{i\varphi}$, where $r>>1$ and $\varphi\in (-\pi/4,\pi/4)$. Set $t_1=t r^{1/\al}$, $h_1(r,\varphi,t_1)=r^{-1}h(re^{i\varphi}, t_1 r^{-1/\al})$,
 substitute $t=t_1 r^{-1/\al}$ and $h(\xi,t)=rh_1(r,\varphi,t_1)$ into \eq{Volterra} and change the variable $s=r^{-1/\al}s_1$.
 We obtain the Volterra equation for $h_1$:
 \bbe\label{VolterraR}
h_1(r,\varphi,t_1)=\frac{1}{\Ga(\al)}\int_0^{t_1}(t_1-s_1)^{\al-1}F_1(r,\varphi, h_1(r,\varphi,s_1))ds_1,
\ee
where
\bbe\label{eqF1}
F_1(r,\varphi, h_1)=-\frac{1}{2}(e^{i2\varphi}+ie^{i\varphi}r^{-1})+\ga(ie^{i\varphi}\rho\nu -r^{-1})h_1+\frac{(\ga\nu)^2}{2}h_1^2.
\ee
Since $F_1(r,\varphi, h_1)$ is uniformly bounded as a function of $r$, the equation \eq{VolterraR} can be integrated
accurately if $t_1$ is not too large; if $t_1$ is large, the interpolation errors accumulate. 
Therefore, in a region $t\le A|\xi|^{-1/\al}$, where $A$ is moderately large, we solve \eq{Volterra} using a grid with the step of the order of $|\xi|^{-1/\al}$, and in the region $t\in [A|\xi|^{-1/\al}, T]$, we use
a grid independent of $\xi$.
\footnote{We construct grids that are unions of two uniform grids for simplicity. One
can use more complicated grids. The only essential requirement is that the step on $[0,A|\xi|^{-1/\al}]$ must be much finer
than the ones on $[A|\xi|^{-1/\al},T]$.}
This requires the straightforward recalculation of the coefficients in fractional Adams procedures:

For $k=0,1,\ldots, M_\xi-1$ and $j=0,1,\ldots, k$, calculate
\beqast
a_{k+1,k+1}&=&\frac{1}{\Ga(\al+2)}(t_{k+1}-t_k)^{\al}, \\
a_{0,k+1}&=&\frac{1}{\Ga(\al+2)}((\al+1)t_k^\al+\frac{1}{t_1}\left((t_{k+1}-t_1)^{\al+1}-t_{k+1}^{\al+1}\right),\eqast
and, in the cycle $j=1,1,\ldots, k$, calculate
\beqast
a_{j,k+1}
&=&\frac{1}{\Ga(\al+2)}\left\{\frac{(t_{k+1}-t_{j-1})^{\al+1}}{t_j-t_{j-1}}+
\frac{(t_{k+1}-t_{j+1})^{\al+1}}{t_{j+1}-t_{j}}\right.
\\
&&-(t_{k+1}-t_j)^{\al+1}\left.
\left[\frac{1}{t_{j+1}-t_j}+\frac{1}{t_{j}-t_{j-1}}\right]\right\}
\eqast

\subsection{Reasons for the popularity of CM method}\label{DistortionCM} 
First, a choice of $\om_1<-1$ means that the call option prices are calculated. The call price curve being convex,
the interpolation increases the call prices, which become positive in the deep OTM region even where
the calculated numerically oscillating price of the OTM options is negative. If the put-call duality is used
to price put options, the prices of deep OTM puts also increase and become positive. Hence, it becomes easier 
to satisfy the no-arbitrage condition. If $\om_1\in (-1,0)$, hence, the covered call is evaluated, then there is no convexity,
and, in the deep OTM region where the prices are very small, the prices calculated using Flat iFT can be negative.
The second ``useful" effect of the inaccurate calculations using the CM method is as follows. The resulting
implied volatility curve does look like a proper smile even when the correct curve is  (a segment of) an essentially straight line.
Numerical examples that we produce show that with the CM choices of the line of integration $\om_1=-1.1$ and $\om_1=-1.5$, Flat iFT  can produce
a smile instead of an approximately straight slope as well; the interpolation increases the effect.
The third effect stems from the wide-spread belief that the universal prescription in \cite{carr-madan-FFT} can be 
applied for pricing in any model and wide regions in the $(K,T)$ space, without a proper analysis of the strip of analyticity and rate of decay of $\Phi(\xi,T)$. We demonstrate that accurate calculations are possible only with  appropriate choices
of the parameters of the scheme for several regions in the $(K,T)$-plane. In the
result, the CM method
produces, typically, errors of the order of E-07 or worse, hence, is the CM method is not applicable for pairs $(K/S_0,T)$ 
\begin{enumerate}[1)]
\item
with small $T$, when $\Phi(\xi,T)$ decays very slowly, and the truncation error is large;
\item
far from the tails, where the OTM option prices are small, and the integrand highly oscillates;
\item close to maturity, where
even marginally accurate calculations are possible only in a very small vicinity of the spot;
\item
 for $T$ from a moderately long intervals, if the strip of analyticity shrinks as $T$ increases.
 (This is an effect typical for the Heston model; one expect that the same effect can be observed for the rough Heston model);
 \item
as the fixed $\om_1$ is getting closer to the boundaries of the admissible interval $(\mum(T),\mup(T))$,
 the discretization error explodes.   
 \end{enumerate}

\section{Figures and tables}\label{s:figures and tables}

%\section{Figures}\label{s:figures}

\begin{figure}
\begin{tabular}{cc}
\begin{subfigure}[h]{0.45\textwidth}

 \centering
    \includegraphics[width=0.9\textwidth,height=0.4\textheight]{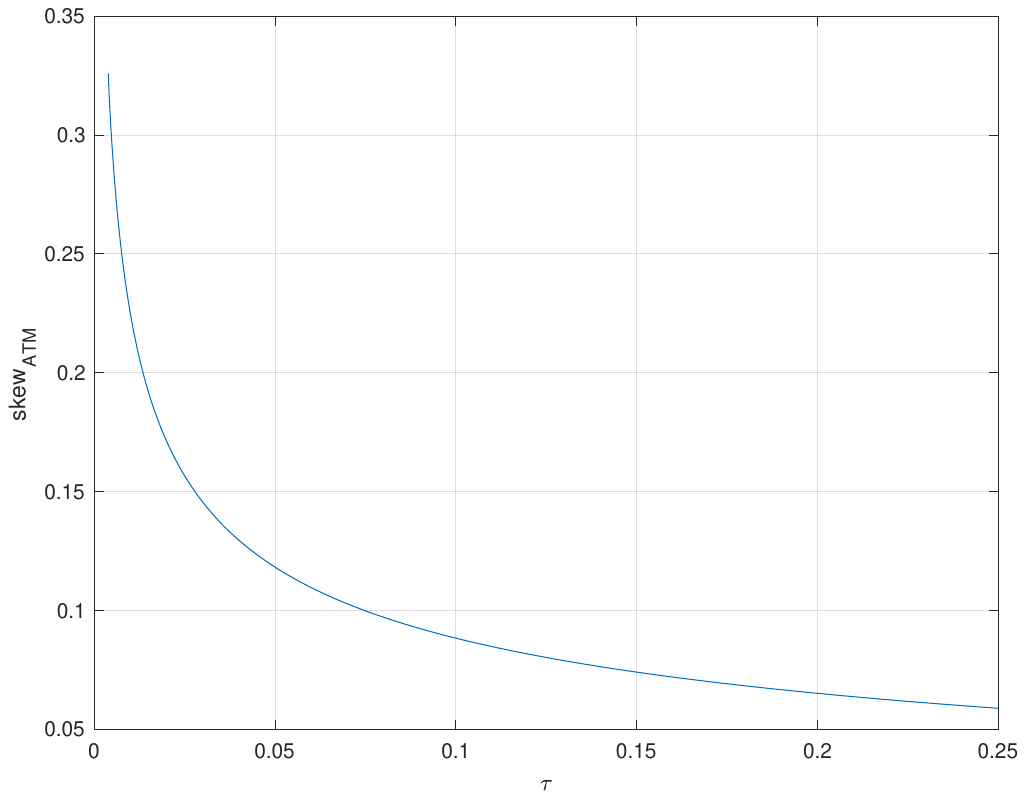}
    \caption{}\label{SkewSet1SinhXi}
\end{subfigure}
&
\begin{subfigure}[h]{0.45\textwidth}
\centering
    \includegraphics[width=0.9\textwidth,height=0.4\textheight]{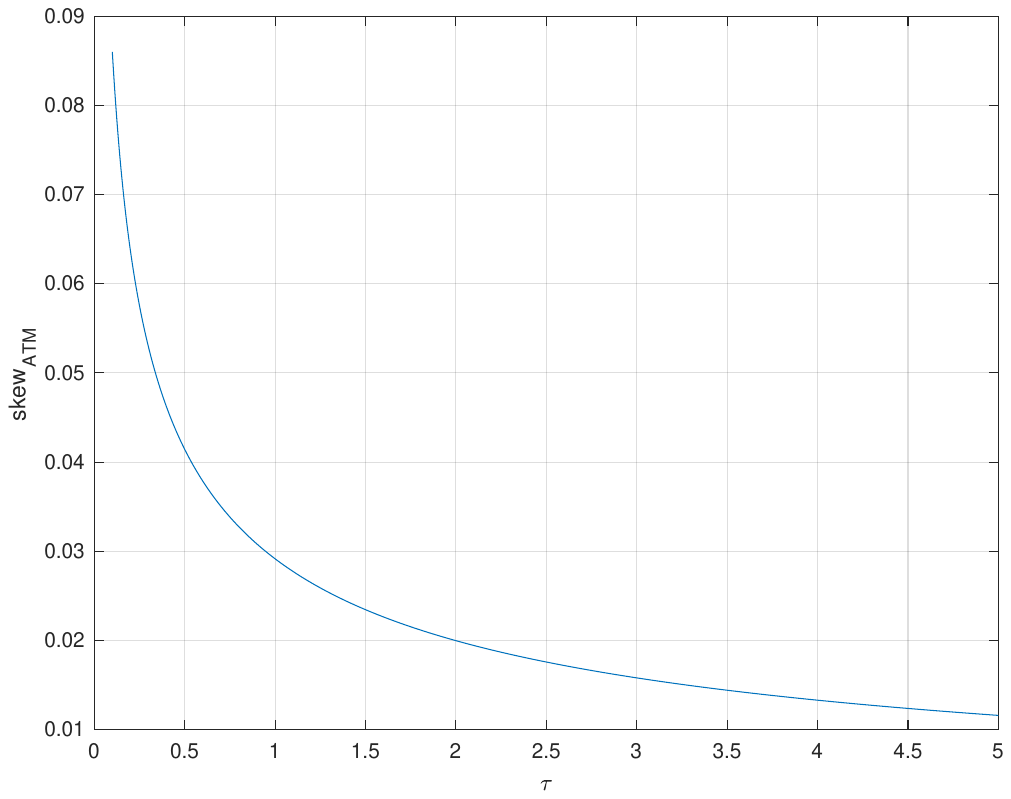}\caption{} \label{SkewSet1SinhXiLong}

\end{subfigure}

\end{tabular}
\caption{ATM  skew;
the parameters are in \eq{parEuRos}.
 } 
\label{Set1Skews}
\end{figure}

\begin{figure}
    \includegraphics[width=1\textwidth,height=0.6\textheight] {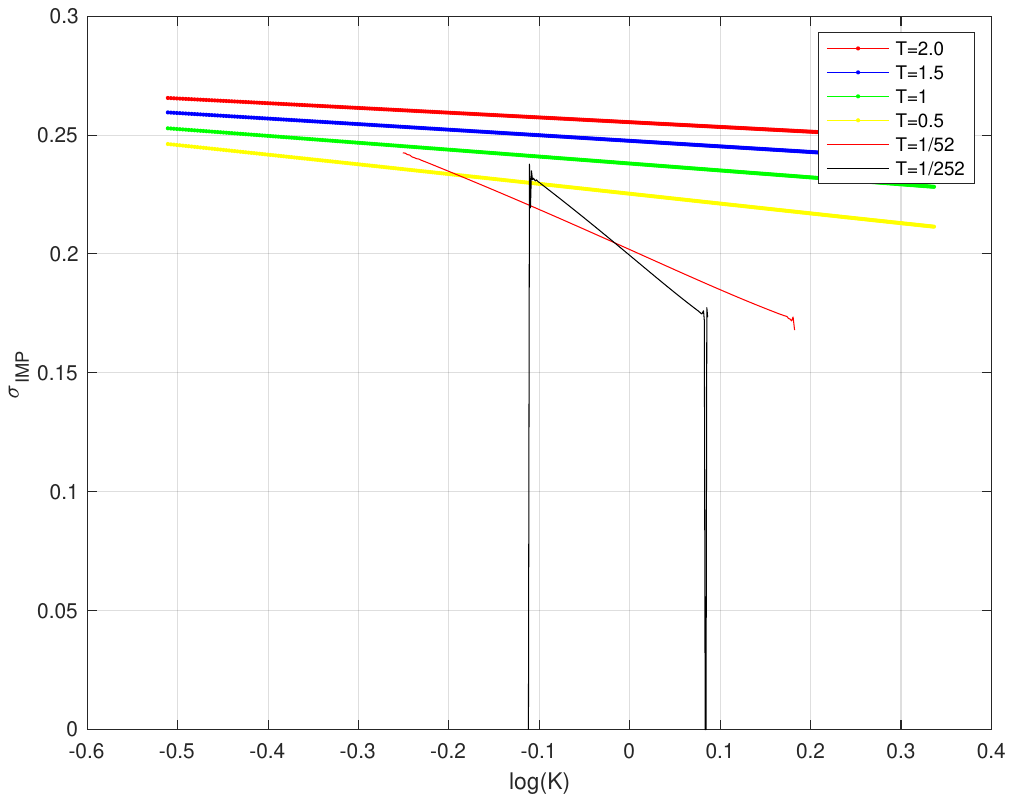}%{CM15eKA3skew.pdf}
    \caption{Implied volatility curves; the parameters are in \eq{parEuRos}.} \label{Set1Curves}
\end{figure}

\begin{figure}
    \includegraphics[width=1\textwidth,height=0.6\textheight] {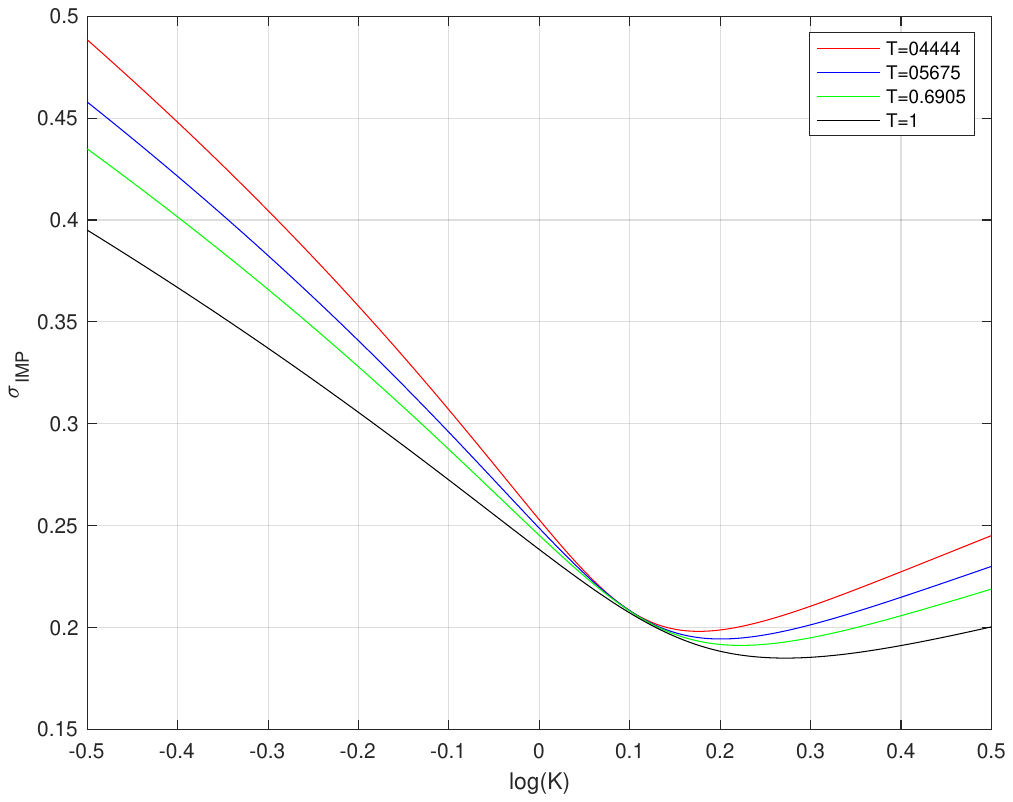}%{CM15eKA3skew.pdf}
    \caption{Implied volatility curves. The parameters  $\al=0.512,$
   $\ga=0.88,$
    $\rho=-0.7$,
    $\nu=0.96,$
    $\theta=0.016$,
    $v=0.148$, are the result of calibration to the real data in \cite[p.27]{Imperial2020}. The implied volatilities calculated using the Lewis and Adams methods and shown on Fig.~2.7 in \cite{Imperial2020} are somewhat different, in the tails especially, where the true difference between the empirical implied volatilities and the ones calculated in the rough Heston model with the calibrated parameters is significantly larger than shown on Fig.~2.7 in \cite{Imperial2020} differ from the correct curves shown above, in the tails especially. Note that on Fig.~2.7 in \cite{Imperial2020}, the range of log-strikes is asymmetric, and depends on maturity: $\ln K\in [-0.3, 0.35]$  for maturities $T=0.6905$ and $T=1$, and   $\ln K\in [-0.25, 0.35]$ for $T=0.4444$ and $T=0.5675$. A natural guess is that the results of calculations in the symmetric range $\ln K\in [-0.35, 0.35]$ are unsatisfactory. } \label{ImperialCurves}
\end{figure}

\begin{figure}
\begin{tabular}{cc}
\begin{subfigure}[h]{0.45\textwidth}

 \centering
    \includegraphics[width=0.9\textwidth,height=0.4\textheight]{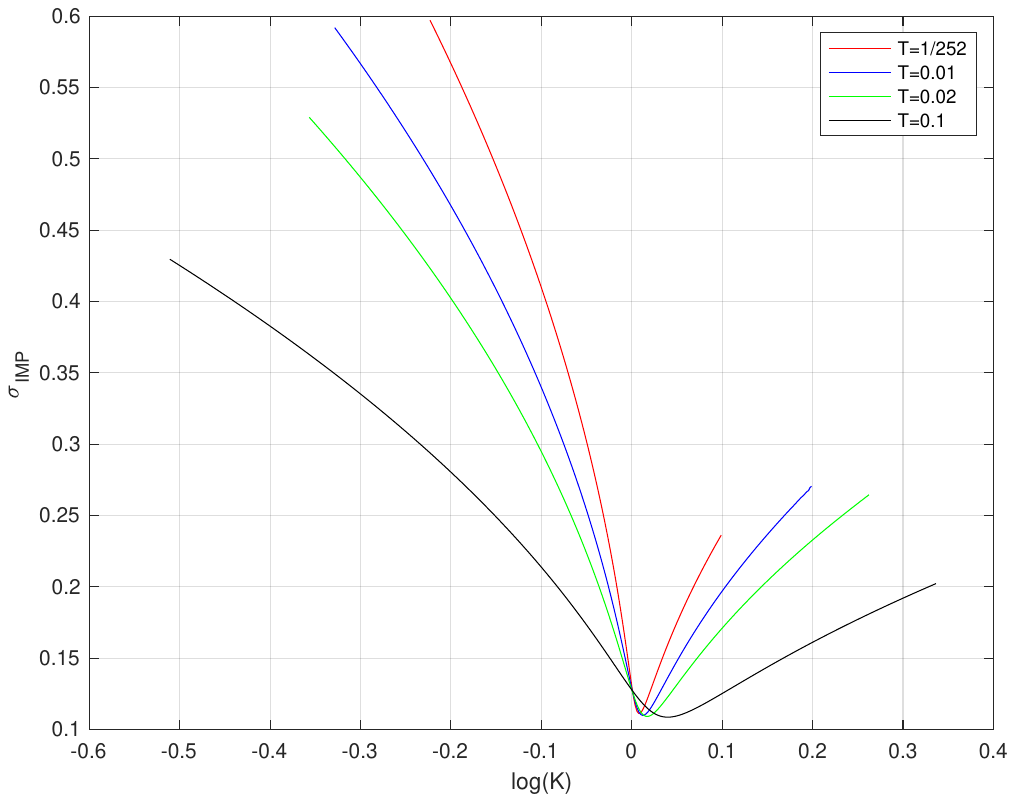}
    \caption{}\label{Set2Short}
\end{subfigure}
&
\begin{subfigure}[h]{0.45\textwidth}
\centering
    \includegraphics[width=0.9\textwidth,height=0.4\textheight]{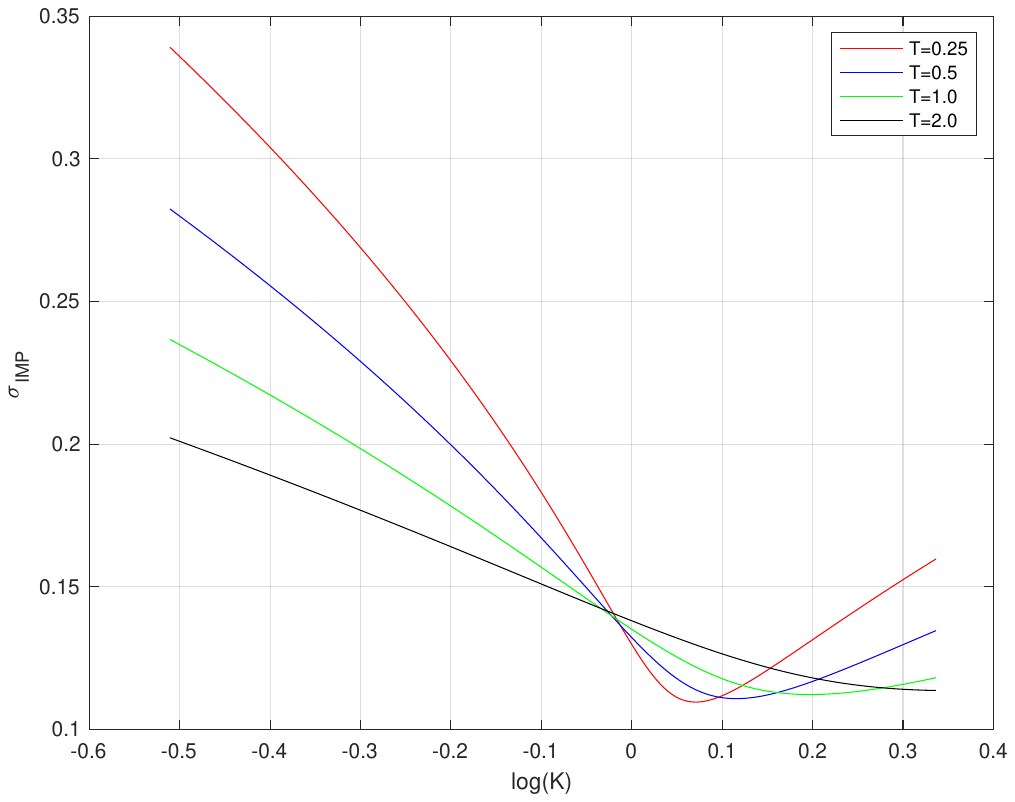}\caption{} \label{Set2CurvesLogKMod}
\end{subfigure}

\end{tabular}
\caption{Implied volatility curves in the rough Heston model  (Example in \cite[Sect. 6.2]{KamuranEmreErkan2020}); parameters $\al=0.6$, $\ga=2$, 	$\rho=-0.6$,
$\theta=	0.025$, $\nu=0.2$, $v_0=0.025$; $S_0=1$. 
%$\sg_{IMP}=0$ means that the price is outside the no-arbitrage bounds.
 } 
\label{Set2Curves}
\end{figure}

\begin{figure}
\begin{tabular}{cc}

 \begin{subfigure}[h]{0.45\textwidth}

 \centering
    \includegraphics[width=0.8\textwidth,height=0.35\textheight]{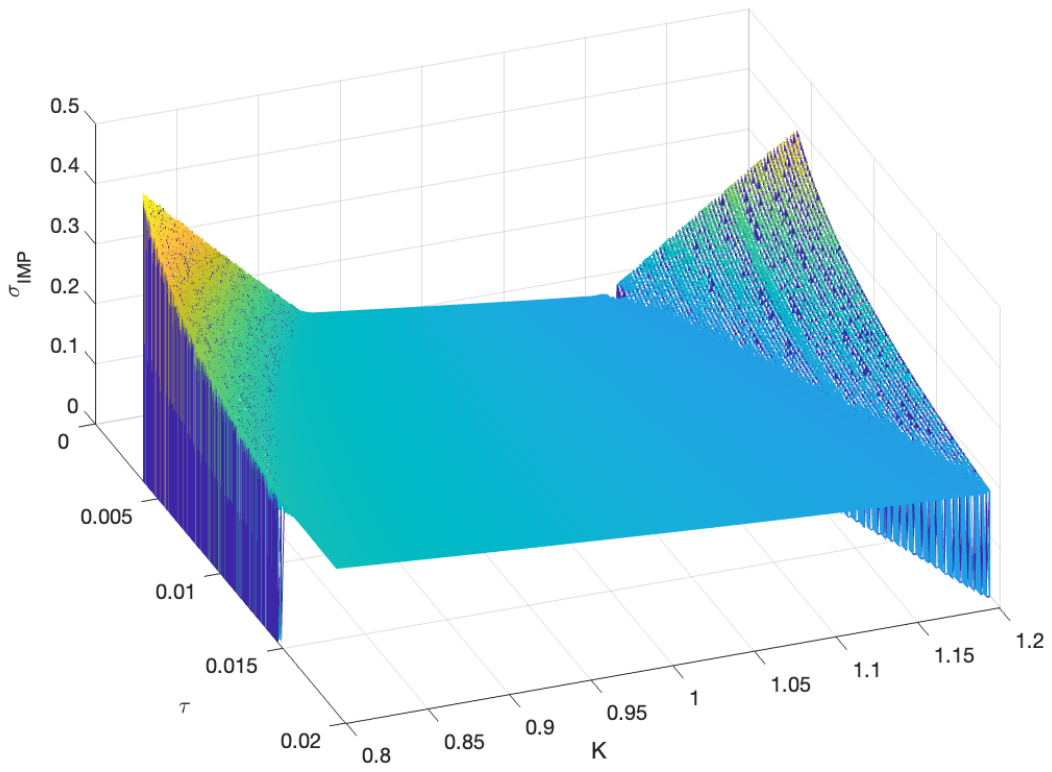}
    \caption{SINH}\label{Set1SurfaceB}
\end{subfigure}
&
\begin{subfigure}[h]{0.45\textwidth}
\centering
    \includegraphics[width=0.8\textwidth,height=0.35\textheight]{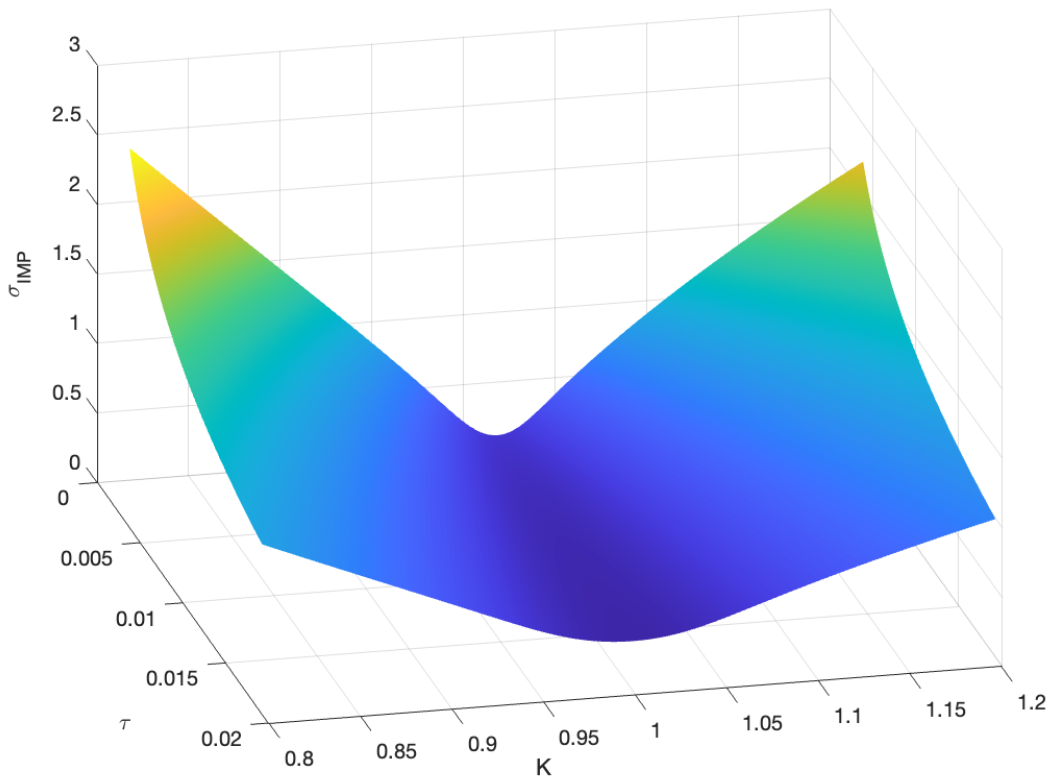}\caption{CM, $\om_1=-1.1$} \label{Set1CMsurfaceT152om1m11}
\end{subfigure}
\\ 
\begin{subfigure}[h]{0.45\textwidth}
\centering
    \includegraphics[width=0.8\textwidth,height=0.35\textheight]{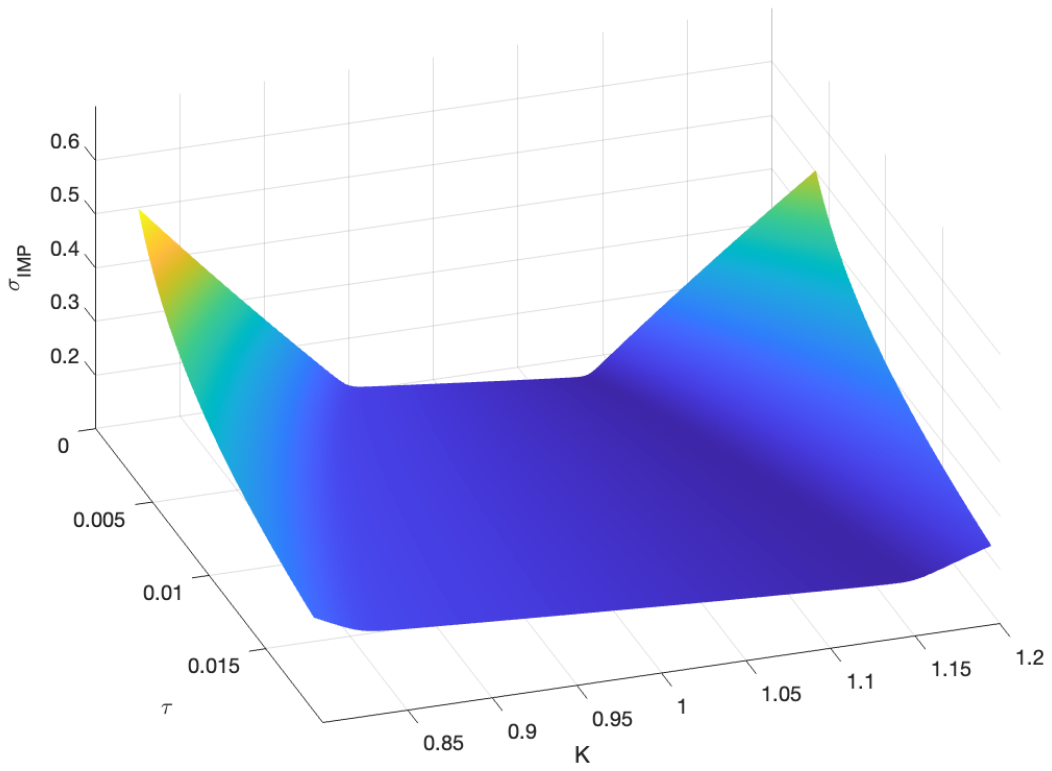} \caption{CM, $\om_1=-1.5$}\label{Set1CMsurfaceT152om1m15} \end{subfigure}
&
\begin{subfigure}{0.45\textwidth}
   \centering
    \includegraphics[width=0.8\textwidth,height=0.35\textheight]{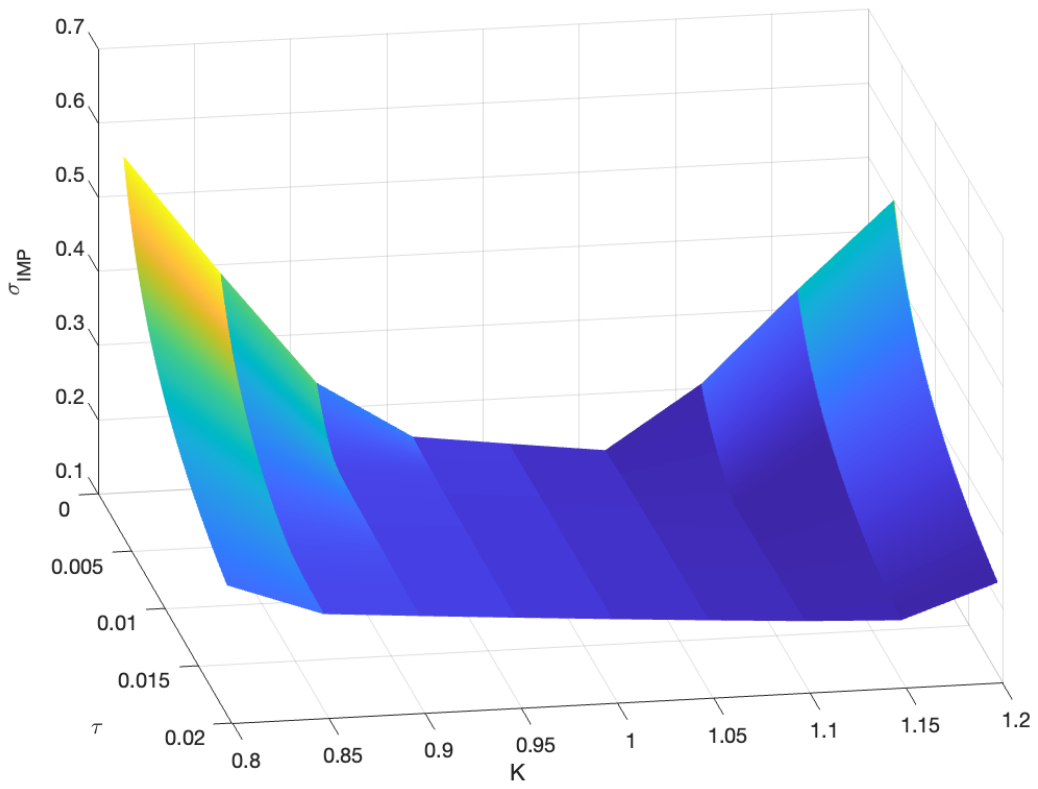}
    \caption{CM, $\om_1=-1.5$ }\label{Set1CMsurfaceT152om1m15R}
\end{subfigure}
\end{tabular}
\vskip-0.2cm
\caption{\small Implied volatility surfaces in the rough Heston model 
\cite[Example 5.1]{EuchRosenbaum2019}.
% $\al=0.62,     \ga=0.1,   \theta=0.3156,   \nu=0.331,  \rho=-0.681,   v=0.0392$, 
for time to maturity in the range (1 day, 1 week); spot $S_0=1$. If the price is outside the no-arbitrage bounds, $\sg_{IMP}$ is set to 0.
Panel (A): surface is calculated using the SINH-acceleration and 
 the modified Adams method, the parameters are as in Fig.~\ref{Set2ImpSurfacesT152}. 
 Irregular parts of the surface are where the OTM vanilla prices are smaller than E-10.
 Panels B-D: Flat iFT is used with $\ze=0.125, N=8,192$ and $\om_1=-1.1, -1.5, -1.5$, respectively, and the modified Adams method with $M=2000$. 
Irregular parts of the surface are where the OTM vanilla price is smaller than E-06. 
Panel (D) shows the effect of the interpolation: implied volatilities are
calculated at points of a sparse grid, in the result, the interpolated surface is higher than the one on Panel (C),
and the smiles are more regular. }
\label{Set1ImpVolsurfacesXiT152}
\end{figure}

\begin{figure}
\begin{tabular}{cc}
 \begin{subfigure}[h]{0.45\textwidth}

 \centering
    \includegraphics[width=0.9\textwidth,height=0.4\textheight]{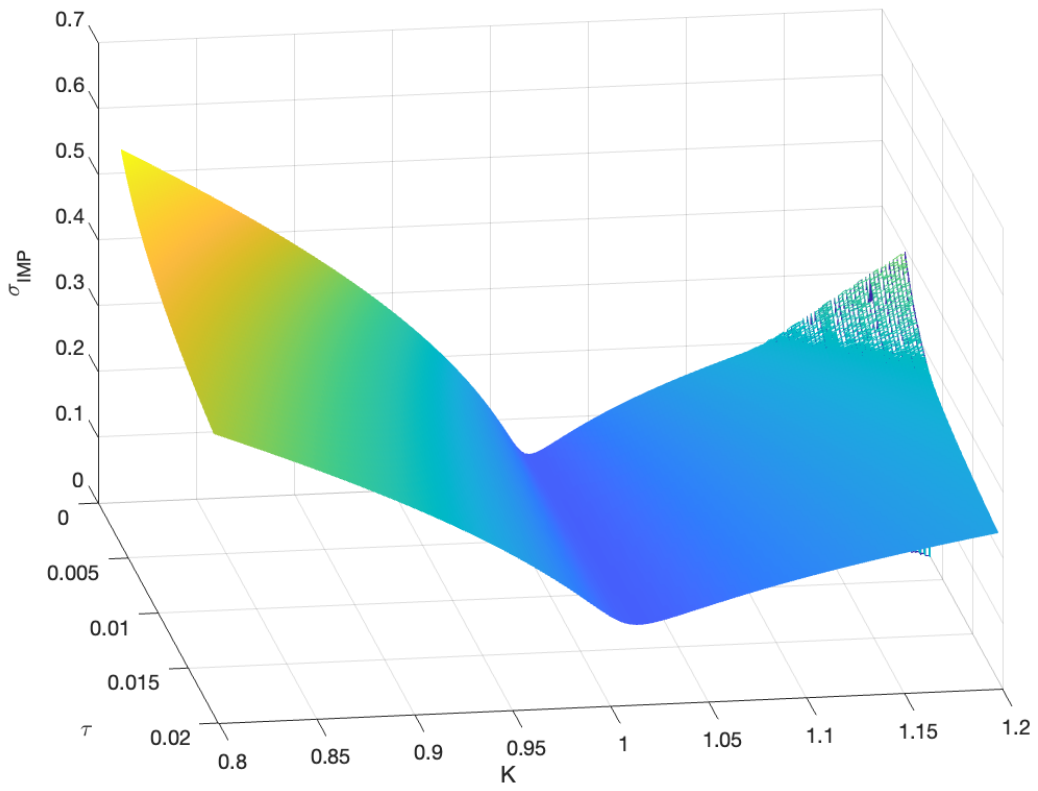}
    \caption{}\label{Set2SurfaceB}
\end{subfigure}
&
\begin{subfigure}[h]{0.45\textwidth}
\centering
    \includegraphics[width=0.9\textwidth,height=0.4\textheight]{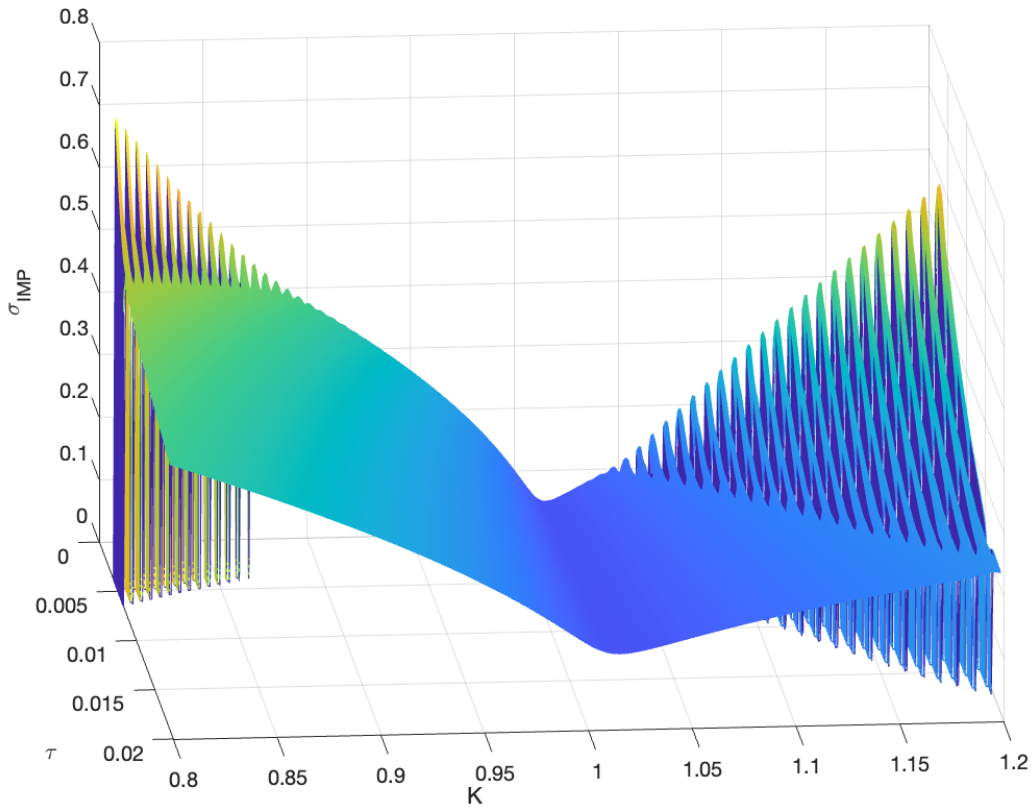}\caption{} \label{CMSet2T152om1m15}
\end{subfigure}
\end{tabular}
\caption{\small Volatility surface in the rough Heston model with parameters $\al=0.6,
    \ga=2,
    \theta=0.0225,
    \nu=0.2,
    \rho=-0.6,
    v=0.0225$, for time to maturity in the range (1 day, 1 week); spot $S_0=1$. If the price is outside the no-arbitrage bounds, $\sg_{IMP}$ is set to 0.
Panel (A): calculated using the sinh acceleration with different contour deformations for puts and calls
and simplified trapezoid rule with
 the same step $\ze=0.074956525$ and $N=94$.
 CPU time for evaluation
at all 401,1000 points is 1.67 sec., the average over 1000 runs (for the Heston model, the CPU time
is 29.1 msec.) The irregular part of the surface is where the OTM option price is smaller than E-10.
Panel (B): using Flat iFT with $\om_1=-1.5, \ze=0.125, N=8,192$. 
The irregular part of the surface is where the OTM option price is smaller than 5E-06.
In both cases, 
 the modified Adams method is used, with $M=1000$ and step $\De=1/52000$ (resp., $M=2000$, 
 and $M=2000$, $\De=1/104000$). Larger $M$ is needed to decrease the accumulated error of evaluation
 at 8,192 nodes. With $M=1000$, the irregular part is significantly larger.
 }
 
\label{Set2ImpSurfacesT152}
\end{figure}

\begin{table}
\caption{\small Dependence of implied volatilities (rounded) in the rough Heston model   on the numerical scheme.  
Example in \cite[Sect. 6.2]{KamuranEmreErkan2020}; parameters $\al=0.6$, $\ga=2$, $\rho=-0.6$,	
$\theta=	0.025$, $\nu=0.2$, $v_0=0.025$). Spot $S=1$, maturity $T=1/52$ years (1 week).}
{\tiny
\begin{tabular}{c|ccccccccc}
\hline\hline
$K$ & 0.8	& 0.85 &	0.9 &	0.95	& 1 &	1.05 &	1.1	& 1.15 &	1.2
\\
SINH & 0.4269 &	0.3686 &	0.3039&	0.2274 &	0.1280 &	0.1313 &	0.1687 &	0.2053 &	0.2053\\
iFT(0.25, 4096) & (*) & 0.3390 &	0.3009 &	0.2269 &	0.1280 &	0.1260 & (*) & (*) &(*) \\
FFT(0.25,4096) & (*) & (*) & 0.3000 &	0.2270 &	0.1279 &	0.1263 & (*) & (*) & (*) \\
iFT(0.125,9182) & 0.4273 &	0.3687 &	0.3039 &	0.2274 &	0.1280 &	0.1313 &	0.1687 &	0.2236 & (*)\\
FFT(0.125,9182) & (*) & 0.3539 &	0.3030 &	0.2274 &	0.1280 &	0.1315 &	0.1694 &	0.2175 &(*)\\
\end{tabular}
}
\begin{flushleft}{\tiny SINH - method of the present paper, $\om=0.2$ for puts, $\om=-0.2$ for calls.\\
iFT$(\ze,N)$: iFT with $\om_1=-0.5$ (Lewis-Lipton choice) and uniform grid, step $\ze$, $N$ terms.\\
FFT$(\ze,N)$: version of CM method based on FFT and interpolation, with $\om_1=-0.5$, step $\ze$, $N$ terms.\\
(*): price outside the no-arbitrage bounds.\\
\vskip-0.2cm
$\sg_{IMP}(1.2)$ in SINH-line is unreliable because the absolute value of the OTM option price is smaller than $10^{-12}$.
}
\end{flushleft}
\label{table: iFT-FFT}
\end{table}

\begin{table}
\caption{\small Prices and implied volatilities of OTM and ATM options in the rough Heston model with the parameters \eq{parEuRos}, and absolute and relative errors of prices and implied volatilities. Sinh-acceleration with $N=7$ is used.
Spot $S_0=1000$, $T=0.5$.
%Parameters: $\al=0.62$, $\ga=	0.1$, 	$\rho=-0.681$, $\theta=	0.3156$, $\nu=0.331$, $v_0=0.0392$.
}
{\small 
\begin{tabular}{c|ccccccccc}
\hline\hline
$K$ & 800 & 850 & 900 & 950 & 1000 & 1050 & 1100 & 1150 &1200\\\hline
$V$ & 6.1236 & 	12.7531 &	23.9424 & 	40.7426 &	 63.6803 &	42.7574 &	 27.4488 &	16.7859 &	9.7636\\
$err_V$ & 0.012 &	-0.0054 & 	-0.0056&	0.12 &	0.18 &	0.21 &	0.20 &	0.11 &	0.0022\\
$rel.err_V$  & 0.0019 &	-0.0042 &	-0.0023 &	0.0026 &	 0.0029&	 0.0049 &	0.0072 &	 0.0064 &	0.00023
\\\hline
$\sg_{IMP}$ & 0.23466 &	0.23171 & 	0.22967 &	0.22786 &	0.22598 &	0.22406 &	0.22219 &	0.22007 &	0.21780\\
$err_{\sg} $ & 1.2E-04 &	-3.4E-04 &	-2.6E-05 &	4.1E-04 &	6.5E-04	& 7.5E-04 &	8.0E-04 &	5.3E-04 &	1.5 E-05\\
$rel.err_{\sg} $ &  5.0E-04 &	-0.0015 &	-1.1E-04 &	0.0018 & 0.0029 &	0.0034 &	 0.0036 &	0.0024 & 	6.6E-05
\\\hline
\end{tabular}
}
\begin{flushleft}{\tiny
%Prices and absolute and relative errors of prices and implied volatilities are in units of $10^{-3}$, rounded. 
%\\
Parameters of the sinh-acceleration for OTM and ATM puts: $\om_1=0.5$, $b=	0.7700$, $\om=	0$, $\ze=0.4822$, $N=7$.\\
Parameters of the sinh-acceleration for OTM calls: $\om_1=-1.5$, $b=	0.7700$, $\om=	0$, $\ze=0.4822$, $N=7$.\\
 Modification III of the fractional Adams method: number of iterations 2, $M=9$.
 \\
Total CPU time, average over 1000 runs: 0.58 msec.
 
}
\end{flushleft}
\label{table:rough T=0.5}
 \end{table}
 
 \begin{table}
\caption{\small Lewis and Flat-iFT methods. Prices of OTM and ATM options in the rough Heston model with the parameters \eq{parEuRos}, and absolute and relative errors of prices. Spot $S_0=1000$, $T=0.5$.
%Parameters: $\al=0.62$, $\ga=	0.1$, 	$\rho=-0.681$, $\theta=	0.3156$, $\nu=0.331$, $v_0=0.0392$.
}
{\small 
\begin{tabular}{c|ccccccccc}
\hline\hline
$K$ & 800 & 850 & 900 & 950 & 1000 & 1050 & 1100 & 1150 &1200\\\hline
A & 6.2118 &	12.8047 & 	23.8779 & 	40.5500 & 	63.4207 &	42.4882 & 	27.2053 &	16.6537 &	9.7632\\
errors & 0.1000 &	-0.0025 &	-0.070 &	-0.086 &	-0.077 &	-0.062 &	-0.046 &	-0.026&	0.0018\\
  rel.errors & 0.016 &	-2.0E-04 &	-0.0029	& -0.0021 &	-0.0012 &	-0.0015	& -0.0017&	-0.0016 &	1.8E-04
\\\hline
B & 6.1135 &	12.8035 &	23.9368 &	40.6247 &	63.4921 &	42.5497 &	27.2501 &	16.6735&	9.7519\\
errors & 0.0018 &	-0.0037 &	-0.011 &	-0.012 &	-0.0054 &	-4.1E-04 &	-0.0014 &	-0.0063&	-0.0095 \\
 rel.errors & 2.8E-04 &	-2.9E-04 &	-4.6E-04 &	-2.8E-04 &	-8.6E-05 &	-9.6E-06	& -5.2E-05 &	-3.8E-04	&-9.7E-04 \\\hline
\end{tabular}
}
\begin{flushleft}{\tiny
A: Lewis method with change of variables $\xi=-0.5i+u/(1-u)$, 20-point Gauss-Legendre quadrature and Modification II of the fractional Adams method with 2 iterations and $M=9$. Total CPU time, average over 1000 runs: 1.9 msec.\\
B: Flat iFT-BM method. $\om_1=-0.5, \ze=2.7629$, $N=7$ and Modification II of the fractional Adams method with 2 iterations and $M=9$. Total CPU time, average over 1000 runs: 3.1 msec.

}
\end{flushleft}
\label{table:rough T=0.5, Lewis and iFT-BM}
 \end{table}
 
 \begin{table}
\caption{\small Long maturity case: $T=5$. SINH-acceleration, Lewis and Flat-iFT methods. Prices of OTM and ATM options in the rough Heston model with the parameters \eq{parEuRos}, and absolute and relative errors of prices. Spot $S_0=1000$.
%Parameters: $\al=0.62$, $\ga=	0.1$, 	$\rho=-0.681$, $\theta=	0.3156$, $\nu=0.331$, $v_0=0.0392$.
}
{\small 
\begin{tabular}{c|ccccccccc}
\hline\hline
$K$ & 800 & 850 & 900 & 950 & 1000 & 1050 & 1100 & 1150 &1200
\\\hline
A & 141.1813 &	166.1588 &	192.9274 &	221.3824 &	251.4193 &	232.9352 &	215.8300 &	200.0077 &	185.3770\\
err. & -0.013 &	-0.019 &	-0.0230 &	-0.026 &	-0.028 &	-0.029 &	-0.029 &	-0.029 &	-0.028
\\
  rel.err. & -9.38E-05 &	-1.1E-04 &	-1.2E-04 &	-1.2E-04 &	-1.1E-04	&-1.2E-04 &	-14.E-04 &	-1.5E-04 &	-1.5E-04
\\\hline
B & 141.1035 &	166.0792 &	192.8534 &	221.3209 &	251.3769 &	232.91891 &	215.8475 &	200.0673 &	185.4877\\
err. & -0.089 & 	-0.096 &	-0.094 &	-0.085 &	-0.068 &	-0.043 &	-0.0093 &	0.033 &	0.085
 \\
 rel.err. & -6.3E-04 &	-5.8E-04 &	-4.9E-04&	-3.8E-04 &	-2.7E-04 &	-1.8E-04 &	-4.3E-05	&1.7E-04&	4.6E-04
   \\\hline
C & 141.5498 &	166.3816 &	192.9561 &	221.1782 &	250.9531 &	231.3579&	214.6252 &	199.1734 &	184.8998 \\
err. & 0.36 &	0.21 &	0.0083 &	-0.23	 & -0.49 &	-1.6 &	-1.2 &	-0.86 &	-0.50\\
  rel.err. & 0.0025&	0.0012 &	4.3E-05 &	-0.0010 &	-0.0020 &	-0.0069 &	-0.0057 &	-0.0043 &	-0.0027
   \\\hline

\end{tabular}
}
\begin{flushleft}{\tiny
A: Lewis method with  change of variables $\xi=-0.5i+u/(1-u)$, 15-point Gauss-Legendre quadrature and Modification II of the fractional Adams method with 2 iterations and $M=50$. Total CPU time, average over 1000 runs: 21.2 msec.\\
B: Flat iFT-BM method with $\om_1=-0.5, \ze=2.7629$, $N=4$ and Modification II of the fractional Adams method with 2 iterations and $M=15$. Total CPU time, average over 1000 runs: 1.12 msec.\\
C: SINH-acceleration with $\om_1=0.5,	b=0.7699, \om=0, \ze=0.3858, N=5$ for OTM and ATM puts 
and $\om_1=-1.5,	b=0.7699, \om=0, \ze=0.3858, N=5$ for OTM calls. Modification III of the fractional Adams method with 2 iterations and $M=10$ is used. Total CPU time, average over 1000 runs: 0.58 msec.

}
\end{flushleft}
\label{table:T=5}
 \end{table}

 \begin{table}
\caption{\small Prices and  implied volatilities of OTM and ATM options in the rough Heston model with the parameters \eq{parEuRos}, and absolute and relative errors of prices and implied volatilities.
Sinh-acceleration with $N=12$ is used.
Spot $S_0=1000$, $T=1/12$.
%Parameters: $\al=0.62$, $\ga=	0.1$, 	$\rho=-0.681$, $\theta=	0.3156$, $\nu=0.331$, $v_0=0.0392$.
}
{\small 
\begin{tabular}{c|ccccccccc}
\hline\hline
$K$ & 800 & 850 & 900 & 950 & 1000 & 1050 & 1100 & 1150 &1200\\\hline

$V$ & 0.0200 &	 0.1140 &	1.1556 & 	6.7265 & 	23.901 &	6.8184 &	 1.2122 &	0.14036 &	0.01296 \\
$err_V$ & 0.015 &	0.0058 &	0.012 &	0.0029 &	0.012 &	0.0093 &	0.0070 &	0.015 &	0.0056\\
$rel.err_V$  & 2.77 & 	0.054 &0.010	 & 4.3E-04 &	5.1E-04 &	0.0014 &	0.0058 &	0.123 &	0.077
\\\hline
$\sg_{IMP}$ & 0.25248 & 	0.22387 &	0.21776 &	0.21235 &	0.20764 &	0.20309 &	0.19887 &	0.19719 &	0.19915\\
$err_{\sg} $ &0.024 &	0.0013 &	4.5E-04 & 	3.7E-04 &	1.1E-04 &	1.1E-04 &	2.3E-04 &	0.0026 &	0.0085\\
$rel.err_{\sg} $ &  0.11 &	0.059 &	0.0021 &	1.7E-04 & 	5.1E-04 &	5.5E-04 &	0.0012 &	0.014 &	0.044

\\\hline
\end{tabular}
}
\begin{flushleft}{\tiny
Prices and absolute and relative errors of prices and implied volatilities are in units of $10^{-3}$, rounded. 
\\
Parameters of the sinh-acceleration for OTM and ATM puts: $\om_1=0.5$, $b=	0.7700$, $\om=	0$, $\ze=0.3858$, $N=12$.\\
Parameters of the sinh-acceleration for OTM calls: $\om_1=-1.5$, $b=	0.7700$, $\om=	0$, $\ze=0.3858$, $N=12$.\\
 Modification III  of the fractional  Adams method: number of iterations 2, $M=20$.
 \\
Total CPU time, average over 1000 runs: 0.93 msec.
 
}
\end{flushleft}
\label{table:rough T=1/12}
 \end{table}
 
 \begin{table}
\caption{\small A refined parameter choice for maturity $T=1/12$ years. Prices and implied volatilities of OTM and ATM options in the rough Heston model with the parameters \eq{parEuRos}, and absolute and relative errors of prices and implied volatilities. Spot $S_0=1000$.
%Parameters: $\al=0.62$, $\ga=	0.1$, 	$\rho=-0.681$, $\theta=	0.3156$, $\nu=0.331$, $v_0=0.0392$.
}
{\small 
\begin{tabular}{c|ccccccccc}
\hline\hline
$K$ & 0.80 & 0.85 & 0.90 & 0.95 & 1.00 & 1.05 & 1.10 & 1.15 &1.20\\\hline
$V$ & 0.00538 &	0.10826 &	1.14357 & 	6.72340 & 	23.8966 &	 6.80906 &
	1.20540 &	0.125 & 	0.00742 \\
$err_V$ & 7.5E-05 &	3.0E-05 &	-2.0E-04 &-5.9E-04 &-6.1E-04 &-1.9E-04 &	1.1E-04	&1.2E-04
&	1.0E-04\\
$rel.err_V$  & 0.014 &	2.7E-04 &	-1.8E-04 &	-8.8E-05	& -2.5E-05 &	-2.8E-05 &	9.0E-05
&	9.5E-04 &	0.014
\\\hline
$\sg_{IMP}$ & 0.22826 &	0.22257 &	0.21732 &	0.21231 &	0.20753&	0.20297 &	0.19865 &	0.19456 &	0.19088\\
$err_{\sg} $ & 2.3E-04 &	1.4E-05 &	-1.4E-06 &	-2.3E-06 &	-1.7E-06 &	-3.6E-07
&	3.5E-06 &	1.9E-05 &	2.0E-04\\
$rel.err_{\sg} $ & 0.0010 &	6.4E-05 &	-6.4E-06 &	-1.1E-05	& -8.2E-06 &	-1.8E-06
&	1.8E-05 &	1.0E-04 &	0.0010
\\\hline
\end{tabular}
}
\begin{flushleft}{\tiny
Prices and absolute and relative errors of prices and implied volatilities are in units of $10^{-3}$, rounded. 
\\
Parameters of the sinh-acceleration for OTM and ATM puts: $\om_1=0.5$, $b=	0.7700$, $\om=	0.2$, $\ze=0.2054$, $N=24$.\\
Parameters of the sinh-acceleration for OTM calls: $\om_1=-1.5$, $b=	0.7700$, $\om=	-0.2$, $\ze=0.2054$, $N=24$.\\
 Modification III  of the fractional  Adams method: number of iterations 2, $M=40$.
 \\
Total CPU time, average over 1000 runs: 2.7 msec.
 
}
\end{flushleft}
\label{table:rough T=1/12 refined}
 \end{table}
 
 \begin{table}
\caption{\small Errors of COS method. ``Call" prices (rounded) in the rough Heston model; parameters are $ \al=0.6,  \ga= 0.1,\ \theta = 0.3156,  \nu = 0.331, \rho=-0.681,  v = 0.0392$,
$r=0.3$, $S_0=100$, $T=1$.
 Errors shown are for OTM ``put" and ATM and OTM  ``call"  options.}
 {\small
\begin{tabular}{c|lll}
\hline\hline
$K$ & 80 & 100 & 120\\\hline
$BB$ & 16.21398847353 & 7.08261889490 & 2.638141762   \\\hline
$V_{fast}$&16.21390583 & 7.07778568 &   2.634570606 \\
$Err(V_{fast})$ &-8.0E-05 & -0.0048 &  -0.0036    \\
$RelErr(V_{fast})$ &-5.7E-05  & -6.8E-04  & -0.0014  \\\hline
$V_{COS}$ & 16.1349 & 6.6198  & 2.0529 \\
$Err_{COS}$ & -7.9E-03 & -0.46  & -0.59 \\
$ RelErr_{COS}$ &-5.7E-03 & -6.5E-03& -0.22 \\\hline
\end{tabular}
}
\begin{flushleft}{\tiny $BB$: sinh-acceleration and modified Adams method of the present paper; $M=31,623$,  
$\om_1=-0.5$, $b=0.7699$, $\om=0$, $\ze=0.11481$, $N=62$, $M=31,623$. For $T=1$ and $K=80,100,120$, absolute errors are smaller than $10^{-10}$. \\
$V_{COS}$: prices (Table 1 on \cite[p.52]{KamuranEmreErkan2020}) calculated using COS method with $N=160$ and the fractional Adams method with 2000 terms;
CPU time in the range 1-1.1 sec. for each pair.
\\
$V_{fast}$,
parameters of the sinh-acceleration for OTM and ATM puts: $\om_1=0.5$, $b=	0.7700$, $\om=	0$, $\ze=0.4822$, $N=7$.\\
$V_{fast}$, parameters of the sinh-acceleration for OTM calls: $\om_1=-1.5$, $b=	0.7700$, $\om=0$, $\ze=0.4822$, $N=7$.\\
 Modification III  of the fractional  Adams method: number of iterations 2, $M=7$.
 \\
Total CPU time, average over 1000 runs: 0.74 msec.

}
\end{flushleft}
\label{table: COS1_K}
\end{table}

\begin{table}
\caption{\small Moderate maturities, spot $S_0=1$. Relative errors (rounded) of calculations of OTM and ATM puts $(K\le 1$) and OTM calls $(K>1$)  in the rough Heston model with the parameters \eq{parEuRos} and CPU time (in msec., the average over 1000 runs)  for several numerical schemes.  
 }
{\tiny
\begin{tabular}{c|rrrrrrrrr|r}
\hline\hline
$T=2$ & & & & & & & & &  & Time\\\hline
$K$ & 0.80 & 0.85 & 0.90& 0.95 & 1.00 & 1.05 & 1.10 & 1.15 & 1.20 &\\\hline
SINH & -1.5E-05 &	-1.2E-05 &	-9.6E-06 &	-8.0E-06 &	-6.7E-06 &
-7.9E-06 &	-9.3E-06 &	-1.1E-05 &	-1.3E-05 & 169.6\\
$V_H$ & 3.9E-06 &	-2.6E-06	& 1.4E-06	& 3.8E-06	& 1.3E-06 & 1.1E-06	
&-1.8E-06	& 1.6E-06 &	-3.1E-06 &\\
Flat iFT-BM & -7.2E-06 &	-8.6E-06 &	-8.6E-06 &	-7.3E-06 &	-5.5E-06
& -5.2E-06 &	-4.46E-06	&-3.6E-06	 &-3.2E-06 & 97.7\\
Flat iFT & 2.7E-07 &	-1.6E-06 &	-2.4E-06	& -2.2E-06 &	-1.5E-06
& -1.0E-06 &	-4.9E-08 &	1.0E-06 &	1.8E-06 & 661.5\\
Lewis 30 & 4.4E-06	& -1.9E-07 &	-1.4E-06 &	-1.3E-06 &	-8.3E-07
& -5.7E-07 &	-2.4E-07 &	8.5E-08 &	3.7E-07 & 400.8
\\\hline
\end{tabular}
}

\begin{flushleft}{\tiny 
SINH: $\om_1=-0.5$, $b=0.769884522$, $\om=0$, 	$\ze=0.285754315$, $N=12$, Modification II with $M=317$.
%CPU time:  169.6
\\
$V_H$: hybrid method of \cite{RoughNotTough} \\
Flat iFT-BM: $\sg_0=1$, $\om_1=-0.5$, 	$\ze=0.717626524$, $N=	16$, Modification II with $M=317$.
% CPU time: 97.7
 \\
Flat FT: $\om_1=-0.5	$,	$\ze=0.109637386$, $N=110$ , Modification II with $M=317$.
%CPU time: 661.5 
\\
Lewis 30: Lewis method and Gauss-Legendre quadrature with 30 terms, Modification II with $M=317$.
 %CPU time: 400.8.
\\
CPU time is for the evaluation of $\Phi(\xi_k, \tau_m)$, for $k=0,1,\ldots, N$, $m=1,\ldots, 317$.\\
Flat iFT-BM  is used with the parallelization w.r.t. $\xi$.\\
SINH, Flat iFT and Lewis method are used without the parallelization w.r.t. $\xi$. \\
For the Lewis method, the nodes and weights are precalculated.\\
\vskip-0.2cm
For $V_H$, the CPU time is in the range 593-667 msec. per strike.  
}
\end{flushleft}

{\tiny
\begin{tabular}{c|rrrrrrrrr|r}
\hline\hline
$T=1$ & & & & & & & & & & Time \\\hline
%$K$ & 0.80 & 0.85 & 0.90& 0.95 & 1.00 & 1.05 & 1.10 & 1.15 & 1.20\\\hline
SINH & -2.1E-05 &	-1.5E-05 &	-1.1E-05 &	-8.0E-06 &	-6.1E-06	&
-7.8E-06 &	-1.0E-05 &	-1.3E-05 &	-1.7E-05 & 295.6\\
$V_H$ & -1.8E-05 &	9.9E-06 &	-4.5E-06 &	7.2E-06 &	-1.9E-06 &
-6.7E-07 &	-1.1E-05 &	-1.0E-05	& -1.9E-05 &\\

Flat iFT-BM & 9.3E-06 &	3.5E-06 &	-2.1E-06 &	-3.3E-06 &	-1.6E-06 &
7.3E-07 &	4.1E-06 &	6.0E-06 &	3.6E-06 & 102.9\\

Flat iFT & 1.1E-05 &	3.6E-06 &	-3.0E-06	& -4.0E-06 &	-1.8E-06 &
1.1E-06 &	5.1E-06 &	7.2E-06 &	4.14E-06 & 980.8\\

Lewis 30 & 1.2E-04 &	3.4E-05 &	-1.706 &	-7.7E-06 & -5.1E-06 & 
-3.4E-06 &	-7.4E-07 &	2.3E-06 &	5.8E-06 & 144.1
\\\hline
\end{tabular}
}
\begin{flushleft}{\tiny 
SINH: $\om_1=-0.5$, $b=0.769884522$, $\om=0$, 	$\ze=0.285754315$, $N=14$, Modification II with $M=399$.
\\
$V_H$: hybrid method of \cite{RoughNotTough} \\
Flat iFT-BM: $\sg_0=0.5$, $\om_1=-0.5$, 	$\ze=0.717626524$, $N=	22$, Modification II with $M=317$. \\
Flat FT: $\om_1=-0.5	$,	$\ze=0.0877$, $N=200$, Modification II with $M=317$.\\
Lewis 30: Lewis method and Gauss-Legendre quadrature with 30 terms, Modification II with $M=317$.\\
For $V_H$, the CPU time is in the range 548-582 msec. per strike. 
 
}
\end{flushleft}

{\tiny
\begin{tabular}{c|rrrrrrrrr|r}
\hline\hline
$T=0.5$ & & & & & & & & & & Time \\\hline
%$K$ & 0.80 & 0.85 & 0.90& 0.95 & 1.00 & 1.05 & 1.10 & 1.15 & 1.20\\\hline
SINH & 4.4E-05 &	5.1E-05 &	-7.9E-06	& -1.9E-05 &	-2.7E-06 &
-3.5E-06 &	-37E-06	& -5.1E-06 &	-8.6E-06 & 329.6\\

$V_H$ & 3.3E-05 &	-2.0E-05 &	2.8E-06 &	-5.2E-06 &	-8.6E-06 &
-1.7E-06 &	-1.9E-05 &	7.7E-06 &	-4.3E-05 & \\

Flat iFT-BM & 4.4E-05 &	5.1E-05 &	-7.9E-06 &	-1.9E-05 &	-2.7E-06 &
1.8E-05 &	2.6E-05 &	-1.2E-05 &	-9.4E-05 & 107.3\\

Flat iFT & -1.4E-05 &	1.0E-05 &	2.1E-06 &	-4.4E-06 &	-2.1E-06 &
2.5E-06 &	6.2E-06 &	-1.73E-06 &	-2.2E-05 & 1,192.3\\

Lewis 35 & 7.7E-04 &	1.2E-04 &	-3.2E-05 &	-1.9E-05 &	-2.9E-06 &
1.5E-06 &	2.2E-06 &	5.8E-06 &	2.9E-05 & 465.4
\\\hline
\end{tabular}
}
\begin{flushleft}{\tiny 
SINH: $\om_1=-0.5$, $b=0.769884522$, $\om=0$, 	$\ze=0.1836992027$, $N=23$, Modification II with $M=317$.
\\
$V_H$: hybrid method of \cite{RoughNotTough} \\
Flat iFT-BM: $\sg_0=0.5$, $\om_1=-0.5$, 	$\ze=0.789389176$, $N=	30$, Modification II with $M=317$. \\

Flat FT: $\om_1=-0.5	$,	$\ze=0.0877$, $N=200$, Modification II with $M=317$.\\

Lewis 35: Lewis method and Gauss-Legendre quadrature with 35 terms, Modification II with $M=317$. \\

For $V_H$, the CPU time is in the range 666-689 msec. per strike. 

}
\end{flushleft}

\label{table:rel_errors_moderate}

 \end{table}
 
  \begin{table}
\caption{\small Short maturities,  spot $S_0=1$. Relative errors (rounded) of calculations of OTM and ATM puts $(K\le 1$) and OTM calls $(K>1$)  in the rough Heston model with the parameters \eq{parEuRos} for several numerical scheme and CPU time (in msec., the average over 1000 runs). 
  }

{\tiny
\begin{tabular}{c|rrrrrrrrr|r}
\hline\hline
$T=1/12$ & & & & & & & & & & Time \\\hline
$K$ & 0.80 & 0.85 & 0.90& 0.95 & 1.00 & 1.05 & 1.10 & 1.15 & 1.20 & \\\hline
SINH & 1.4E-05 &	-2.1E-04 &	-1.2E-05	& -2.9E-07 &	-1.8E-06 &
-5.8E-07	& 5.8E-06	& -1.3E-04 & 2.9E-03 & 415.7\\

$V_H$ & -0.057 &	-0.0018 &	3.5E-04	& -8.7E-05 &	-3.3E-05 &
-3.3E-05 &	-1.35E-05	& -2.4E-04&	7.8E-03 & \\

Flat iFT-BM & -0.092 &	1.5E-03 &	1.2E-04 &	-5.0E-05 &	1.6E-05 &
-3.4E-05 &	-1.1E-04 &	5.2E-04	& -0.15 & 133.3\\

Flat iFT & -5.6 &	0.13 &	0.017 &	-0.0083& 0.0024 &
-9.4E-04 &	-0.047 &	0.42 &	2.7 & 2,341.2\\

Lewis 80 & 0.045 &	-0.0073 &	1.7E-04 &	-5.8E-06 & 6.9E-08 &	-3.8E-07	& 
4.0E-06 &	-3.3E-05 &	-0.013 & 1,062.3
\\\hline
\end{tabular}
}
\begin{flushleft}{\tiny 
SINH: $\om_1=-0.5$, $b=0.769884522$, $\om=0$, 	$\ze=0.1836992027$, $N=28$, Modification II with $M=317$.
\\
$V_H$: hybrid method of \cite{RoughNotTough} \\
Flat iFT-BM: $\sg_0=0.5$, $\om_1=-0.5$, 	$\ze=0.717626524$, $N=	80$, Modification II with $M=317$. \\

Flat FT: $\om_1=-0.5	$,	$\ze=0.0877$, $N=450$, Modification II with $M=317$.\\
Lewis 80: Lewis method and Gauss-Legendre quadrature with 80 terms, Modification II with $M=317$. \\
For $V_H$, the CPU time is in the range 410-423 msec. per strike.

}
\end{flushleft}

{\tiny
\begin{tabular}{c|rrrrrrr|r}
\hline\hline
$T=1/52$ & & & & & & & &  Time \\\hline
$K$ &  0.85 & 0.90& 0.95 & 1.00 & 1.05 & 1.10 & 1.15 &\\\hline
SINH & -0.42 &	1.5E-03 &	-1.6E-05 &	-6.6E-06 &
	-2.3E-04	& -0.043 & -205 & 154.8 \\
$V_H$ & 	(**) &	0.013 & 0.085 &	0.016 &
0.096 &	0.32 &	0.82&	 \\
Flat iFT-BM & 
26.5 &	2.8E-03 &	-1.1E-04 &	-9.4E-07 &
1.3E-04 & 	0.075 &	1,030& 339.3\\
Flat iFT & 1,167 &	0.71 &	3.9E-04 &	1.5E-04 &
	1.7E-03 &	-1.7 & -49,413 & 1,664.2\\
Lewis 100 & 25,177 &	1.2 &	 3.5E-04 &	4.3E-07 &
6.3E-05 &	0.60 &	-119,127 & 187.8
\\\hline
\end{tabular}
}
\begin{flushleft}{\tiny 
At $K=0.8$ and $K=1.2$, the prices of OTM options are smaller than $10^{-12}$, and the benchmark prices
cannot be calculated using double precision arithmetic.\\

SINH, puts: $\om_1=0.325762041$, $b=1.014615984$, $\om=0.2$, 	$\ze=0.145086905$, $N=38$, Modification II with $M=100$\\
SINH, calls: $\om_1=-1.325762041$, $b=1.014615984$, $\om=-0.2$, 	$\ze=0.145086905$, $N=38$, Modification II with $M=100$.
\\
$V_H$: hybrid method of \cite{RoughNotTough}; (**): the call price in \cite{RoughNotTough} implies that the price of the put is 0. \\
Flat iFT-BM: $\sg_0=0.5$, $\om_1=-0.5$, 	$\ze=0.717626524$, $N=	200$, Modification II with $M=100$. \\

Flat FT: $\om_1=-0.5	$,	$\ze=0.07309159$, $N=1500$, modification II with $M=100$.\\
Lewis 100: Lewis method and Gauss-Legendre quadrature with 100 terms, Modification II with $M=100$. \\
The order of the errors of Flat iFT-BM, Flat FT and Lewis 100 does decrease if $N$ increases further. \\
\vskip-0.2cm
For $V_H$, the CPU time is in the range 125-164 msec. per strike. 
}
\end{flushleft}

{\tiny
\begin{tabular}{c|rrr|r}
\hline\hline
$T=1/252$  & & & &   Time \\\hline
$K$ &   0.95 & 1.00 & 1.05  &\\\hline
SINH & 	-2.7E-03 &	4.7E-07 & 9.0E-03 & 212.1\\
$V_H$ & 	11.2 &	1.7E-04 &
18.3 &		 \\
Flat iFT-BM & -0.51 &	1.E-04 & 18.8 & 557.8\\
Flat iFT & -17.8 &	3.1E-03 & -370 & 1,664.2\\
Lewis 100 & 6.3&	-2.2E-05 & 270 & 190.1 
\\\hline
\end{tabular}
}
\begin{flushleft}{\tiny 
At $K=0.80, 0.85, 0.90$ and $K=1.10, 1.15, 1.20$, the prices of OTM options are smaller than $10^{-12}$, and the benchmark prices
cannot be calculated accurately using double precision arithmetic.\\

SINH, puts: $\om_1=0.325762041$, $b=1.014615984$, $\om=0.2$, 	$\ze=0.145086905$, $N=46$, Modification II with $M=100$.\\
SINH, calls: $\om_1=-1.325762041$, $b=1.014615984$, $\om=-0.2$, 	$\ze=0.145086905$, $N=46$, Modification II with $M=100$.
\\
$V_H$: hybrid method of \cite{RoughNotTough}. \\
Flat iFT-BM: $\sg_0=0.5$, $\om_1=-0.5$, 	$\ze=0.717626524$, $N=	350$, Modification II with $M=100$. \\

Flat FT: $\om_1=-0.5	$,	$\ze=0.07309159$, $N=1500$, Modification II with $M=100$.\\
Lewis 100: Lewis method and Gauss-Legendre quadrature with $N=100$ terms, Modification II with $M=100$. \\
The order of the errors of Flat iFT-BM, Flat FT and Lewis  does decrease if $N$ increases further.\\
\vskip-0.2cm
For $V_H$, the CPU time is in the range 154-196 msec. per strike.  
}
\end{flushleft}

\label{table:rel_errors_short}
 \end{table}
 
  \begin{table}
\caption{\small Implied volatilities for options of short maturities in Table~\ref{table:rel_errors_short}.}

{\tiny
\begin{tabular}{c|rrrrrrrrr}
\hline\hline
$T=1/12$ & & & & & & & & & \\\hline
$K$ &  0.80 & 0.85 & 0.90& 0.95 & 1.00 & 1.05 & 1.10 & 1.15 & 1.20\\\hline
BB & 0.2280 &	0.2226 &	0.2173 &	0.2123 &	0.2075 &	0.2030 &	0.1986 &	0.1945 &	0.1907\\
SINH & 0.2280 &	0.2225 &	0.2173 &	0.2123 &	0.2075 &	0.2030 &	0.1986 &	0.1945 &	0.1907  \\
$V_H$ & 	0.2271 &	0.2225 &	0.2173 &	0.2123&	0.2075 &	0.2030 &	0.1986 &	0.1944 &	0.1907	 \\
Flat iFT-BM & 
0.2265 &	0.2226 &	0.2173 &	0.2123 &	0.2075 &	0.2030 &	0.1986 &	0.1947 &	0.1884\\
Flat iFT & (**) &	0.2257 &	0.2181 &	0.2116 &	0.2080 &	0.2030 &	0.1968 &	0.2029 &	0.2116 \\
Lewis 100 & 0.2243 &	0.2226 &	0.2173 &	0.2123 &	0.2075 &	0.2030 &	0.1986 &	0.1945 &	0.1911

\\\hline
\end{tabular}

{\tiny
\begin{tabular}{c|rrrrrrrrr}
\hline\hline
$T=1/52$ & & & & & & &  & & \\\hline
$K$ & 0.8 & 0.85 & 0.90& 0.95 & 1.00 & 1.05 & 1.10 & 1.15  & 1.20\\\hline
BB &  0.2383 &	0.2288& 	0.2195&	0.2105 &0.2018
&0.1935 &	0.1857 & 0.1786 & 0.1737\\
SINH &  
0.2450 &	0.2288 &	0.2195 &	0.2105 &	0.2018&	0.1935 &	0.1857 &	0.1786&	0.1703\\
$V_H$ & (**) & 	(**) &	0.2197 &	0.2138 &	0.2051 &	0.1968 &	0.1889 &	0.1818 & 0.1843 \\
Flat iFT-BM & (*) & 0.2600 &	0.2196 &	0.2105 &	0.2018 &	0.1935 &	0.1866 &	0.2291 & 0.2929\\
Flat iFT& 0.4029 & 0.3147 &	0.2280 &	0.2107 &	0.2019 &	0.1936&(*) &	(*)  & 0.3071\\
Lewis 100& 0.5883 & 0.3928&	0.2321 &	0.2106 &	0.2018&0.1935 &	0.1913 &	(*)  & (*)
\\\hline
\end{tabular}
}

{\tiny
\begin{tabular}{c|rrr}
\hline\hline
$T=1/252$  & & &     \\\hline
$K$ &   0.95 & 1.00 & 1.05  \\\hline
BB & 0.2154 &	0.1994 &	0.1841\\
SINH & 	0.2154 &	0.1994 &	0.1841 \\
$V_H$ & 	0.2552 &	0.1994 &	0.2174 \\
Flat iFT-BM & 0.2068 &	0.1994 &	0.2178 \\
Flat iFT & (*)	&0.2000 &	(*) \\
Lewis 100 & 0.2456 &	0.1994 &	0.2661  
\\\hline
\end{tabular}
}

\begin{flushleft}{\tiny BB: benchmark.\\
(*): the price outside the no-arbitrage bounds.\\
(**): the put price is smaller than $10^{-12}$.\\
}
\end{flushleft}
}
\label{table:implvol_short}

\end{table}

 \begin{table}
\caption{\small Call options in  the rough Heston model with the parameters \eq{parEuRos}, and the parameters and errors
of sinh-acceleration and hybrid methods.  Spot $S_0=1$, $T=2$.
}
{\small 
\begin{tabular}{c|c|cr|cr}
\hline\hline
$K$ & BB &    $V_{fast}$ & Err & $V_H$ & Err \\
0.8    &     
0.254300800136
     & 0.254299995 & -8.0E-07 & 
0.254301 &  2.0E-07 \\	   
 0.85    &  
 0.222091202277
    & 0.222090343& -8.6E-07 & 
 0.222091 &  -2.0E-07\\
 0.9     &   0.192897881619
      &   0.192896973 & -9.1E-07 &
0.192898 & 1.2E-07 \\
 0.95    &  
 0.166675570337
       &   0.166674623 & -9.5E-07 & 
0.166676 &  4.3E-07 \\
 1.00  &    
 0.143318830614
    &            0.143317856 & -9.7E-07&
0.143319 &  1.7E-08\\
 1.05 &     
 0.122675884585
  &                0.122674898 & -9.9E-07 &
0.122676 &  1.2E-07\\
1.10 &       
0.104562201605
 &                 0.104561215 &  -9.9E-07 &
 0.104562 & -2.0E-07\\
 1.15 &     
 0.088772875331
  &                0.088771901 &  -9.7E-07 &
0.088773 &  1.2E-07 \\
 1.20  &     
 0.075093239373
 &               0.075092286 & -9.6E-07 &
0.075093 & -2.4E-07\\
\\\hline
\end{tabular}
}
\begin{flushleft}{\tiny
BB: sinh-acceleration (parameters: $\om_1=-0.5, b=0.7699, \om=0, \ze=0.1006, N=66$), and Modification II of the Adams method with $M=20000$, CPU time 191 sec, 1 run. The differences of prices with the ones obtained with sinh-acceleration (puts and calls are calculated separately, $\om=\pm 0.2$, $\om_1>0$ and $\om_1<0$, respectively, $\ze=0.8901$, $N=67$)  and Modification II of Adams method with $M=20000$ are less than E-12.
\\
$V_{fast}$: sinh-acceleration, parameters $\om_1=-0.5$, $b=0.769884522$, $\om=0$,  $\ze=0.2858$, $N=12$, and Modification II of Adams method with $M=317$, CPU time 69.6 msec, average over 1000 runs.\\
$V_H$: hybrid method of \cite{RoughNotTough} (prices are rescaled). CM method is used with $\om_1=-2.1$;
$\ze$ and $N$ are not specified. \\
\vskip-0.2cm The CPU time is in the range 593-667 msec. per strike. 
}
\end{flushleft}

\label{table:T=2}
 \end{table}

 \begin{table}
\caption{\small Prices of call options  in the rough Heston model with the parameters \eq{parEuRos}, and the parameters and errors
of different numerical schemes.  Spot $S_0=1$, $T=1$.
}
{\small
\begin{tabular}{c|c|cr|cr}
\hline\hline
$K$ & BB &   $V_{fast}$ & Err & $V_H$ & Err \\
0.8    &             
0.221366383102
   &0.221365934 & -4.5E-07 & 
0.221366&  -3.8E-07 \\	   
 0.85    &    
 0.183528673034        &0.183528167& -5.1E-07 & 
 0.183529 &  3.3E-07\\
 0.9     &         
 0.149672232338
   &  0.149671686& -5.5E-07 &
0.149672 & -2.3E-07 \\
 0.95    &  
  0.120058500285
         &   0.120057930 & -5.7E-07 & 
0.120059 &  5.0E-07 \\
 1.00  &
 0.094737183593
    &          0.094736601 & -5.8E-07&
0.094737 &  -1.8E-07\\
 1.05 & 
 0.073563056962
  &              0.073562473 & -5.8E-07 &
0.073563 &  -5.7E-07\\
1.10 & 
0.056234603424
 &              0.056234026 &  -5.8E-07 &
 0.056234 & -6.0E-07\\
 1.15 & 
 0.042343450278
 &             0.042342884 &  -5.6E-07 &
0.042343 &  -4.5E-07 \\
 1.20  & 0.031424605687 &             0.031424055 & -5.5E-07 &
0.031424 & -6.1E-07\\
\\\hline
\end{tabular}
}
\begin{flushleft}{\tiny
BB: parameters of the scheme but $N$ and approximately the same CPU time and errors are as in the case $T=2$;
$N=72$.\\
$V_{fast}$: 
sinh-acceleration, parameters $\om_1=-0.5$, $b=0.7699$, $\om=0$,  $\ze=0.285754315$, $N=49$, and Modification II of the Adams method with $M=399$. CPU time 75.9 msec, average over 1000 runs.
\\
\vskip-0.18cm
$V_H$: hybrid method of \cite{RoughNotTough} (prices are rescaled). CPU time is in the range $548-582$ per strike.}
\end{flushleft}

\label{table:T=1}

 \end{table}

 \begin{table}
\caption{\small Prices of call options in the rough Heston model with the parameters \eq{parEuRos}, and the parameters and errors
of different numerical schemes.  Spot $S_0=1$, $T=0.5$.
}
{\small
\begin{tabular}{c|c|cr|cr}
\hline\hline
$K$ & BB &    $V_{fast}$ & Err & $V_H$ & Err \\
0.8    &      
0.206111802644
    & 
       0.206111447
    & -3.6E-07 & 
0.206112 &  2.0E-07 \\	   
 0.85    &  
 0.162807253574
    & 
      0.162806725
    & -5.3E-07 & 
 0.162807 &  -2.5E-07\\
 0.9     &   
 0.123947936854
          &            0.123947253  
          & -6.8E-07 &
0.123948 & 6.3E-08 \\
 0.95    &  
 0.090636214294
        &
          0.090635436          
        & -7.8E-07 & 
0.090636 &  -2.1E-07 \\
 1.00  &    0.063497549406
         &             0.063496759       
         & -7.9E-07&
0.063497 &  -5.5E-07\\
 1.05 &     
 0.042550077891
      &          0.042549356   
      & -7.2E-07 &
0.042550 &  -7.8E-07\\
1.10 &     
0.027251523720
       & 
          0.027250926     
       &  -6.0E-07 &
 0.027251 & -5.2E-07\\
 1.15 &     
 0.016679874270
        & 
           0.016679422  
        &  -4.5E-07 &
0.016680 &  1.3E-07 \\
 1.20  &    
 0.009761422351
       & 
          0.009761109
       & -4.2E-08 &
0.009761 & -4.2E-07\\
\\\hline
\end{tabular}
}
\begin{flushleft}{\tiny
BB: parameters of the  scheme but $N$ and, approximately, CPU time and errors are as in the case $T=2$;
$N=76$.
\\
$V_{fast}$: sinh-acceleration, parameters $\om_1=-0.5$, $b=0.7699$, $\om=0$,  $\ze=0.1338$, $N=49$, and Modification II of the Adams method with $M=120$. CPU time 25.9 msec, average over 1000 runs.
\\
\vskip-0.18cm
$V_H$: hybrid method of \cite{RoughNotTough} (prices are rescaled). CPU time is in the range $666-689$ msec. per strike.} \end{flushleft}

\label{table:T=0.5}

 \end{table}

 \begin{table}
\caption{\small Prices of call options in the rough Heston model with the parameters \eq{parEuRos}, and the parameters and errors
of different numerical schemes.  Spot $S_0=1$, $T=1/12$.
}
{\small
\begin{tabular}{c|c|rr|rr}
\hline\hline
$K$ & BB  &   $V_{fast}$ & Err & $V_H$ & Err \\
0.8    &    0.200005303706    
      & 0.2000053034 & -2.0E-10 & 
0.200005 &  -3.0E-07 \\	   
 0.85    &   0.150108198525
         & 0.1501081720 & -2.6E-08 & 
 0.150108 &  -2.0E-07\\
 0.9     & 
   0.101143602025
 
  &         0.1011436020 & -3.5E-08&
 0.101144 & 4.0E-07 \\
 0.95    &  
 0.056723587708
     &  0.0567235335 & -5.4E-08 & 
0.056723 &  -5.9E-07 \\
 1.00  &    
 0.023896784255
     &         0.0238966874& -9.6E-08&
0.023896 &  -7.8E-07\\
 1.05 &   
 0.006809088971
  &          0.0068090654& -2.3E-08 &
0.006809&  -8.9E-07\\
1.10 &      
0.001205286641
 &          0.0012052942&  7.6E-09 &
 0.001205 & -2.9E-07\\
 1.15 &       0.000124980334 &        0.0001249659 &  -2.4E-08 &
0.000124 &  -9.8E-07 \\
 1.20  &   7.3234038E-06 & 7.345E-06 & 2.1E-08 &
7.32E-06 & -3.4E-09\\
\\\hline
\end{tabular}
}
\begin{flushleft}{\tiny
BB: parameters of the scheme but $N$ and, approximately,  CPU time and errors  are as in the case $T=2$;
$N=77$. \\
$V_{fast}$: 
sinh-acceleration, parameters $\om_1=-0.5$, $b=0.7699$, $\om=0$,  $\ze=0.1840$, $N=28$, and Modification II of the Adams method with $M=156$. CPU time 76.9 msec, average over 1000 runs.\\
\vskip-0.18cm
$V_H$: hybrid method of \cite{RoughNotTough} (prices are rescaled). CPU time is in the range $410-423$ msec. per strike. }
\end{flushleft}

\label{table:T=1/12}

 \end{table}

 \begin{table}
\caption{\small Prices of call options in the rough Heston model with the parameters \eq{parEuRos} and the parameters and errors
of different numerical schemes.  Spot $S_0=1$, $T=1/52$.
}
{\small
\begin{tabular}{c|r|rr|rr}
\hline\hline
$K$ & BB & $V_{fast}$ & Err &    $V_H$ & Err \\
0.8    &      0.200000000000031              & 
0.2000000000001
 & 8.3E-14 & 0.2 &  -3.1E-14 \\	   
 0.85    &   0.150000000807382     
   &0.1500000008074 & -6.7E-13  & 0.15 &  -8.1E-10\\
 0.9     &    0.100001974464705   
   &.           0.100001974456335 & -8.4E-12 &
0.10002 & 2.6E-08 \\
 0.95    &     
         0.050452602188597
      & 0.050452601708825    
      & -4.8E-10  & 
0.050491 &  3.8E-05 \\
 1.00  &
        0.011166584429206
    & 0.011166583626682& 8.0E-10&
0.011347 &  1.8E-04\\
 1.05 &      
     0.000375237694910& 
     0.000375237900273
   & 2.0E-10 &
4.113E-04 &  3.6E-05\\
1.10 & 6.99705E-07  & 6.698E-07

&  3.7E-12 &
 9.22E-07& 2.2E-07\\
 1.15 & 3.748E-11  & 3.748E-11 &  8.6E-15 &
6.82E-11 &  3.1E-11 \\
 1.20  &6.19E-17 & 3.76E-17 & -2.4E-17 &
1.8E-15 & 1.8E-15\\
\\\hline
\end{tabular}
}
\begin{flushleft}{\tiny 	
BB: two sets of sinh-deformations for OTM and ATM puts and calls; Modification III of the Adans method with $M=20000$.
CPU times are 135.2 and 132.9 sec., 1 run. Sets of parameters of sinh-deformations:\\
 (1) $\om_1=0.325762041$, $b=	1.014615984$, $\om=0.2$, $\ze=0.069857441$, $N=93$;\\
 (2) $\om_1=-1.325762041$, $b=	1.014615984$, $\om=-0.2$, $\ze=0.069857441$, $N=93$;\\
The differences with prices obtained using similar deformations with $\om=\pm 0.1$ are less than E-11.\\
$V_{fast}$:
two sets of sinh-deformations and Modification III of the Adans method with $M=1000$.
CPU times are 588.6 and 582.2 msec., the averages over 1000 runs. Sets of parameters of sinh-deformations\\ 
 (1) $\om_1=0.429259757$, $b=	0.868680815$, $\om=0.1$, $\ze=0.1002$, $N=69$;\\
 (2) $\om_1=-1.429259757$, $b=	0.868680815$, $\om=-0.1$, $\ze=0.1002$, $N=69$;\\
 \vskip-0.2cm
 $V_H$: hybrid method of \cite{RoughNotTough} (prices are rescaled). The CPU time is in the range $125-164$ msec.
 per strike.}
\end{flushleft}

\label{table:T=1/52}

 \end{table}

 \begin{table}
\caption{\small Prices of calls options in the rough Heston model with the parameters \eq{parEuRos}, and 
the parameters and errors of  different numerical schemes.  Spot $S_0=1$, $T=1/252$.
}
{\small
\begin{tabular}{c|rr|r|r|r|r}
\hline\hline
$K$ & $BB$ &    $Err$ &  & $V_H$ & $Err_{V_H}$ \\
 0.95    &  
 0.05000024558
        &2.3E-11&  & 
0.050003&  1.4E-09 \\
 1.00  &    0.0050111443   & 1.4E-09 & &
0.005012 &  8.6E-07\\
 1.05 &     3.3118E-08 &  3.1E-10 & &
6.39E-07 &  6.1E-07
\\\hline
\end{tabular}
}
\begin{flushleft}{\tiny 	Prices of OTM puts for $K=0.8-0.9, 1.1-1.2$ are smaller than E-12. Prices of OTM calls are smaller than E-15.
\\
For $K=1.1,1.15,1.2$, \cite{RoughNotTough} lists prices of OTM calls of the order of E-07 (apparently, artifacts of CM method used),
and prices of OTM puts at $K=0.8, 0.85, 0.9$ are assigned value 0.\\
$BB$: two sets of sinh-deformations and Modification III of the Adans method with $M=10000$.
CPU times are 430 and 494 sec., 1 run. Sets of parameters: \\
(1) $\om_1=0.429259757$, $b=0.868680815$, $\om=0.1$, $\ze=0.075145263$, $N=110$;\\
 (2) $\om_1=-1.429259757$, $b=0.868680815$, $\om=-0.1$, $\ze=0.075145263$, $N=110$;\\
 $V_{fast}$: the results obtained with with $\om_1=0.325762041,-1.325762041$, $\om=0.2$, $b=1.014615984$, $\ze=0.07335$,	$N=89$, and Modification III of the Adams method with $M=1000$. CPU times are 672.2 and 722.0 msec, the average over 1000 runs.
\\

$V_H$: hybrid method of \cite{RoughNotTough} (prices are rescaled). CPU time is in the range $154-196$ msec. per strike.
\\ Different sets of parameters of the scheme give prices  with differences of the order 
5E-0.9 for $K=0.8$; 5E-10 for $K=0.85$ for OTM puts, and of the order of $E-10$ for OTM calls.
In \cite{RoughNotTough}, 
prices of OTM calls are of the order of $E-07$.

}
\end{flushleft}
\label{table:T=1/252}

 \end{table}

\end{document}